\documentclass{article}

\usepackage{ijcai23}

\usepackage{times}
\usepackage{soul}
\usepackage{url}
\usepackage[hidelinks]{hyperref}
\usepackage[utf8]{inputenc}
\usepackage[small]{caption}
\usepackage{graphicx}
\usepackage{amsmath}
\usepackage{amsthm}
\usepackage{booktabs}
\usepackage{algorithm}
\usepackage{algorithmic}
\usepackage[switch]{lineno}
\usepackage{cleveref}
\usepackage{enumitem}

\newcommand\arxiv{}
\ifx\arxiv\undefined
  \usepackage[appendix=strip]{apxproof}
\else
  \usepackage[appendix=append]{apxproof}
\fi

\usepackage{makecell}
\usepackage[dvipsnames]{xcolor}
\usepackage{enumitem}
\usepackage{bm}
\usepackage{amsfonts}
\usepackage{amssymb}
\usepackage{nicefrac}
\usepackage{subcaption}

\newtheorem{example}{Example}
\newtheorem{theorem}{Theorem}

\newtheoremrep{apxproposition}{Proposition}
\newtheoremrep{apxlemma}{Lemma}
\newtheoremrep{apxtheorem}{Theorem}
\newtheoremrep{apxcorollary}{Corollary}

\usepackage[backgroundcolor=magenta!10!white]{todonotes}

\colorlet{simonColor}{Yellow!30!white}
\colorlet{christelColor}{Purple!30!white}
\colorlet{clemensColor}{Cyan!50!white}
\colorlet{hansColor}{Orange!30!white}

\makeatletter
\def\namedlabel#1#2{\begingroup
    #2%
    \def\@currentlabel{#2}%
    \phantomsection\label{#1}\endgroup
}
\makeatother

\ifdefined \theorem \else
  \newtheorem{theorem}{Theorem}
\fi

\newtheorem{corollary}[theorem]{Corollary}

\theoremstyle{definition}
\newtheorem{definition}{Definition}

\theoremstyle{remark}
\ifdefined \example \else
  \newtheorem{example}{Example}
\fi
\newtheorem*{example*}{Example}
\theoremstyle{remark}

\newtheorem*{remark*}{Remark}

\DeclareUnicodeCharacter{2212}{-}

\numberwithin{equation}{section}

\newcommand{\idf}[1]{\mathbf{1}_{#1}}

\newcommand{\af}{{\bm a}}

\newcommand{\bfat}{{\bm b}}

\newcommand{\uf}{{\bm u}}
\newcommand{\vf}{{\bm v}}
\newcommand{\wf}{{\bm w}}

\newcommand{\Ac}{\mathcal{A}}

\newcommand{\Dc}{\mathcal{D}}

\newcommand{\Rc}{\mathcal{R}}

\newcommand{\Bool}{\mathbb{B}}

\newcommand{\R}{\mathbb{R}}
\newcommand{\E}{\mathbb{E}}
\newcommand{\N}{\mathbb{N}}

\newcommand{\boolfuncs}{\Bool(X)}
\newcommand{\cgs}{\mathbb{G}(X)}

\newcommand{\switchmain}[2]{\mathrm{flip}_{#1}(#2)}
\newcommand{\switch}[2]{#2{\oplus}\idf{#1}}

\newcommand{\switchf}[2]{#2[#1/\overline #1]}
\newcommand{\switchfset}[3]{#2[#3/\overline #3 : #3 \in #1]}

\newcommand{\tiff}{\quad\text{iff}\quad}

\newcommand{\subjto}{\;\mathrm{s.t.}\;}

\newcommand{\Shapley}{{\bf Sh}}
\newcommand{\Banzhaf}{{\bf Bz}}
\newcommand{\ZeroVF}{{\bf Z}}
\newcommand{\var}{\mathtt{var}}

\newcommand{\image}{\mathtt{image}}

\newcommand{\cofactor}[2]{#1_{#2}}
\newcommand{\ass}[2]{#1{/}#2}

\newcommand{\Adiff}{\mathrm{D}}
\newcommand{\marginal}{\partial}
\newcommand{\dep}{\mathtt{dep}}
\newcommand{\rhostep}{ {\rho_{\mathrm{step}}} }
\newcommand{\rhofrac}{ {\rho_{\mathrm{frac}}} }
\newcommand{\rhoexp}{ {\rho_{\mathrm{exp}}} }
\newcommand{\scs}{\mathrm{scs}}

\newcommand{\mscs}{\mathrm{mscs}}
\newcommand{\shortestDist}{\mathrm{d}}

\newcommand{\cbanzhaf}{c_{\mathrm{Bz}}}
\newcommand{\cshapley}{c_{\mathrm{Sh}}}

\newcommand{\makeprop}[1]{\textnormal{\textsc{#1}}}

\newcommand{\Blame}{{\bf {B}}}
\newcommand{\MBlame}{{\bf {MB}}}
\newcommand{\Inf}{{\bf {I}}}
\newcommand{\JW}{{\bf {JW}}}

\newcommand{\Imp}{\mathfrak{I}}
\newcommand{\ImpCG}{{\mathfrak{E}}}

\newcommand{\HKR}{\mathrm{H}}

\newcommand{\kappaquad}{{\kappa_{\mathrm{quad}}}}
\newcommand{\kappalog}{{\kappa_{\mathrm{log}}}}
\newcommand{\kappaabs}{{\kappa_{\mathrm{abs}}}}

\newcommand{\atmost}{\texttt{atmost}}

\title{A Unifying Formal Approach to Importance Values in Boolean Functions}
\author{
Hans Harder$^{1,2}$
\and
Simon Jantsch$^2$\and
Christel Baier$^{2,4}$\And
Clemens Dubslaff$^{3,4}$
\affiliations
$^1$Paderborn University, Paderborn, Germany\\
$^2$Dresden University of Technology, Dresden, Germany \\
$^3$Eindhoven University of Technology, Eindhoven, The Netherlands \\
$^4$Centre for Tactile Internet with Human-in-the-Loop (CeTI), Dresden, Germany\\
\emails
hans.harder@uni-paderborn.de,
\{simon.jantsch, christel.baier\}@tu-dresden.de,
c.dubslaff@tue.nl
}

\begin{document}
\maketitle
\begin{abstract}
    
Boolean functions and their representation through logics, circuits, machine learning classifiers,
or binary decision diagrams (BDDs) play a central role in the design and analysis of computing systems. 
Quantifying the relative impact of variables on the truth value by means of \emph{importance values}
can provide useful insights to steer system design and debugging.
In this paper, we introduce a uniform framework for reasoning about such values,
relying on a generic notion of \emph{importance value functions (IVFs)}.
The class of IVFs is defined by axioms motivated from several notions of importance values
introduced in the literature, including Ben-Or and Linial's \emph{influence} and 
Chockler, Halpern, and Kupferman's notion of \emph{responsibility} and \emph{blame}.
We establish a connection between IVFs and game-theoretic concepts such as \emph{Shapley} and
\emph{Banzhaf} values, both of which measure the impact of players on outcomes in cooperative games.
Exploiting BDD-based symbolic methods and projected model counting, we devise and evaluate
practical computation schemes for IVFs.

\end{abstract}

\section{Introduction}

Boolean functions arise in many areas of computer science and mathematics, e.g., in circuit design,
formal logics, coding theory, artificial intelligence, machine learning, and system analysis
~\cite{crama_hammer_2011,o2014analysis}.
When modeling and analyzing systems through Boolean functions, many design decisions are
affected by the relevance of variables for the outcome of the function.
Examples include noise-reduction components for important input variables to increase 
reliability of circuits, prioritizing important variables in decision-making of protocols, 
or the order of variables in BDDs~\cite{bryantSymbolicBooleanManipulation1992,bartlettComparisonTwoNew2001}.
Many ideas to quantify such notions of \emph{importance} of variables in Boolean functions 
have since been considered in the literature. 
To mention a few, \emph{influence}~\cite{ben-orCollectiveCoinFlipping1985} 
is used to determine power of actors in voting schemes, \cite{hammerEvaluationStrengthRelevance2000} 
devised measures based on how constant a function becomes depending on variable assignments,
\emph{blame}~\cite{chocklerResponsibilityBlameStructuralModel2004} quantifies the average \emph{responsibility}~\cite{chocklerWhatCausesSystem2008}
of input variables on the outcome of circuits or on causal
reasoning, and
the \emph{Jeroslow-Wang value}~\cite{jeroslowSolvingPropositionalSatisfiability1990} quantifies 
importance of variables in CNFs to derive splitting rules for SAT-solvers~\cite{hookerBranchingRulesSatisfiability1995}.
Closely related are notions of impact in cooperative games, e.g.,
through the \emph{Shapley value}~\cite{shapleyValueNPersonGames1953}
or the \emph{Banzhaf value}~\cite{banzhafWeightedVotingDoesn1965}.

Although some of the aforementioned concepts %
are of quite different nature 
and serve different purposes, they share some common ideas.
This raises the question of what characteristics importance values 
have and how the notions of the literature relate.
The motivation of this paper is to advance the understanding of importance values, independent 
of concrete applications. For this purpose, we introduce a generic 
axiomatic framework that constitutes the class
of \emph{importance value functions (IVFs)}. 
Our axioms are motivated by properties one would intuitively expect from IVFs, e.g., 
that independent variables have no importance or that permutations do not change 
importance values. %
We show basic relationships within and between IVFs and provide new insights
for existing and new importance measures. 
By connecting Boolean functions and cooperative games through \emph{cooperative game mappings (CGMs)} and using Shapley and Banzhaf values, 
we show how to generically derive new IVFs .
All aforementioned notions of importance values from the literature %
satisfy our IVF axioms, showing that we provide a \emph{unifying framework} 
for all these notions, including CGM-derived ones.

Most notions of importance are known to be computationally hard, e.g., computing 
influence or the Shapley value is \#P-complete~\cite{traxler2009variable,FaiKer92,DenPap94}.
We address computational aspects by devising practical computation schemes for IVFs using 
projected model counting~\cite{aziz2015exists} and BDDs. %

\paragraph{Contributions and outline.}
In summary, our main contribution is an axiomatic definition of IVFs for variables in Boolean functions (\Cref{sec:ivfs_general}),
covering notions of importance from the literature (\Cref{ssec:instances_blame,ssec:instances_influence}).
Moreover, we derive novel IVFs by linking Boolean functions with cooperative games and related values (\Cref{ssec:cgms}).
Finally, we provide practical computation schemes for IVFs (\Cref{sec:computation}).

\paragraph{Supplemental material.} 
\ifx\arxiv\undefined
    All proofs %
    can be found in an extended version at (...).
\else
    This is a preprint of the paper accepted at the 32nd International Joint Conferences on Artificial Intelligence (IJCAI'23). It includes proofs and other additional material in the appendix.
\fi
An implementation of the computing schemes for IVFs can be found at \href{https://github.com/graps1/impmeas}{https://github.com/graps1/impmeas}.

\section{Preliminaries}
\begin{toappendix}
    \subsection*{Notation} 
\end{toappendix}
Let $X = \{x,y,z,\dots\}$ be a finite set of $n = |X|$ variables, which we assume to be fixed throughout
the paper.

\paragraph{Assignments.} An \emph{assignment over $U \subseteq X$} is a function 
$\uf  \colon U \rightarrow \{0,1\}$, written in the form 
$\uf = \ass x 0; \ass y 1; \dots$.
We denote assignments by bold lower-case letters and their domains by 
corresponding upper-case letters.
If $\uf$ and $\vf$ have disjoint domains, we write their \emph{concatenation} as $\wf = \uf; \vf$ with $W = V \cup U$ and $
\wf(x) = \uf(x)$ if $x \in U$ and $\wf(x) = \vf(x)$ if $x \in V$.
The \emph{restriction} of $\uf$ to a domain $S \subseteq U$ is denoted by $\uf_S$.
For a permutation $\sigma$ of $X$, we define $\sigma \uf$ as the assignment over $\sigma(U)$ with $(\sigma \uf)(x) = \uf(\sigma^{-1}(x))$.

\paragraph{Boolean functions.} We call $f,g,h,\dots : \{0,1\}^X \rightarrow \{0,1\}$ \emph{Boolean functions}, collected in a set $\boolfuncs$.
We write $g = x$ if $g$ is the indicator function of $x$, and we write $\overline g$ for negation, $f \lor g$ for disjunction, $fg$ for conjunction and
$f \oplus g$ for exclusive disjunction.
The \emph{cofactor of $f$ w.r.t. an assignment $\vf$} 
is the function $\cofactor f \vf$ that always sets variables in $V$ to the value given by $\vf$, and is defined as $\cofactor f \vf(\uf) = f(\vf; \uf_{U \setminus V})$.
The \emph{Shannon decomposition of $f$ w.r.t. variable $x$} is a decomposition rule
stating that $f = x \cofactor f {\ass x 1} \lor \overline x \cofactor f {\ass x 0}$ holds, 
where $\cofactor f {\ass x 1}$ and $\cofactor f {\ass x 0}$ are the \emph{positive} and \emph{negative} cofactor of $f$ w.r.t. $x$. 
For a Boolean function $f$, variable $x$, and Boolean function or variable $s$, let $f[x/s] = s \cofactor f {\ass x 1} \lor \overline s \cofactor f {\ass x 0}$ be the function that replaces $x$ by $s$.
For example, if 
$
    f = y \lor xz,
$
then $\cofactor f {\ass x 1} = y \lor z$ and $\cofactor f {\ass x 0} = y$. Moreover, for $s = x_1 x_2$, we have 
$$
    f[x/s] = s (y \lor z) \lor \overline s y = y \lor s z = y \lor x_1 x_2 z.
$$
For $\mathbin{\sim} \in \{ \leq, \geq, = \}$, we write $f \sim g$ if $f(\uf) \sim g(\uf)$ is true for all assignments.
We collect the variables that $f$ depends on in
the set $\dep(f) = \{ x \in X : \cofactor f {\ass x 1} \neq \cofactor f {\ass x 0} \}$.
If $\vf$ is an assignment with $\dep(f) \subseteq V$, then $f(\vf)$ denotes the only possible value that $\cofactor f \vf$ can take.

We say that $f$ is \emph{monotone in} $x$ if $\cofactor f {\ass x 1} \geq \cofactor f {\ass x 0}$, and call $f$ \emph{monotone} if $f$ is monotone in all of its variables.
Furthermore, $f$ is the \emph{dual} of $g$ if $f(\uf) = \overline g(\overline \uf)$, where $\overline \uf$ is the variable-wise negation of $\uf$.
We call $f$ \emph{symmetric} if $f = \sigma f$ for all permutations $\sigma$ of $X$, where $\sigma f(\uf) = f(\sigma^{-1} \uf)$.

\paragraph{Expectations.} We denote the expectation of $f$ w.r.t. the uniform distribution over $D$ by $\E_{d \in D}[f(d)]$ for $f: D \rightarrow \R$. We only consider 
cases where $D$ is finite, so 
\[
    \E_{d\in D}[f(d)]= \frac 1 {|D|} \sum_{d\in D} f(d).
\]
If the domain of $f$ is clear, we simply write $\E[f]$.
For $f\in\boolfuncs$, $\E[f]$ is the fraction of satisfying assignments of $f$. %

\begin{toappendix}

Since most of the ``heavy lifting'' is done here in the appendix, we need to introduce some more notation in order to keep everything as simple as possible.

\begin{itemize}
    \item We consider pseudo Boolean functions, which take $\{0,1\}^X$ to $\R$ and can be composed using $+,-$, etc. Cofactors, the set of dependent variables and comparison operations are defined as for Boolean functions.
    \item If two assignments $\uf,\vf$ are defined over the same domain, we can compose them variable-wise using operations such as $\land$, $\lor$, and $\oplus$. I.e., $(\uf \oplus \vf)(x) = \uf(x) \oplus \vf(x)$, etc.
    \item For $S\subseteq X$, denote the corresponding indicator function by $\idf{S}$, with $\idf{S}(x) = 1$ iff $x \in S$. The domain of $\idf{S}$ depends on the context, but contains at least $S$. Write $\idf{x} = \idf{ \{ x \} }$.
    \item Instead of substituting only one variable, we can use substitution chains of the form $f[x_i/s_i : i \in N]$ or similar if the order of replacement does not matter. 
\end{itemize}
Note that:
\begin{itemize}
    \item For any permutation $\sigma$ of $X$, the variable-wise interpretations of $\lor,\land,$ etc. imply that $\sigma$ distributes over these operations when applied to functions or assignments, e.g. $\sigma(\uf \land \vf) = \sigma \uf \land \sigma \vf$ or $\sigma(\overline \uf) = \overline {\sigma \uf }$ or $\sigma(\uf; \vf) = \sigma \uf; \sigma \vf$ or $\sigma( \overline g) = \overline{ \sigma g}$. The same holds for sets $A,B \subseteq X$ with e.g. $\sigma(A \setminus B) = \sigma(A) \setminus \sigma(B)$. 
    \item The identity $\switch{S}{\uf} = \switchmain{S}{\uf}$ applies, which makes it obvious that the flipping-operation can be written as a composition of assignments. For more basic properties, we refer to \Cref{thr:perm-props}.
\end{itemize}

\begin{apxlemma}%
    \label{thr:perm-props}
    Let $\sigma$ be a permutation of $X$, $h$ be a Boolean function and $T \subseteq S \subseteq X$ be arbitrary subsets. Then
    \begin{align}
        &\{0,1\}^S \rightarrow \{0,1\}^{\sigma(S)}, \uf \mapsto \sigma (\switch{T}{\uf})
            \text{ is a bijection.} 
            \label{eq:perm-prop-2} \\
        &(\sigma \uf)_{\sigma(T)} = \sigma (\uf_T)
            && \forall \uf \in \{0,1\}^S 
            \label{eq:perm-prop-6} \\
        &\switch{\sigma(T)}{\sigma \uf} = \sigma(\switch{T}{\uf})
            && \forall \uf \in \{0,1\}^S 
            \label{eq:perm-prop-10} \\ 
        &\cofactor {(\sigma h)} {\sigma \uf} = \sigma (\cofactor h \uf)
            && \forall \uf \in \{0,1\}^S 
            \label{eq:perm-prop-4} %
    \end{align}
\end{apxlemma}
\begin{proof}
    \hfill 
    \begin{enumerate}[leftmargin=1cm]
        
        \item[\eqref{eq:perm-prop-2}] This mapping can be inverted by taking $\{0,1\}^{\sigma(S)}$ back to $\{0,1\}^S$ with $\vf \mapsto \switch{T}{(\sigma^{-1}\vf)}$. Thus, it is a bijection.
        
        \item[\eqref{eq:perm-prop-6}] If $x \in \sigma(T)$ and $y = \sigma^{-1}(x)$, then 
        \begin{align*}
            (\sigma \uf)_{\sigma(T)}(x) = \sigma \uf(x) = \uf(y).
        \end{align*}
        Since $y \in T$ by construction,
        \begin{align*}
            \uf(y) = \uf_T(y) = \uf_T(\sigma^{-1}(x)) = \sigma(\uf_T)(x).
        \end{align*}
        Furthermore, if $x \not\in \sigma(T)$, then $(\sigma \uf)_{\sigma(T)}$ is not defined on $x$ and neither is $\sigma(\uf_T)$ due to $y \not\in T$. Therefore, both assignments must be equal.

        \item[\eqref{eq:perm-prop-10}] Follows from \eqref{eq:perm-prop-6},
        \begin{align*}
            \switch{\sigma(T)}{\sigma \uf} 
            &= (\sigma \uf)_{X \setminus \sigma(T)}; \overline{\sigma\uf}_{\sigma(T)} \\
            &= \sigma(\uf_{X \setminus T}); \sigma(\overline\uf_T) \\
            &= \sigma(\uf_{X \setminus T}; \overline\uf_T) \\
            &= \sigma(\switch{T}{\uf}).
        \end{align*}

        \item[\eqref{eq:perm-prop-4}] For all $\vf \in \{0,1\}^X$,
        \begin{align*}
            \cofactor {(\sigma h)} {\sigma \uf}(\vf) 
            &= \sigma h (\sigma \uf; \vf_{X \setminus \sigma(S)}) 
                & \text{definition} \\
            &= \sigma h (\sigma (\uf; \sigma^{-1}(\vf_{X \setminus \sigma(S)}))) 
                & \vf_{X \setminus \sigma(S)} = \sigma (\sigma^{-1} (\vf_{X \setminus \sigma(S)})) \\
            &= h(\uf; \sigma^{-1}(\vf_{X \setminus \sigma(S)})) 
                & \text{resolving $\sigma$} \\
            &= h(\uf; (\sigma^{-1}\vf)_{X \setminus S}) 
                & \eqref{eq:perm-prop-6} \\
            &= \cofactor h \uf(\sigma^{-1}\vf) 
                & \text{definition} \\
            &= \sigma(\cofactor h \uf)(\vf).
        \end{align*}

    \end{enumerate}
\end{proof}

There are a few properties of the derivative operator which are useful to point out. First, taking the derivative with respect to a function that does not depend on $x$ results in the constant zero function \eqref{eq:diff-op-prop-2}. Second, it is multiplicative in the sense that functions which do not depend on $x$ can be moved in front of the operator \eqref{eq:diff-op-prop-3}. And third, it is additive in the sense of \eqref{eq:diff-op-prop-5}:

\begin{apxlemma}%
\label{thr:propertiesDerivativeOperator}
    Let $f, h, g$ be Boolean functions such that $x \not\in \dep(g)$. Then
    \begin{align}
        &\Adiff_x g = 0 
            \label{eq:diff-op-prop-2} \\ 
        &\Adiff_x (g f) = g \Adiff_x f
            \label{eq:diff-op-prop-3} \\ 
        &\Adiff_x (f \oplus h) = \Adiff_x f \oplus \Adiff_x h 
            \label{eq:diff-op-prop-5}.
    \end{align}
\end{apxlemma}
\begin{proof}
    \eqref{eq:diff-op-prop-2} follows from the fact that $\cofactor g {\ass x 1} = \cofactor g {\ass x 0}$. For \eqref{eq:diff-op-prop-3}, we get $\cofactor {(g f)} {\ass x 1} \oplus \cofactor{(g f)} {\ass x 0} = g (\cofactor f {\ass x 1} \oplus \cofactor f {\ass x 0})$ since $\cofactor g {\ass x 1} = \cofactor g {\ass x 0} = g$. Finally, for \eqref{eq:diff-op-prop-5}, note that 
    $ %
        \cofactor {(f \oplus h)} {\ass x 1} \oplus \cofactor {(f \oplus h)} {\ass x 0} = (\cofactor f {\ass x 1} \oplus \cofactor f {\ass x 0}) \oplus (\cofactor h {\ass x 1} \oplus \cofactor h {\ass x 0}).
    $ %
\end{proof}

Expectations are frequently used throughout this paper, so let us establish a few simple equalities. First, one can compute the expected value of a pseudo Boolean function with respect to every superset of its dependent variables. Second, if two pseudo Boolean functions depend on mutually exclusive variables (dependence defined as for Boolean functions), then the expectation of their product can be decomposed. And third, expected values of modular Boolean functions admit decompositions as well:

\begin{apxlemma}%
    \label{thr:decompositionOfExpectations}
    \hfill
    \begin{enumerate}
        \item[\namedlabel{prop:expectations-1}{\textnormal{(i)}}] If $f$ is a pseudo Boolean function and $\dep(f) \subseteq R$, then $\E_{\uf \in \{0,1\}^X}[f(\uf)] = \E_{\uf \in \{0,1\}^R}[f(\uf)]$.
        \item[\namedlabel{prop:expectations-2}{\textnormal{(ii)}}] If $f,g$ are pseudo Boolean functions such that $\dep(f) \cap \dep(g) = \varnothing$, then $\E[ f \cdot g ] = \E[f ] \cdot \E[g]$.
        \item[\namedlabel{prop:expectations-3}{\textnormal{(iii)}}] If $f,g$ are Boolean functions such that $f$ is modular in $g$, then 
        $
            \E[f] = \E[g] \E[\cofactor f {\ass g 1}] + (1-\E[g]) \E[\cofactor f {\ass g 0}].
        $ \label{thr:decompositionOfExpectations-3}
    \end{enumerate}
\end{apxlemma}
\begin{proof}
    \hfill
    \begin{enumerate}[leftmargin=1cm]
        \item[\ref{prop:expectations-1}]
        \begin{align*}
            \E_{\uf \in \{0,1\}^X}[f(\uf)] 
            &= \E_{\vf \in \{0,1\}^{X \setminus R}}[ \E_{\uf \in \{0,1\}^R}[ f(\vf;\uf) ] ]
                && \text{uniform distribution} \\
            &= \E_{\vf \in \{0,1\}^{X \setminus R}}[ \E_{\uf \in \{0,1\}^R}[ f(\uf) ] ]
                && \dep(f) \cap (X \setminus R) = \varnothing \\
            &= \E_{\uf \in \{0,1\}^R}[f(\uf)].
                && \E[\text{const.}] = \text{const.}
        \end{align*}
        \item[\ref{prop:expectations-2}]
        \begin{align*}
            \E[f \cdot g] 
            &= \E_{\uf \in \{0,1\}^{\dep(f)}} [\E_{\vf \in \{0,1\}^{X \setminus \dep(f)}}[ f(\uf;\vf) \cdot g(\uf;\vf) ]] 
                && \text{uniform distribution} \\
            &= \E_{\uf \in \{0,1\}^{\dep(f)}} [\E_{\vf \in \{0,1\}^{X \setminus \dep(f)}}[ f(\uf) \cdot g(\vf) ]] 
                && \dep(f) \cap \dep(g) = \varnothing \text{ and } \dep(g) \subseteq X \setminus \dep(f)\\
            &= \E_{\uf \in \{0,1\}^{\dep(f)}}[ f(\uf)] \cdot \E_{\vf \in \{0,1\}^{X \setminus \dep(f)}}[ g(\vf) ] 
                \\
            &= \E_{\uf \in \{0,1\}^{X}}[ f(\uf)] \cdot \E_{\vf \in \{0,1\}^{X}}[ g(\vf) ] 
                && \text{using \ref{prop:expectations-1}} \\
            &= \E[ f] \cdot \E[ g ].
        \end{align*}
        \item[\ref{prop:expectations-3}]
        \begin{align*}
            \E[f] 
            &= \E[g \cofactor f {\ass g 1} + (1-g) \cofactor f {\ass g 0}] 
                && g \cofactor f {\ass g 1} \lor \overline g \cofactor f {\ass g 0} = g \cofactor f {\ass g 1} + (1-g) \cofactor f {\ass g 0} \\
            &= \E[g \cofactor f {\ass g 1}] + \E[(1-g) \cofactor f {\ass g 0}] 
                && \text{linearity of $\E$} \\
            &= \E[g] \E[\cofactor f {\ass g 1}] + (1-\E[g]) \E[\cofactor f {\ass g 0}].
                && \text{independence, using \ref{prop:expectations-2}}
        \end{align*}
    \end{enumerate}
\end{proof}
\end{toappendix}

\begin{toappendix}
    \subsubsection*{Modular Decompositions} 
\end{toappendix}
\paragraph{Modular decompositions.}
We introduce a notion of \emph{modularity} to capture independence of subfunctions
as common in the theory of Boolean functions and related 
fields~\cite{ashenhurstDecompositionSwitchingFunctions1957,birnbaumModulesCoherentBinary1965,shapleyCompoundSimpleGames1967,biochDecompositionBooleanFunctions2010}. %
Intuitively, $f$ is modular in $g$ if $f$ treats $g$ like a subfunction and otherwise ignores all variables that $g$ depends on. 
We define modularity in terms of a \emph{template function} $\ell$
in which $g$ is represented by a variable $x$:

\begin{definition} %
  Let $f,g \in \boolfuncs$. We call $f$ \emph{modular in $g$} if $g$ is not constant and 
  there is $\ell \in \boolfuncs$ and $x \in X$ such that $\dep(\ell) \cap \dep(g) = \varnothing$
  and $f = \ell[x/g]$. 
  If $\ell$ is monotone in $x$, then $f$ is \emph{monotonically modular in $g$}.
\end{definition}

\begin{toappendix}
    The following proposition shows that the cofactors of a modular Boolean function are indeed uniquely determined: 
    \begin{apxproposition} 
        \label{thr:modularUniquenessOfCofactors}
        Suppose $f$ is modular in $g$ and let $s, t \in \boolfuncs$ and $x,y \in X$ such that $f = s[y/g] = t[x/g]$ with $\dep(g) \cap (\dep(s) \cup \dep(t)) = \varnothing$. Then $\cofactor s {\ass y 1} = \cofactor t {\ass x 1}$ and $\cofactor s {\ass y 0} = \cofactor t {\ass x 0}$.
    \end{apxproposition}
    \begin{proof}
        Recall that
        \begin{align*}
            f 
              &= g \cofactor s {\ass y 1} \lor \overline g \cofactor s {\ass y 0} \\
              &= g \cofactor t {\ass x 1} \lor \overline g \cofactor t {\ass x 0}.
        \end{align*}
        
        Since $g$ is not constant,
        pick $\uf_1,\uf_0 \in \{0,1\}^{\dep(g)}$ such that $\cofactor g {\uf_1} = 1$ and $\cofactor g {\uf_0} = 0$. Then, due to the fact that the cofactors $\cofactor s {\ass y 1}, \cofactor s {\ass y 0}, \cofactor t {\ass x 1}$ and $\cofactor t {\ass x 0}$ do not depend on any variable of $g$,
        \begin{align*}
            &\cofactor s {\ass y 1} = \cofactor f {\uf_1} = \cofactor t {\ass x 1}, \\
            &\cofactor s {\ass y 0} = \cofactor f {\uf_0} = \cofactor t {\ass x 0}.
        \end{align*}
    \end{proof}
\end{toappendix}

\noindent If $f$ is modular in $g$ with $\ell$ and $x$ as above, then
    $f(\uf) = \ell(\wf)$,
where $\wf$ is defined for $y \in X$ as
\begin{center}
    $\wf(y) = \begin{cases} g(\uf) & \text{ if } y = x, \text{ and } \\ \uf(y) & \text{otherwise}. \end{cases}$
\end{center} 
Thus, the value computed by $g$ is assigned to $x$ and then used by $\ell$, which 
otherwise is not influenced by the variables that $g$ depends on.
For example, $f = x_1 \lor z_1 z_2 x_2$ is modular in $g = z_1 z_2$ since $f$ can be obtained by
replacing $x$ in $\ell = x_1 \lor x x_2$ by $g$. 
Note that $\dep(\ell) = \{x,x_1,x_2\}$ and $\dep(g) = \{z_1,z_2\}$ are disjoint.
This property is crucial, since it ensures $f$ and $g$ are coupled
through variable $x$ only. 

If $f$ is modular in $g$, then the cofactors $\cofactor \ell {\ass x 1}$ and $\cofactor \ell {\ass x 0}$ must be unique since $g$ is not constant.
\ifx\arxiv\undefined
\else
(See \Cref{thr:modularUniquenessOfCofactors} in the appendix.)
\fi
Hence, we can define the \emph{cofactors of $f$ w.r.t. $g$} as $\cofactor f {\ass g 1} = \cofactor \ell {\ass x 1}$ and $\cofactor f {\ass g 0} = \cofactor \ell {\ass x 0}$. The instantiation 
is reversed by setting $f[g/x] = x \cofactor f {\ass g 1} \lor \overline x \cofactor f {\ass g 0}$.

\begin{toappendix}
    \subsubsection*{Boolean Derivatives} 
\end{toappendix}
\paragraph{Boolean derivatives.}
We frequently rely on the \emph{derivative of a Boolean function $f$ w.r.t. variable $x$,} 
$$\Adiff_x f = \cofactor f {\ass x 1} \oplus \cofactor f {\ass x 0},$$ 
which encodes the undirected change of $f$ w.r.t. $x$. 
For example, $f = x \lor y$ has the derivative $\Adiff_x f = \overline y$, with the intuition that $x$ can only have an impact if $y$ is set to zero.
Furthermore, if $f$ is modular in $g$, we define the \emph{derivative of $f$ w.r.t. $g$} as $\Adiff_g f = \cofactor f {\ass g 1} \oplus \cofactor f {\ass g 0}$. 
Given this, we obtain the following lemma corresponding to the chain rule known in calculus:

\begin{toappendix}
    The following property mirrors the chain rule as known in calculus. It follows from some basic rules:
\end{toappendix}

\begin{apxlemmarep}
    \label{thr:derivativeDecomposition}
    Let $f$ be modular in $g$ and $x \in \dep(g)$. Then
    \begin{center} $\Adiff_x f = (\Adiff_x g) (\Adiff_g f).$ \end{center}
\end{apxlemmarep}
\begin{appendixproof}
    \begin{align*}
        \Adiff_x f 
        &= \Adiff_x( g \cofactor f {\ass g 1} \oplus (1 \oplus g) \cofactor f {\ass g 0}) 
            && g \cofactor f {\ass g 1} \lor \overline g \cofactor f {\ass g 0} = g \cofactor f {\ass g 1} \oplus (1 \oplus g) \cofactor f {\ass g 0} \\
        &= (\Adiff_x g ) \cofactor f {\ass g 1} \oplus (( \Adiff_x 1) \oplus (\Adiff_x g)) \cofactor f {\ass g 0} 
            && \eqref{eq:diff-op-prop-5}, \eqref{eq:diff-op-prop-3} \\
        &= (\Adiff_x g) (\cofactor f {\ass g 1} \oplus \cofactor f {\ass g 0}) 
            && \eqref{eq:diff-op-prop-2} \\
        &= (\Adiff_x g) (\Adiff_g f).
    \end{align*}
\end{appendixproof}

\begin{toappendix}
\end{toappendix}

\section{Importance Value Functions} \label{sec:ivfs_general}

In this section, we devise axiomatic properties %
that should be fulfilled by 
every \emph{reasonable} importance attribution scheme. %

For a Boolean function $f$ and a variable $x$, we quantify the importance 
of $x$ in $f$ by a number $\Imp_x(f) \in \R$, computed by some \emph{value function} $\Imp$.
Not every value makes intutive sense when interpreted as the ``importance''
of $x$, so we need to pose certain restrictions on $\Imp$.

We argue that $\Imp$ should be bounded, with $1$ marking the highest and $0$ the lowest importance; that functions which are independent of a variable should rate these variables the lowest importance (e.g., $\Imp_x(f) = 0$ if $f = y \lor z$); that functions which depend on one variable only should rate these variables the highest importance (e.g., $\Imp_x(f)=1$ for $f=x$); that neither variable names nor polarities should play a role in determining their importance (e.g.,  $\Imp_x(x \overline z) = \Imp_z(x \overline z)$, cf. ~\cite{slepianNumberSymmetryTypes1953,golombClassificationBooleanFunctions1959}):

\begin{definition}[IVF]\label{def:importanceValueFunction}
    A \emph{value function} is a mapping of the form $\Imp\colon X{\times}\boolfuncs \rightarrow \R$ 
    with $(x,f) \mapsto \Imp_x(f)$.
    An \emph{importance value function (IVF)} is a value function $\Imp$ where for all $x,y\in X$, permutations 
    $\sigma\colon X\rightarrow X$, and $f,g,h \in\Bool(X)$:
    \begin{enumerate}[leftmargin=1.8cm]
        \item[(\namedlabel{prop:B}{\makeprop{Bound}})] $0 \leq \Imp_x(f) \leq 1$.
        \item[(\namedlabel{prop:DUM}{\makeprop{Dum}})] $\Imp_x(f) = 0$ if $x\not\in\dep(f)$.
        \item[(\namedlabel{prop:DIC}{\makeprop{Dic}})] $\Imp_x(x) = \Imp_x(\overline x) = 1$.
        \item[(\namedlabel{prop:TI}{\makeprop{Type}})]
          (i) $\Imp_x(f) = \Imp_{\sigma(x)}(\sigma f)$ and  \\
          (ii) $\Imp_x(f) = \Imp_x(\switchf{y}{f})$.
        \item[(\namedlabel{prop:SUD}{\makeprop{ModEC}})] $\Imp_x(f) \geq \Imp_x(h)$ if \\
        (i) $f$ and $h$ are monotonically modular in $g$, \\ 
        (ii) $\cofactor f {\ass g 1} \geq \cofactor h {\ass g 1}$ and $\cofactor h {\ass g 0} \geq \cofactor f {\ass g 0}$, and \\
        (iii) $x \in \dep(g)$.
    \end{enumerate}
\end{definition}
\ref{prop:B}, \ref{prop:DUM} for ``dummy'', \ref{prop:DIC} for ``dictator'' and \ref{prop:TI} for ``type invariance'' were discussed above.
\ref{prop:SUD} (for ``modular encapsulation consistency'') is the only property that allows 
the inference of non-trivial importance inequalities in different functions.
Let us explain its intuition.
We say that \emph{$f$ encapsulates $h$ on $g$} if these functions satisfy (i) and (ii) from \ref{prop:SUD}. 
Intuitively, together with (i), condition (ii) states that \emph{if one can control the output of $g$, it is both easier to satisfy $f$ than $h$} (using $\cofactor f {\ass g 1} \geq \cofactor h {\ass g 1}$) \emph{and to falsify $f$ than $h$} (using $\cofactor h {\ass g 0} \geq \cofactor f {\ass g 0}$).
We argue in \ref{prop:SUD} that if $f$ encapsulates $h$ on $g$, then $g$'s impact on $f$ is higher
than on $h$, and thus, the importance of variables in $\dep(g)$ (cf. (iii)) should be also higher w.r.t. $f$ than w.r.t. $h$.

\begin{toappendix}
  The Winder preorder according to \cite{hammerEvaluationStrengthRelevance2000} (Section 3 therein) is defined for a Boolean function $f$ and variables $x,y$ as 
  \[ 
    x \geqslant_{f} y \tiff \cofactor f {\ass x 1;\ass y 0} \geq \cofactor f {\ass x 0; \ass y 1}.
  \]
  Similar to \cite{hammerEvaluationStrengthRelevance2000} (Theorem 3.1 therein), we focus only on the case were $f$ is monotone in $x$ and $y$. \ref{prop:SUD} can then be seen as a generalization of this order in two senses: first, we abstract this order to the level of modules. Second, we abstract from an inequality between variables in one function to an inequality between variables in two functions. The following proposition now links the Winder preorder with \ref{prop:SUD} and \ref{prop:TI}:
  \begin{apxproposition}
  \label{prop:winder}
    Suppose $\Imp$ is an importance value function and let $f$ be a Boolean function monotone in $x$ and $y$. Then: 
    \[ 
      x \geqslant_f y \implies \Imp_x(f) \geq \Imp_y(f).
    \]
  \end{apxproposition}
  \begin{proof}
    Suppose that $x \geqslant_f y$. Define $h = \sigma f$, where $\sigma$ is the permutation of $X$ that swaps $x$ and $y$. Since $\cofactor f {\ass x c; \ass y d}$ does not depend on $x$ nor $y$, note that $\cofactor f {\ass x c; \ass y d} = \sigma(\cofactor f {\ass x c; \ass y d})$, which in turn implies $\cofactor f {\ass x c; \ass y d} = \cofactor h {\ass x d; \ass y c}$ according to \Cref{thr:perm-props}. Thus, we can derive $\cofactor f {\ass x 1} \geq \cofactor h {\ass x 1}$ and $\cofactor h {\ass x 0} \geq \cofactor f {\ass x 0}$ since
    \begin{align*}
      \cofactor f {\ass x 1} \geq \cofactor h {\ass x 1} \Longleftrightarrow \begin{cases}
        \cofactor f {\ass x 1; \ass y 1} \geq \cofactor h {\ass x 1; \ass y 1} \\
        \cofactor f {\ass x 1; \ass y 0} \geq \cofactor h {\ass x 1; \ass y 0} 
      \end{cases}
      \quad\text{and}\quad
      \cofactor h {\ass x 0} \geq \cofactor f {\ass x 0} \Longleftrightarrow \begin{cases}
        \cofactor h {\ass x 0; \ass y 0} \geq \cofactor f {\ass x 0; \ass y 0} \\
        \cofactor h {\ass x 0; \ass y 1} \geq \cofactor f {\ass x 0; \ass y 1}.
      \end{cases}
    \end{align*} 
    Since both $f$ and $h$ are monotone in $x$, apply \ref{prop:SUD} to derive $\Imp_x(f) \geq \Imp_x(h)$. Using \ref{prop:TI}, conclude $\Imp_x(h) = \Imp_y(f)$.
  \end{proof}
\end{toappendix}

\begin{example*} 
    Let $f = x_1 x_2 \lor x_3 x_4 x_5$, $h = x_3 x_4 \lor x_1 x_2 x_5$, and $\Imp$ be an IVF.
    Then $f$ encapsulates $h$ on $g = x_1 x_2$, since
      \begin{center}$
    \underbrace{1}_{\cofactor f {\ass g 1}}
    \ \geq\ \underbrace{x_3 x_4 \lor x_5}_{\cofactor h {\ass g 1}}
    \ \geq\ \underbrace{x_3 x_4}_{\cofactor h {\ass g 0}}
    \ \geq\ \underbrace{x_3 x_4 x_5}_{\cofactor f {\ass g 0}}.$\end{center}
    We then get $\Imp_{x_1}(f) \geq \Imp_{x_1}(h)$ 
    by application of \ref{prop:SUD}. Swapping $x_1$ with $x_3$ and $x_2$ with $x_4$, 
    we obtain a permutation $\sigma$ such that $h = \sigma f$. By \ref{prop:TI}, 
    we derive $\Imp_{x_1}(h) = \Imp_{x_3}(f)$. Using \ref{prop:TI} on the other variables yields
    \begin{center} %
    $\Imp_{x_1}(f) = \Imp_{x_2}(f) \geq \Imp_{x_3}(f) = \Imp_{x_4}(f) = \Imp_{x_5}(f).$\end{center} %
\end{example*}

Together with \ref{prop:TI}, \ref{prop:SUD} implies the \emph{Winder preorder}, which is similar in spirit
(see~\cite{hammerEvaluationStrengthRelevance2000}). However, \ref{prop:SUD} generalizes 
to modular decompositions and allows inferring importance inequalities w.r.t. to different 
functions.
\ifx\arxiv\undefined
\else
(See~\Cref{prop:winder} in the appendix). 
\fi

\paragraph{Biased and unbiased.}
We say that an IVF is \emph{unbiased} if $\Imp_x(g) = \Imp_x(\overline g)$
holds for all Boolean functions $g$ and variables $x$.
That is, unbiased IVFs measure the impact of variables without any preference for one particular
function outcome, while biased ones quantify the impact to enforce a function to return one or
zero. Biased IVFs can, e.g., be useful when the task is to assign responsibility values for the 
violation of a specification.

\subsection{Further Properties}\label{sec:further_properties}
We defined IVFs following a conservative approach, collecting minimal requirements on IVFs.
Further additional properties can improve on the 
predictability and robustness of IVFs.
\begin{definition}
  A value function $\Imp$ is called 
   \label{def:optionalProperties}
    \begin{itemize}[noitemsep] %
        \item \emph{rank preserving}, if for all $f,g \in \boolfuncs$ such that $f$ is modular 
        	in $g$ and $x,y\in \dep(g)$:
        \begin{center}$\Imp_x(g) \geq \Imp_y(g) \ \implies \ \Imp_x(f) \geq \Imp_y(f), $\end{center}
        \item \emph{chain-rule decomposable}, if for all $f,g \in \boolfuncs$ such that $f$ is 
        	modular in $g$ and $x \in \dep(g)$:
          \begin{center}$\Imp_x(f) \ = \ \Imp_x(g) \Imp_g(f),$\end{center}
        where $\Imp_g(f) = \Imp_{x_g}(f[g/x_g])$ for some $x_g \not\in \dep(f)$,
        \item and \emph{derivative dependent}, if for all $f,g \in \boolfuncs$, $x \in X$:
        \begin{center}$\Adiff_x f \geq \Adiff_x g \ \implies \ \Imp_x(f) \geq \Imp_x(g).$\end{center}
    \end{itemize}
    We also consider \emph{weak} variants of \emph{rank preserving} and \emph{chain-rule decomposable} where $f$ ranges only over functions that are
    \emph{monotonically} modular in $g$.
\end{definition}
\paragraph{Rank preservation.}
Rank preservation states that the relation between two variables should not change if the function 
is embedded somewhere else. 
This can be desired, e.g., during a modeling process in which distinct Boolean functions are composed or 
fresh variables added, where rank preserving IVFs maintain the relative importance order of variables.
We see this as a useful but optional 
property of IVFs since an embedding could 
change some parameters of a function that might be relevant for the relationship of both variables. 
For example, if $f = g z$ with $z \not\in \dep(g)$, then %
the relative number of satisfying assignments is halved compared to $g$. If $x$ is more important 
than $y$ in $g$ but highly relies on $g$ taking value one, it might be that this relationship 
is reversed for $f$ (cf. example given in \Cref{ssec:instances_blame}).

\paragraph{Chain-rule decomposability.}
If an IVF is chain-rule decomposable, then the importance of a variable in a module is the product of (i) its importance w.r.t. the module and (ii) the importance of the module w.r.t. the function. 
Many values studied in this paper
satisfy this property (\Cref{sec:instances}). 

\begin{example*}
Let $f = x_1 \oplus \cdots \oplus x_m$, and let $\Imp$ be a chain-rule decomposable IVFs with $\Imp_x(x\oplus y)=\alpha$. Since $f$ is modular in $g = x_1 \oplus \cdots \oplus x_{m-1}$, and $g$ modular in $x_1 \oplus \cdots \oplus x_{m-2}$, etc., we can apply the chain-rule property iteratively to get 
\begin{center}
$ 
  \Imp_{x_1}(f) 
  = \Imp_{x_1}(g)\Imp_g(f)
  = \Imp_{x_1}(g) \alpha
  = \dots = \alpha^{m-1},
$
\end{center}
where we use \ref{prop:TI} to derive $\Imp_g(f) = \Imp_{x_g}(x_g \oplus x_m) = \alpha$.
\end{example*}

\paragraph{Derivative dependence.}
Derivative dependence %
states that an IVF should quantify the \emph{change} 
a variable induces on a Boolean function. It can be used to derive, e.g., the inequality 
$\Imp_{x_1}(x_1 \oplus x_2 x_3) \geq \Imp_{x_1}(x_2 \oplus x_1 x_3)$, which is not possible solely using 
\ref{prop:SUD} since $x_1 \oplus x_2 x_3$ is neither monotone in $x_1$ nor in $x_2$. If a value function $\Imp$ 
(that is not necessarily an IVF) is derivative dependent, then this has some interesting implications.
First, $\Imp$ is unbiased and satisfies \ref{prop:SUD}. 
Second, if $\Imp$ is weakly chain-rule decomposable (weakly rank preserving), 
then it is also chain-rule decomposable (rank preserving). Finally, if $\Imp$ 
satisfies \ref{prop:DIC} and \ref{prop:DUM}, then it is also bounded by zero and one. 
As a consequence, if $\Imp$ is derivative dependent and satisfies \ref{prop:DIC}, 
\ref{prop:DUM}, and \ref{prop:TI}, then $\Imp$ is an IVF.
\ifx\arxiv\undefined
\else
(See \Cref{thr:implicationsOfDD} in the appendix.)
\fi

\begin{toappendix}
  The following proposition formalizes and proves claims that were made in \Cref{sec:ivfs_general} with respect to various implications of derivative dependence:
  \begin{apxproposition}
    \label{thr:implicationsOfDD}
    Suppose $\Imp$ is a derivative-dependent value function. Then,
    \begin{enumerate}
        \item It is unbiased.
        \item It satisfies \ref{prop:SUD}.
        \item If $\Imp$ is weakly chain rule decomposable, it is chain rule decomposable.
        \item If $\Imp$ is weakly rank preserving, then it is rank preserving.
        \item If $\Imp$ satisfies \ref{prop:DIC} and \ref{prop:DUM}, then it satisfies \ref{prop:B}.
    \end{enumerate}
  \end{apxproposition}
  \begin{proof}
      \hfill
      \begin{enumerate}
          \item Follows from the fact that $\Adiff_x g = \Adiff_x \overline g$ for all Boolean functions $g$ and variables $x$: First note $\overline g = 1 \oplus g$, then apply \eqref{eq:diff-op-prop-5} and \eqref{eq:diff-op-prop-2}.
          
          \item Assume that $f$ encapsulates $h$ on $g$ and $x \in \dep(g)$. We want to show that $\Imp_x(f) \geq \Imp_x(h)$ must follow. First note that due to $\cofactor f {\ass g 1} - \cofactor f {\ass g 0} \geq \cofactor h {\ass g 1} - \cofactor h {\ass g 0} \geq 0$ and $\cofactor f {\ass g 1} - \cofactor f {\ass g 0} = \cofactor f {\ass g 1} \oplus \cofactor f {\ass g 0}$ and $\cofactor h {\ass g 1} - \cofactor h {\ass g 0} = \cofactor h {\ass g 1} \oplus \cofactor h {\ass g 0}$, it is indeed the case that $\Adiff_g f \geq \Adiff_g h$. Furthermore, by \Cref{thr:derivativeDecomposition}, mind that $\Adiff_x f = (\Adiff_x g) (\Adiff_g f)$ and $\Adiff_x h = (\Adiff_x g) (\Adiff_g h)$. Due to $\Adiff_x g \geq 0$, it must also be the case that $\Adiff_x f \geq \Adiff_x h$. Application of derivative dependence now shows $\Imp_x(f) \geq \Imp_x(h)$.
          
          \item Suppose that $\Imp$ is weakly chain rule decomposable, that $f$ is modular in $g$ and $x \in \dep(g)$. We want to show $\Imp_x(f) = \Imp_x(g) \Imp_g(f)$. Let us define $h = g  ( \cofactor f {\ass g 1} \lor \cofactor f {\ass g 0} ) \lor \overline g \cofactor f {\ass g 1} \cofactor f {\ass g 0}$. It is easy to see that $h$ is monotonically modular in $g$ and that $\Adiff_g h = \Adiff_g f$. Furthermore, by \Cref{thr:derivativeDecomposition}, note that $\Adiff_x f = (\Adiff_x g) (\Adiff_g f) = (\Adiff_x g) (\Adiff_g h) = \Adiff_x h$. Therefore, since $\Imp$ is derivative dependent and weakly chain rule decomposable,
          \[ \Imp_x(f) = \Imp_x(h) = \Imp_x (g)\Imp_g(h) =  \Imp_x(g) \Imp_g(f). \] 
          
          \item Suppose that $\Imp$ is weakly rank preserving, that $f$ is modular in $g$ and $x,y \in \dep(g)$ such that $\Imp_x(g) \geq \Imp_y(g)$. We want to show $\Imp_x(f) \geq \Imp_y(f)$. If we define $h$ as in 3., we can derive $\Imp_x(f)= \Imp_x(h)$ and $\Imp_y(f) = \Imp_y(h)$ due to $\Imp$'s derivative dependence. Furthermore, note that $\Imp_x(h) \geq \Imp_y(h)$ must hold by weak rank preservation. So $\Imp_x(f) \geq \Imp_y(f)$ holds as well.
          
          \item Suppose $f$ is an arbitrary Boolean function and $x$ a variable. We want to show $0 \leq \Imp_x(f) \leq 1$. Mind that $\Adiff_x f \leq 1 = \Adiff_x x$, so $\Imp_x(f) \leq \Imp_x(x) = 1$ due to \ref{prop:DIC} and $\Imp$'s derivative dependence. Furthermore, if $g$ is an arbitrary Boolean function with $x\not\in\dep(g)$, then $\Adiff_x g = 0 \leq \Adiff_x f$, which allows the application of \ref{prop:DUM} and derivative dependence to conclude $0 = \Imp_x(g) \leq \Imp_x(f)$.
      \end{enumerate}
  \end{proof}
\end{toappendix}

\subsection{Induced Relations}
In this section, we will establish foundational relations between IVFs.
Recall that $f$ is a \emph{threshold function} if
\begin{center}
  $f(\uf) = 1 \tiff \textstyle\sum_{x \in X} w_x \uf(x) \geq \delta \quad \forall \uf \in \{0,1\}^X,$
\end{center}
where $\{ w_x \}_{x \in X} \subseteq \R$ is a set of weights and $\delta {\in} \R$ a threshold.

\begin{apxtheoremrep}%
    \label{thr:propertyRelationsOfVFs}
    Let $\Imp$ be an IVF, $f,g,h{\in}\boolfuncs$, $x,y{\in}X$. Then:
    \begin{enumerate}[label=(\arabic*)]
        \item If $f$ is symmetric, then $\Imp_x(f) = \Imp_y(f)$.
        \item If $\Imp$ is unbiased and $f$ is dual to $g$, then $\Imp_x(f) = \Imp_x(g)$.
        \item If $f$ is a threshold function with weights $\{ w_x \}_{ x \in X} \subseteq \R$,
        then $|w_x| \geq |w_y|$ implies
        $ %
          \Imp_x(f) \geq \Imp_y(f).
        $ %
        \item If $f$ is monotonically modular in $g$ and $x \in\dep(g)$, then $\Imp_x(g) \geq \Imp_x(f)$.
        \item If $\Imp$ is derivative dependent and $x\not\in \dep(g)$, then $\Imp_x(h \oplus g) = \Imp_x(h)$.
        \item If $\Imp$ is (weakly) chain-rule decomposable, then it is (weakly) rank preserving.
    \end{enumerate}
\end{apxtheoremrep}
\begin{appendixproof}
    \hfill
    \begin{enumerate}
        \item Define $\sigma$ to be the permutation that swaps $x$ and $y$. Then $\sigma f = f$ by symmetry of $f$, so we get $\Imp_x(f) = \Imp_{\sigma(x)}(\sigma f) = \Imp_y(f)$ due to \ref{prop:TI}.
        
        \item Follows directly from the fact that $f = \switchfset{X}{\overline{g}}{y}$. So $\Imp_x(f) = \Imp_x(\overline g)$ from repeated application of \ref{prop:TI}. Finally, $\Imp_x(\overline g) = \Imp_x(g)$ due to $\Imp$'s unbiasedness.

        \item We first assume that all weights are non-negative and show the general case later. Define $h = \sigma f$, where $\sigma$ is the (self-inverse) permutation that swaps $x$ and $y$. Note that $h(\uf) = 1$ iff 
        \[ 
          \uf(x) w_y + \uf(y) w_x + \sum_{z \in X \setminus \{ x, y \} } \uf(z) w_z \geq \delta.
        \]
        Let us demonstrate that \ref{prop:SUD} is applicable to derive $\Imp_x(f) \geq \Imp_x(h)$: as $w_x \geq 0$ and $w_y \geq 0$, both $f$ and $h$ must be monotone in $x$. Further, due to $w_x \geq w_y$, we also have $\cofactor f {\ass x 1} \geq \cofactor h {\ass x 1}$ and $\cofactor h {\ass x 0} \geq \cofactor f {\ass x 0}$, i.e. $f$ encapsulates $h$ on $x$. So use \ref{prop:SUD} to get $\Imp_x(f) \geq \Imp_x(h)$ and then \ref{prop:TI} to get $\Imp_x(f) \geq \Imp_{\sigma(x)}(\sigma h) = \Imp_{y}(f)$.
        
        Finally, suppose now that $f$ is a threshold function with possibly negative weights. Let $Y = \{ z \in X : w_z < 0 \}$ and define $\ell = \switchfset{Y}{f}{z}$. Then $\ell(\uf) = 1$ iff 
        \begin{align*}
          & &&\sum_{z \in Y} (1-\uf(z)) w_z + \sum_{z \in X \setminus Y} \uf(z) w_z \geq \delta \\
          &\Longleftrightarrow &&\sum_{z \in X} \uf(z) | w_z | \geq \delta - \sum_{z \in Y} w_z.
        \end{align*}
        Thus, $\ell$ is a threshold function with non-negative weights and we can apply \ref{prop:TI} and the same argument as above to get $\Imp_x(f) = \Imp_x(\ell) \geq \Imp_y(\ell) = \Imp_y(f)$.

        \item Note that 
        $ 
          \cofactor g {\ass g 1} = 1 \geq \cofactor f {\ass g 1}
        $ 
        and 
        $ 
          \cofactor f {\ass g 0} \geq 0 = \cofactor g {\ass g 0}.
        $
        Since $f$ is monotonically modular in $g$ and $x \in \dep(g)$ by assumption, we can apply \ref{prop:SUD} and conclude $\Imp_x(g) \geq \Imp_x(f)$.

        \item Assume that $\Imp$ is derivative dependent. Since $\Adiff_x (h \oplus g) = \Adiff_x h \oplus \Adiff_x g = \Adiff_x h$ due to \eqref{eq:diff-op-prop-5} and \eqref{eq:diff-op-prop-2}, note that $\Imp_x(h \oplus g) = \Imp_x(h)$ must hold.

        \item Assume that $\Imp$ is (weakly) chain rule decomposable. We want to show that it is (weakly) rank preserving. Suppose $f$ is (monotonically) modular in $g$ and pick $x,y \in \dep(g)$ such that $\Imp_x(g) \geq \Imp_y(g)$. Since $\Imp$ is (weakly) chain rule decomposable, $\Imp_x(f) = \Imp_x(g) \Imp_g(f)$ and $\Imp_y(f) = \Imp_y(g) \Imp_g(f)$. Now apply the fact that $\Imp_g(f) \geq 0$ to conclude $\Imp_x(f) \geq \Imp_y(f)$. Therefore, $\Imp$ must be (weakly) rank preserving.
       
    \end{enumerate}
\end{appendixproof}
For the case of threshold functions, \Cref{thr:propertyRelationsOfVFs} shows in (3) that any IVF will rank variables according to their absolute weights.
In (4), it is stated that the if a function is monotonically embedded somewhere, 
the importance of variables in that function can only 
decrease, e.g., $\Imp_x(x y) \geq \Imp_x(x y z)$. 
Moreover, in (5), if derivative dependence is satisfied, $\oplus$-parts 
without the variable can be dropped. 
As a consequence, $\Imp_x(f)=1$ whenever $f$ is a parity 
function and $x \in \dep(f)$.

\section{Instances of Importance Value Functions}\label{sec:instances}
In this section, we show that IVFs can be instantiated with
several notions for importance values from the literature and thus provide 
a unifying framework.
\subsection{Blame}\label{ssec:instances_blame}
\begin{toappendix}
    \subsection*{Blame} 
\end{toappendix}
Chockler, Halpern, and Kupferman's (CHK) notions of \emph{responsibility}~\cite{chocklerWhatCausesSystem2008} 
and \emph{blame}~\cite{chocklerResponsibilityBlameStructuralModel2004}
measure the importance of %
$x$ in 
$f$ through the number of variables 
that have to be flipped in an assignment $\uf$ until $x$ becomes \emph{critical}, i.e., ``flipping'' $x$ changes the outcome of $f$ to its complement.
Towards a formalization, let 
\begin{center}
  $\switchmain{S}{\uf}(x) = \begin{cases} 
    \overline \uf(x) & \text{if } x \in S \\
    \uf(x) & \text{otherwise} \end{cases}$
\end{center}
denote the assignment that flips variables in $S$.
We now rely on the following notion of critical set:

\begin{definition}[Critical sets]
    \label{def:smallestCriticalSet}
    A \emph{critical set} of $x\in X$ in $f\in\boolfuncs$ under assignment 
    $\uf$ over $X$ is a set $S \subseteq X{\setminus}\{ x \}$ where
    \begin{center}
      $f(\uf) = f(\switchmain{S}{\uf}) \text{ and } f(\uf) \neq f(\switchmain{S\cup\{x\}}{\uf}).$
    \end{center}
    We define $\scs^\uf_x(f)$ as the size of the smallest critical set, and 
    set $\scs^\uf_x(f) = \infty$ if there is no such critical set.
\end{definition}

\begin{example*}
  The set $S = \{ y \}$ is critical for $x$ in $f = x \lor y$ under $\uf = \ass x 1; \ass y 1$.
  It is also the smallest critical set. On the other hand, there is no critical set if $\uf = \ass x 0; \ass y 1$. 
\end{example*}

The responsibility of $x$ for $f$ under $\uf$ is inversely related to $\scs^\uf_x(f)$. 
Using the following notion of a \emph{share function}, we generalize the original notion of 
responsibility~\cite{chocklerWhatCausesSystem2008}: %

\begin{toappendix}
  Let us introduce a modified version of Definition \ref{def:smallestCriticalSet}, which introduces the dependence on a constant $c \in \{0,1\}$. This is a bit more useful for the coming proofs and statements.

  \begin{definition}
      Let $f$ be a Boolean function, $x$ be a variable, $c \in \{0,1\}$ be a constant and $\uf$ be an assignment over $X$. Define
      \begin{align*}
        \scs^\uf_x(f,c) = \min_{S\subseteq X \setminus \{ x \}} |S| \quad\subjto\quad c = f(\switch{S}{\uf}) \neq f(\switch{x}{\switch{S}{\uf}}).
      \end{align*}
      If there is no solution, set $\scs^\uf_x(f,c) = \infty$.
  \end{definition}

  Note that $\scs^\uf_x(f) = \scs^\uf_x(f, f(\uf))$. We say that a subset $S \subseteq X \setminus \{ x \}$ is a \emph{candidate solution} of $\scs^\uf_x(f,c)$ if it satisfies $c = f(\switch{S}{\uf}) \neq f(\switch{x}{\switch{S}{\uf}})$. We say that $S$ is an \emph{optimal solution} if $|S|$ is minimal as well. When showing an inequality $\scs^\uf_x(f,c) \geq \scs^\vf_z(g,e)$, it suffices to show that every candidate solution of $\scs^\uf_x(f,c)$ translates to a candidate solution of $\scs^\vf_z(g,e)$ of equal cardinality.
\end{toappendix}

\begin{definition}[Share function]
    Call $\rho\colon \N \cup \{\infty\} \rightarrow \R$ a \emph{share function} if
    (i) $\rho$ is monotonically decreasing, 
    (ii) $\rho(\infty) = \lim_{n \rightarrow \infty} \rho(n) = 0$, and 
    (iii) $\rho(0) = 1$.
\end{definition}
\noindent In particular, we consider three instances of share functions:
    \begin{itemize}[noitemsep]
        \item $\rhoexp(k) = \nicefrac{1}{2^{k}}$,
        \item $\rhofrac(k) = \nicefrac{1}{(k{+}1)}$,
        \item $\rhostep(k) = 1$ for $k = 0$ and $\rho(k) = 0$ otherwise.
    \end{itemize} 
Given a share function $\rho$, the \emph{responsibility} of $x$ for $f$ under $\uf$ 
is defined as $\rho(\scs^\uf_x(f))$. Note that $\rhofrac(\scs^\uf_x(f))$ implements the
classical notion of responsibility~\cite{chocklerWhatCausesSystem2008}.
While responsibility corresponds to the size of the smallest critical set 
in a fixed assignment, CHK's \emph{blame}~\cite{chocklerWhatCausesSystem2008} is a global perspective and fits our notion of value function.
It is the expected value of the responsibility (we restrict ourselves to uniform distributions):

\begin{definition}[Blame]
    For a share function $\rho$, we define the \emph{$\rho$-blame} as value function 
    $\Blame^\rho$ %
    where for any $x\in X$, $f\in\boolfuncs$:
    \begin{center}
        $\Blame^\rho_x(f) = \E_{\uf \in \{0,1\}^X}[ \rho( \scs^{\uf}_x(f) ) ].$
    \end{center}
\end{definition}

\begin{example*}
  Let $f = x\lor y$. To compute the importance of $x$ we can count the number of times $\scs^\uf_x(f)=0,1,2,\dots,\infty$ occurs if $\uf$ ranges over the assignments for $\{x,y\}$: $\scs^\uf_x(f)=\infty$ happens once, $\scs^\uf_x(f) = 0$ happens twice, and $\scs^\uf_x(f)=1$ occurs once. Therefore, 
  $$
    \Blame^\rho_x(f) = \nicefrac 1 4 \cdot \rho(\infty) 
    + \nicefrac 1 2 \cdot \rho(0) 
    + \nicefrac 1 4 \cdot \rho(1),
  $$
  which is $\nicefrac 5 8$ for $\rho=\rhoexp$.
\end{example*}

\begin{toappendix}
Furthermore, let us introduce the following value, which measures the minimal number of changes one has to make to an assignment to make it satisfy a function:

\begin{definition}
    Let $f$ be a Boolean function and $\uf$ be an assignment over $X$. Define 
    \[ 
      \shortestDist^\uf(f) = \min_{S \subseteq X} |S| \quad\subjto\quad f(\switch{S}{\uf}) = 1.
    \]
    If there is no solution, set $\shortestDist^\uf(f) = \infty$.
\end{definition}

As for $\scs^\uf_x(f,c)$, we say that $S \subseteq X$ is \emph{a candidate solution} of $\shortestDist^\uf(f)$ if $f(\switch{S}{\uf}) = 1$, and if $|S|$ is minimal as well, we say that $S$ is \emph{an optimal solution}. It is now possible to establish some helpful identities regarding the smallest critical set. More precisely, we can see that it does not distinguish between functions and their complements (equality \eqref{eq:scs-prop-1}), that an assignment $\uf$ and its ``neighbour'' $\switch{x}{\uf}$ share similarities (equality \eqref{eq:scs-prop-2}), and that the smallest critical set does not care about variable names or polarities (equalities \eqref{eq:scs-prop-3} and \eqref{eq:scs-prop-4}). Finally, we observe that there is a useful decomposition of $\scs$ for the monotonically modular case (equality \eqref{eq:scs-prop-5}).

\begin{apxlemma}[Properties of $\scs$.]
    \label{thr:scsProperties}
    Suppose $g$ is a Boolean function, $\uf$ an assignment, $x,z$ variables, $c \in \{0,1\}$ a constant and $\sigma$ a permutation. Then
    \begin{align}
        \scs^\uf_x(g, c) &= \scs^\uf_x(\overline g, \overline c) \label{eq:scs-prop-1} \\
        &= \scs^{\switch{x}{\uf}}_x(g, \overline c) \label{eq:scs-prop-2} \\
        &= \scs^{\switch{z}{\uf}}_x(\switchf{z}{g}, c) \label{eq:scs-prop-3} \\
        &= \scs^{\sigma \uf}_{\sigma(x)}( \sigma g, c) \label{eq:scs-prop-4}.
    \end{align}
    If $f$ is monotonically modular in $g$ and $x \in \dep(g)$, then
    \begin{align}
        &\scs^\uf_x(f, c) = \scs^\uf_x(g, c) + \shortestDist^\uf(\cofactor f {\ass g 1} \overline{\cofactor f {\ass g 0}}). \label{eq:scs-prop-5} 
    \end{align}
\end{apxlemma}
\begin{proof}
    \hfill 
    \begin{enumerate}[leftmargin=1cm]
        \item[\eqref{eq:scs-prop-1}] One can simply check that every candidate solution $T$ of $\scs^\uf_x(g,c)$ is also a candidate solution of $\scs^\uf_x(\overline g, \overline c)$ and vice versa. So both values must be equal. The same line of argument is applicable to \eqref{eq:scs-prop-2} and \eqref{eq:scs-prop-3}. I.e., for all $T \subseteq X \setminus \{ x \}$,
        \begin{align*}
          c = g(\switch{T}{\uf}) \neq g(\switch{x}{\switch{T}{\uf}})
          &\Longleftrightarrow
          \overline{c} = \overline{g}(\switch{T}{\uf}) \neq \overline g(\switch{x}{\switch{T}{\uf}})
            &&\eqref{eq:scs-prop-1} \\
          &\Longleftrightarrow
          \overline c = g(\switch{T}{(\switch{x}{\uf})}) \neq g(\switch{x}{\switch{T}{(\switch{x}{\uf})}})
            &&\eqref{eq:scs-prop-2} \\
          &\Longleftrightarrow
          c = \switchf{z}{g}(\switch{T}{(\switch{z}{\uf})}) \neq \switchf{z}{g}(\switch{x}{\switch{T}{(\switch{z}{\uf})}})
            &&\eqref{eq:scs-prop-3}
        \end{align*}

        \item[\eqref{eq:scs-prop-4}] Note that for all $T \subseteq X$, by \eqref{eq:perm-prop-10},
        \begin{align*}
          & &&x \not\in T,\; c = g(\switch{T}{\uf}) \neq g(\switch{x}{\switch{T}{\uf}}) \\
          &\Longleftrightarrow
          &&\sigma(x) \not\in \sigma(T),\;
          c = \sigma g( \switch{\sigma(T)}{\sigma \uf} ) \neq \sigma g( \switch{\sigma(x)}{\switch{\sigma(T)}{\sigma \uf}}),
        \end{align*}
        Since $|T| = |\sigma(T)|$, the values $\scs^\uf_x(g,c)$ and $\scs^{\sigma \uf}_{\sigma(x)}(\sigma g, c)$ must be equal.

        \item[\eqref{eq:scs-prop-5}] 
        Let us show 
        \[ 
          \scs^\uf_x(f,c) = \scs^\uf_x(g,c) + \shortestDist^\uf(\cofactor f {\ass g 1} \overline{\cofactor f {\ass g 0}})
        \] 
        for $x \in \dep(g)$ by considering both directions ``$\leq$'' and ``$\geq$''. First however, note that due to $f$'s modularity, for all $G \subseteq \dep(g)$ and $F \subseteq X \setminus \dep(g)$,
        \begin{align*}
          f(\switch{G \cup F}{\uf}) &= g(\switch{G}{\uf}) \cofactor f {\ass g 1}(\switch{F}{\uf}) \lor \overline{g}(\switch{G}{\uf}) \cofactor f {\ass g 0}(\switch{F}{\uf}),
            && (*) \\
          f(\switch{x}{\switch{{G \cup F}}{\uf}}) &= g(\switch{x}{\switch{G}{\uf}}) \cofactor f {\ass g 1}(\switch{F}{\uf}) \lor \overline{g}(\switch{x}{\switch{G}{\uf}}) \cofactor f {\ass g 0}(\switch{F}{\uf}).
            && (**)
        \end{align*}
        \begin{itemize}
          \item[$\leq\colon$] Let $G,F \subseteq X$ be optimal solutions to $\scs^\uf_x(g,c)$ and $\shortestDist^\uf(\cofactor f {\ass g 1}\overline{\cofactor f {\ass g 0}})$, respectively. (If one of them does not exist, then the right side is $\infty$ and the inequality holds automatically.) I.e., we assume 
          \begin{align*}
            c = g(\switch{G}{\uf}) \neq g(\switch{x}{\switch{G}{\uf}}) \text{ and } 
            \cofactor f {\ass g 1}(\switch{F}{\uf}) = 1 \text{ and } \cofactor f {\ass g 0}(\switch{F}{\uf}) = 0.
          \end{align*}
          Furthermore, mind that $G \subseteq \dep(g)$ and $F \subseteq \dep(\cofactor f {\ass g 1}) \cup \dep(\cofactor f {\ass g 0}) \subseteq X \setminus \dep(g)$ by minimality of $|G|$ and $|F|$. Set $T = F \cup G$. Using $(*)$ and $(**)$,
          \begin{align*}
            \cofactor f {\ass g 1}(\switch{F}{\uf}) = 1 \text{ and } \cofactor f {\ass g 0}(\switch{F}{\uf}) = 0 
            \implies 
            f(\switch{T}{\uf}) = g(\switch{G}{\uf}) \text{ and } f(\switch{x}{\switch{T}{\uf}}) = g(\switch{x}{\switch{G}{\uf}}).
          \end{align*}
          Thus, 
          \begin{align*}
            c = g(\switch{G}{\uf}) \neq g(\switch{x}{\switch{G}{\uf}}) 
            \implies 
            c = f(\switch{T}{\uf}) \neq f(\switch{x}{\switch{T}{\uf}}),
          \end{align*}
          which shows that $T$ is a candidate solution to $\scs^\uf_x(f,c)$. As $F \cap G = \varnothing$, mind that 
          \[ 
            \scs^\uf_x(f,c) \leq |T| = |G|+|F| = \scs^\uf_x(g,c) + \shortestDist^\uf(\cofactor f {\ass g 1}\overline{\cofactor f {\ass g 0}}).
          \]
          \item[$\geq\colon$] Let $T \subseteq X$ be an optimal solution to $\scs^\uf_x(f,c)$, i.e. assume \[ 
            c = f(\switch{T}{\uf}) \neq f(\switch{x}{\switch{T}{\uf}}).
          \]
          Define $G = T \cap \dep(g)$ and $F = T \setminus \dep(g)$. Due to $G \subseteq \dep(g)$ and $F \subseteq X \setminus \dep(g)$ and $T = F \cup G$, mind that we can apply $(*)$ and $(**)$ to get 
          \begin{align*}
            f(\switch{T}{\uf}) \neq f(\switch{x}{\switch{T}{\uf}}) 
            \implies
            \cofactor f {\ass g 1}(\switch{F}{\uf}) \neq \cofactor f {\ass g 0}(\switch{F}{\uf}).
          \end{align*}
          Since $\cofactor f {\ass g 1} \geq \cofactor f {\ass g 0}$, this can only be the case if $F$ is a candidate solution to $\shortestDist^\uf(\cofactor f {\ass g 1} \overline{\cofactor f {\ass g 0}})$. By $(*)$ and $(**)$,
          \begin{align*}
            \cofactor f {\ass g 1}(\switch{F}{\uf}) = 1 \text{ and } \cofactor f {\ass g 0}(\switch{F}{\uf}) = 0
            \implies 
            f(\switch{T}{\uf}) = g(\switch{G}{\uf}) \text{ and }
            f(\switch{x}{\switch{T}{\uf}}) = g(\switch{x}{\switch{G}{\uf}}). 
          \end{align*}
          But then 
          \[ 
            c = g(\switch{G}{\uf}) \neq g(\switch{x}{\switch{G}{\uf}}),
          \]
          which shows that $G$ is a candidate solution to $\scs^\uf_x(g,c)$. Due to $G \cap F = \varnothing$,
          \[ 
            \scs^\uf_x(f,c) = |T| = |G| + |F| \geq \scs^\uf_x(g,c) + \shortestDist^\uf(\cofactor f {\ass g 1}\overline{\cofactor f {\ass g 0}}).
          \]
        \end{itemize}
        
    \end{enumerate}
\end{proof}

\end{toappendix}

Independent of $\rho$, the blame is always an IVF:

\begin{apxtheoremrep}
    \label{thr:blameSatisfactionOfProperties}
    \hspace{-.2em}$\Blame^\rho$\,is\,an\,unbiased\,IVF\,for\,any\,share\,function\,$\rho$.
\end{apxtheoremrep}
\begin{appendixproof}
  \hfill
  \begin{enumerate}[leftmargin=2.5cm]
    \item[$\Blame^\rho$ is unbiased] Note that \eqref{eq:scs-prop-1} also implies $\scs^\uf_x(g) = \scs^\uf_x(\overline g)$ for all Boolean functions $g$ and variables $x$. Then it is easy to see that $\Blame^\rho_x(g) = \Blame^\rho_x(\overline g)$.
    
    \item[$0 \leq \Blame^\rho_x(f) \leq 1$] Follows from the fact that $\Blame^\rho$ is a convex combination of terms limited by $0$ and $1$.
    \item[\ref{prop:DIC}] Let $x$ be a variable. We want to show $\Blame^\rho_x(x)=1$. $\Blame^\rho_x(\overline x)=1$ then follows from \ref{prop:TI}. Note that $\scs^\uf_x(x) = 0$ and thus $\rho(\scs^\uf_x(x)) = 1$ for all assignments $\uf$, so $\Blame^\rho_x(x) = 1$ follows.
    
    \item[\ref{prop:DUM}] Let $f$ be a Boolean function and $x \not\in \dep(f)$ be a variable. We want to show $\Blame^\rho_x(f) = 0$. Note that $\scs^\uf_x(f) = \infty$ for all assignments $\uf$ since flipping $x$ can never make a difference -- the set of candidate solutions to $\scs^\uf_x(f)$ must therefore be empty. But then $\rho(\scs^\uf_x(f)) =0$ for all assignments $\uf$ and thus $\Blame^\rho_x(f) = 0$.
   
    \item[\ref{prop:TI}] Suppose $f$ is a Boolean function, $\sigma$ a permutation and $x,y$ variables. We want to show (i) $\Blame^\rho_x(f) = \Blame^\rho_{\sigma(x)}(\sigma f)$ and (ii) $\Blame^\rho_x(f) = \Blame^\rho_x(\switchf{y}{f})$. To see (i), note that \eqref{eq:scs-prop-4} also implies $\scs^\uf_x(f) = \scs^{\sigma \uf}_{\sigma(x)}(\sigma f)$ due to the fact that $f(\uf )= \sigma f (\sigma \uf)$. Hence, 
    \[ 
      \Blame^\rho_x(f) = \E_{\uf \in \{0,1\}^X}[\rho( \scs^{\sigma \uf}_{\sigma(x)}(\sigma f))].
    \] 
    Using the fact that the mapping $\uf \mapsto \sigma \uf$ is a permutation of assignments over $X$ and that all assignments are equally likely, we get 
    \[ 
      \E_{\uf\in\{0,1\}^X}[\rho(\scs^{\sigma \uf}_{\sigma(x)}(\sigma f))] = \E_{\uf \in \{0,1\}^X}[\rho(\scs^\uf_{\sigma(x)}(\sigma f))] = \Blame^\rho_{\sigma(x)}(\sigma f).
    \] 
    For (ii), apply \eqref{eq:scs-prop-3} and $f(\uf) = \switchf{y}{f}(\switch{y}{\uf})$ to get $\scs^\uf_x(f) = \scs^{\switch{y}{\uf}}_x(\switchf{y}{f})$. Then argue analogously as for (i) and use the fact that $\uf \mapsto \switch{y}{\uf}$ is a permutation.

    \item[\ref{prop:SUD}] Let us assume that $f,h,g$ are Boolean functions such that (i) $f$ and $h$ are monotonically modular in $g$, (ii) $\cofactor f {\ass g 1} \geq \cofactor h {\ass g 1}$ and $\cofactor h {\ass g 0} \geq \cofactor f {\ass g 0}$ and (iii) $x \in \dep(g)$. We want to show $\Blame^\rho_x(f) \geq \Blame^\rho_x(h)$. We first show this assuming the following statement:
    \begin{equation}
      \begin{split}
        \rho(\scs^\uf_x(f)) + \rho(\scs^{\switch{x}{\uf}}_x(f)) 
        \geq \rho(\scs^\uf_x(h)) + \rho(\scs^{\switch{x}{\uf}}_x(h)) \label{eq:blameSUDInequality}
      \end{split}
      \qquad \text{for all } \uf \in \{0,1\}^X
    \end{equation}

    Since $\uf \mapsto \switch{x}{\uf}$ is a permutation and we're using the uniform distribution, we have:
    \begin{align*}
      \Blame^\rho_x(f) 
      = \textstyle\frac{1}{2} \E_{\uf \in \{0,1\}^X}[ \rho(\scs^\uf_x(f) )] + \textstyle\frac{1}{2} \E_{\uf \in \{0,1\}^X} [ \rho(\scs^{\switch{x}{\uf}}_x(f)) ]
    \end{align*}
    By linearity of $\E$ and inequality \eqref{eq:blameSUDInequality} we can conclude $\Blame^\rho_x(f) \geq \Blame^\rho_x(h)$.
    \\[0.1cm]
    \emph{Proof of (C.6):}
    If $f(\uf) \neq f(\switch{x}{\uf})$, then the left side is maximal and the inequality holds trivially. So let us assume $f(\uf) = f(\switch{x}{\uf})$. Furthermore, it must then be the case that $h(\uf) = h(\switch{x}{\uf})$: If $h(\uf) \neq h(\switch{x}{\uf})$, then $g(\uf) \neq g(\switch{x}{\uf})$, $\cofactor h {\ass g 1}(\uf) = 1$ and $ \cofactor h {\ass g 0}(\uf) = 0$ by (ii) and (iii), which implies $\cofactor f {\ass g 1}(\uf) = 1$ and $\cofactor f {\ass g 0}(\uf) = 0$ by (ii). But then $f(\uf) \neq f(\switch{x}{\uf})$, contradicting our assumption.
      
    Now that we have established $h(\uf) = h(\switch{x}{\uf})$, let us distinguish into two cases: $c = f(\uf) = h(\uf)$ and $c = f(\uf) \neq h(\uf)$. For the first case, using (ii) and \eqref{eq:scs-prop-5},
    \begin{align*}
      \scs^\uf_x(f) = \scs^\uf_x(g,c) + \shortestDist^\uf(\cofactor f {\ass g 1} \overline{\cofactor f {\ass g 0}}).
    \end{align*}
    By (ii), note that $\cofactor f {\ass g 1} \overline{\cofactor f {\ass g 0}} \geq \cofactor h {\ass g 1} \overline{\cofactor h {\ass g 0}}$, which implies 
    \begin{align*}
      \scs^\uf_x(f) 
      &\leq \scs^\uf_x(g,c) + \shortestDist^\uf(\cofactor h {\ass g 1} \overline{\cofactor h {\ass g 0}}) \\
      &= \scs^\uf_x(h)
    \end{align*}
    due to $h(\uf) = c$. The same reasoning can be applied to get $\scs^{\switch{x}{\uf}}_x(f) \leq \scs^{\switch{x}{\uf}}_x(h)$. This already suffices to obtain \eqref{eq:blameSUDInequality} by noting that $\rho$ is monotonically decreasing. For the case $c = f(\uf) \neq h(\uf)$, let us again apply decomposition \eqref{eq:scs-prop-5} to get 
    \begin{align*}
      \scs^\uf_x(f) 
        &= \scs^\uf_x(g,c) + \shortestDist^\uf(\cofactor f {\ass g 1} \overline{\cofactor f {\ass g 0}}) \\
        &= \scs^{\switch{x}{\uf}}_x(g,\overline c) + \shortestDist^{\switch{x}{\uf}}(\cofactor f {\ass g 1} \overline{\cofactor f {\ass g 0}})  
          && \eqref{eq:scs-prop-2} \text{ and } x \not\in \dep(\cofactor f {\ass g 1} \overline{\cofactor f {\ass g 0}}) \\
        &\leq \scs^{\switch{x}{\uf}}_x(g,\overline c) + \shortestDist^{\switch{x}{\uf}}(\cofactor h {\ass g 1} \overline{\cofactor h {\ass g 0}})
          && \cofactor f {\ass g 1} \overline{\cofactor f {\ass g 0}} \geq \cofactor h {\ass g 1} \overline{\cofactor h {\ass g 0}} \\
        &= \scs^{\switch{x}{\uf}}_x(h).
          && h(\switch{x}{\uf}) = \overline c
    \end{align*}
    We can derive $\scs^{\switch{x}{\uf}}_x(f) \leq \scs^{\uf}_x(h)$ analogously. Using again the fact that $\rho$ is monotonically decreasing, we can conclude that \eqref{eq:blameSUDInequality} holds. 
    
  \end{enumerate}
\end{appendixproof}

\begin{toappendix}
  \begin{table}
    \center
    \resizebox*{0.9\textwidth}{!}{
    \begin{tabular}{rrr rrrr rr rrrr rr}
      \toprule
      \multicolumn{3}{c}{input $\uf$} & \multicolumn{6}{c}{functions} & \multicolumn{6}{c}{$\scs$-values} \\ 
        \cmidrule(r){1-3} \cmidrule(r){4-9} \cmidrule{10-15}
      $x$ & $y$ & $z$ 
        & $f_1(\uf)$ & $f_2(\uf)$ & $f_3(\uf)$ & $f_4(\uf)$ 
        & $f(\uf)$ & $h(\uf)$ & 
        $\scs^{\uf}_x(f_1)$ & $\scs^{\uf}_x(f_2)$ & $\scs^{\uf}_x(f_3)$ & $\scs^{\uf}_x(f_4)$ & $\scs^{\uf}_x(f)$ & $\scs^{\uf}_x(h)$ \\
      \midrule
      0 & 0 & 0 & 0 & 0 & 0 & 0 & 0 & 0 & 0 & 0 & 0 & 0 & 1 & 1 \\
      0 & 0 & 1 & 1 & 1 & 0 & 0 & 1 & 0 & 2 & $\infty$ & 0 & 0 & 0 & 0 \\
      0 & 1 & 0 & 1 & 1 & 1 & 1 & 0 & 0 & 0 & $\infty$ & $\infty$ & 0 & 0 & 0 \\
      0 & 1 & 1 & 1 & 1 & 1 & 1 & 1 & 1 & 1 & $\infty$ & $\infty$ & 0 & 1 & $\infty$ \\
      \midrule
      1 & 0 & 0 & 1 & 1 & 1 & 1 & 0 & 0 & 0 & 0 & 0 & 0 & 1 & $\infty$ \\
      1 & 0 & 1 & 1 & 1 & 1 & 1 & 0 & 1 & 1 & 1 & 0 & 0 & 0 & 0 \\
      1 & 1 & 0 & 0 & 1 & 1 & 0 & 1 & 1 & 0 & 1 & 1 & 0 & 0 & 0 \\
      1 & 1 & 1 & 1 & 1 & 1 & 0 & 1 & 1 & 2 & 2 & 1 & 0 & 1 & 1 \\
      \bottomrule
    \end{tabular}}
    \caption{Functions and $\scs$-values from \Cref{thr:limitationsOfBlame}. \label{tab:limitationsOfBlame}}
  \end{table}
\end{toappendix}

\begin{toappendix}
  \Cref{thr:limitationsOfBlame} shows that some properties are satisfied only if one chooses a very specific share function. Proofs utilize Table \ref{tab:limitationsOfBlame}, which contains $\scs$-values for six different functions. We combine implications of properties such as derivative dependence for these functions in order to derive statements about $\rho$:
\end{toappendix}

In full generality, the blame violates the optional properties for IVFs (see \Cref{sec:further_properties}).
For example, if $\rho \neq \rhostep$, then the $\rho$-blame 
is neither chain-rule decomposable nor derivative dependent,
and one can find counterexamples for the rank-preservation property 
for $\rhofrac$ and $\rhoexp$:

\begin{apxpropositionrep}
  \label{thr:limitationsOfBlame}
  Let $\rho$ be a share function. Then the following statements are equivalent:
  \begin{enumerate}[noitemsep]
    \item[\textnormal{(i)}] $\Blame^\rho$ is weakly chain-rule decomposable,
    \item[\textnormal{(ii)}] $\Blame^\rho$ is derivative dependent, and
    \item[\textnormal{(iii)}] $\rho = \rhostep$. 
  \end{enumerate}
  Further, neither $\Blame^\rhofrac$ nor $\Blame^\rhoexp$ are weakly rank preserving.
\end{apxpropositionrep}
\begin{proof}
  For the equivalence of (i), (ii) and (iii), note that if $\rho = \rhostep$, then $\Blame^\rho$ is simply the influence, which is both chain rule decomposable and derivative dependent as we will see later on (cf. \Cref{thr:influenceMBlameSpecialCase}). Therefore, let us focus on the directions from (i) to (iii) and (ii) to (iii), i.e. we assume that $\Blame^\rho$ is chain rule decomposable (resp. derivative dependent) and we want to show that this implies $\rho = \rhostep$. 
 
  \begin{enumerate}[leftmargin=2.5cm]
    \item[(i) implies (iii):] Define 
    \begin{align*}
      &f_1 = (x \oplus y) \lor z \\
      &f_2 = x \lor y \lor z \\
      &f_3 = x \lor y \\
      &f_4 = x \oplus y
    \end{align*}  
    And we get by $\rho(0) = 1$ and $\rho(\infty) = 0$, using Table \ref{tab:limitationsOfBlame},
    \begin{align*}
      \Blame^\rho_x(f_1) &= \textstyle\frac{1}{2^3}(4 + 2\rho(1) + 2 \rho(2)), \\
      \Blame^\rho_x(f_2) &= \textstyle\frac{1}{2^3}(2 + 2\rho(1) + \rho(2)), \\
      \Blame^\rho_x(f_3) &= \textstyle\frac{1}{2^3}(4 + 2\rho(1)), \\
      \Blame^\rho_x(f_4) &= 1.
    \end{align*}
    Since $\Imp$ is weakly chain rule decomposable and satisfies \ref{prop:TI}, we can see that $\Blame^\rho_x(f_1) = \Blame^\rho_x(f_4)\Blame^\rho_x(f_3) = \Blame^\rho_x(f_3)$ and $\Blame^\rho_x(f_2) = \Blame^\rho_x(f_3)^2$. Solving the first equation implies $\rho(2) = 0$.
    Observe that we have:
    \begin{align*}
    &\Blame^\rho_x(f_3)^2 = \frac{(4+2\rho(1))^2}{64} = \frac{16+16\rho(1)+4\rho(1)^2}{64} 
    \end{align*}
    Now $\Blame^\rho_x(f_3)^2 = \Blame^\rho_x(f_2)$ together with $\rho(2) = 0$ yields:
    \begin{alignat*}{2}
    \Blame^\rho_x(f_3)^2 \; = \; &\frac{16+16\rho(1)+4\rho(1)^2}{64} \; &&= \;\frac{1+\rho(1)}{4} \; =  \; \Blame^\rho_x(f_2)\\
    \iff &\quad 16+16\rho(1)+4\rho(1)^2\; &&=\; 16 + 16\rho(1) \\
    \iff &\quad \rho(1) \; &&=\; 0
    \end{alignat*}
    Since $\rho$ is monotonically decreasing but bounded by zero, $\rho(n) = 0$ for $n \geq 1$. But then $\rho = \rhostep$.
    
    \item[(ii) implies (iii):] Define
    \begin{align*}
      f &= x  y \oplus \overline x z \\
      h &= x  (y \lor z) \oplus \overline x  y  z.
    \end{align*} 
    Note that $\Adiff_x f = \Adiff_x h$. Furthermore, using Table \ref{tab:limitationsOfBlame},
    \begin{align*}
      \Blame^\rho_x(f) &= \textstyle\frac{1}{2^3} (4 + 4 \rho(1)) \\
      \Blame^\rho_x(h) &= \textstyle\frac{1}{2^3} (4 + 2 \rho(1)).
    \end{align*}
    Since $\Imp$ is derivative dependent, we can see that $\Blame^\rho_x(f) = \Blame^\rho_x(h)$ must hold, which implies $\rho(1) = 0$ after resolving both equations. But then $\rho(n) = 0$ for all $n \geq 1$ since $\rho$ is monotonically decreasing. Hence, $\rho = \rhostep$.
  \end{enumerate}
  For the last part, we have automatically generated counterexamples for $\Blame^\rhoexp$ and $\Blame^\rhofrac$. For $\Blame^\rhoexp$, let us consider the following functions,
  \begin{align*}
    g &= \overline x_1 \overline x_0  \overline x_2 \lor 
        x_1  x_0  x_2 \lor 
        x_3  \overline x_0  \overline x_2 \lor 
        x_3  x_0  x_2 \lor 
        x_3  x_1, \\
    f &= g \lor z.
  \end{align*}
  We can clearly see that $f$ is monotonically modular in $g$. Furthermore, $\Blame^\rhoexp$ results in the following ranking,
  \begin{align*}
    &g: \quad x_3:0.6250 < x_0:0.7188 = x_2:0.7188 < x_1:0.7266, \\
    &f: \quad x_3:0.4062 < x_1:0.4980 < x_0:0.5000 = x_2:0.5000 < z:0.6172.
  \end{align*}
  Note how the positions of $x_1$ and $x_0$ are flipped when extending $g$ to $f$, so $\Blame^\rhoexp$ is not weakly rank preserving. A counterexample for $\Blame^\rhofrac$ is given by 
  \begin{align*}
    g &= x_1  \overline x_0  \overline x_2 \lor \overline x_1  x_0 \lor x_3, \\
    f &= g \lor z,
  \end{align*}
  which results in the ranking
  \begin{align*}
    &g: \quad x_2:0.2969 < x_1:0.6302 = x_0:0.6302 < x_3:0.7188 \\
    &f: \quad x_2:0.2250 < z:0.4688 = x_3:0.4688 < x_1:0.4802 = x_0:0.4802
  \end{align*}
  and reverses the positions of $x_0$ and $x_3$.
\end{proof}

To give an example for the reason why the $\rhofrac$-blame is not
weakly rank preserving, consider $g = x_1  \overline x_0 \overline x_2 \lor \overline x_1  x_0 \lor x_3$ and $f = g \lor z$. 
Note that $f$ is clearly monotonically modular in $g$ -- only $z$ is added as fresh variable. 
Nevertheless, the order of $x_0$ and $x_3$ changes:
\begin{align*}
  &\Blame^\rhofrac_{x_0}(g) = 0.6302 < 0.7188 = \Blame^\rhofrac_{x_3}(g) \\
  &\Blame^\rhofrac_{x_0}(f) = 0.4802 > 0.4688 = \Blame^\rhofrac_{x_3}(f).
\end{align*}
Intuitively, this is because by CHK's definition of critical sets: for all 
Boolean functions $h$, variables $x$ and assignments $\uf$,
\begin{center} $h(\uf) = 1, \cofactor h {\ass x 1} \geq \cofactor h {\ass x 0}, \uf(x)=0 \implies \scs^\uf_x(h) = \infty.$ \end{center}
Hence, whenever an assignment $\uf$ satisfies the premise for $x$ in $h$, the responsibility of $x$ for $h$ under $\uf$ will be zero.

For $x_3$, this is more frequently the case in $g$ than in $f$ ($19\%$ vs. $34\%$ of all assignments).
On the other hand, there is always a critical set for $x_0$ in both $f$ and $g$.
Partly for this reason, the importance of $x_3$ decreases more than $x_0$
when switching from $g$ to $f$.

\subsubsection{Modified Blame}
\begin{toappendix}
    \subsubsection*{Modified Blame} 
\end{toappendix}
We modify the definition of critical sets
in order to derive a \emph{modified blame} that 
satisfies more optional properties for a wider class of share functions.

For a Boolean function $f$, an assignment $\uf$ over $X$ and a variable $x$, the \emph{modified $\scs$} 
is defined as the size $\mscs^\uf_x(f)$ of the smallest set $S \subseteq X\setminus \{ x \}$ that satisfies 
\begin{center} $f\big(\switchmain{S}{\uf}\big) \neq f\big(\switchmain{S\cup\{x\}}{\uf}\big).$ \end{center} 
If there is no such set, we set $\mscs^\uf_x(f) = \infty$.

\begin{example*}
  The condition for critical sets is relaxed, hence $\mscs^\uf_x(f)$ provides a lower bound for $\scs^\uf_x(f)$.
  Let for example $f = x \lor y$ and $\uf = \ass x 0; \ass y 1$. Then 
  \begin{center}
    $\mscs^\uf_x(f) = 1 < \infty = \scs^\uf_x(f)$.
  \end{center}
\end{example*}

The definitions for responsibility and blame are analogous for the modified version, replacing $\scs$ by $\mscs$. We denote by $\MBlame^\rho$ the \emph{modified $\rho$-blame},
which is (in contrast to $\Blame^\rho$) always derivative dependent and even 
chain-rule decomposable if $\rho$ is an exponential- or stepping-function:
\begin{toappendix}
  Like \eqref{eq:scs-prop-5}, there is a way to decompose the value $\mscs_x^\uf(f)$ if $f$ is modular in $g$ and $x \in \dep(g)$. However, we can express this decomposition using only $\mscs$:

  \begin{apxlemma}
    \label{thr:decompositionmscs}
    If $f$ is modular in $g$, $\uf$ an assignment over $X$ and $x \in \dep(g)$, then
    \begin{center} %
      $\mscs^\uf_x(f) = \mscs^\uf_x(g) + \mscs^\uf_g (f).$
    \end{center} %
  \end{apxlemma}
  \begin{proof}
  Suppose $S$ is an optimal solution to $\mscs^\uf_x(f)$, and define $G =  S \cap \dep(g)$ and $F = S \cap (X \setminus \dep(g))$. By \Cref{thr:derivativeDecomposition},
    \begin{align*}
      \Adiff_x f(\switch{S}{\uf}) = 1 \quad\Rightarrow\quad \Adiff_x g(\switch{G}{\uf}) = \Adiff_{g}f(\switch{F}{\uf}) = 1,
    \end{align*}
    so $G$ is a candidate solution to $\mscs^\uf_x(g)$ and $F$ a candidate solution to $\mscs^\uf_g(f)$. But then 
    \[ 
      \mscs^\uf_x(f) = |S| = |G|+|F| \geq \mscs^\uf_x(g) + \mscs^\uf_g(f). 
    \]
    For the other direction, suppose that $G$ resp. $F$ are optimal solutions to $\mscs^\uf_x(g)$ and $\mscs^\uf_g(f)$. Define $S = G \cup F$. By \Cref{thr:derivativeDecomposition},
    \begin{align*}
      \Adiff_x g(\switch{G}{\uf}) = \Adiff_{g}f(\switch{F}{\uf}) = 1 \quad\Rightarrow\quad \Adiff_x f(\switch{S}{\uf}) = 1,
    \end{align*}
    which shows that $S$ is a candidate solution to $\mscs^\uf_x(f)$ and thus 
    \[ 
      \mscs^\uf_x(f) \leq |S| \leq |G| + |F| = \mscs^\uf_x(g) + \mscs^\uf_g(f).
    \]
  \end{proof}

    By \Cref{thr:implicationsOfDD}, if $\MBlame^\rho$ satisfies the properties \ref{prop:DIC}, \ref{prop:DUM}, derivative dependence and \ref{prop:TI}, then it must be an unbiased IVF. For the chain-rule property, we rely on \Cref{thr:decompositionmscs}:
\end{toappendix}

\begin{apxtheoremrep}
    \label{thr:MBlameSatisfactionOfProperties}
    $\MBlame^\rho$ is an unbiased, derivative-dependent IVF for any share function $\rho$. 
    If there is $0 \leq \lambda < 1$ so that $\rho(k) = \lambda^k$ for all $k \geq 1$, then $\MBlame^\rho$ is chain-rule decomposable.
\end{apxtheoremrep}
\begin{appendixproof}
  Let us first show that $\MBlame^\rho$ is indeed an unbiased importance value function.
  \begin{enumerate}[leftmargin=1.8cm]
  
    \item[\ref{prop:DIC}] Let $x$ be a variable. We want to show $\MBlame^\rho_x(x) = 1$. $\MBlame^\rho_x(\overline x)=1$ then follows from \ref{prop:TI}. Note that since $\Adiff_x x = 1$, we have $\mscs^\uf_x(x) = 0$ for all assignments $\uf$. Thus, since $\MBlame^\rho_x(x)$ is a convex combination of $\rho(0) = 1$'s, we get $\MBlame^\rho_x(x) = 1$.
    
    \item[\ref{prop:DUM}] Let $f$ be a Boolean function and $x \not\in \dep(f)$ be a variable. We want to show $\MBlame^\rho_x(f) = 0$. Note that $\Adiff_x f = 0$, which implies $\mscs^\uf_x(f) = \infty$ for all assignments $\uf$. Since $\MBlame^\rho_x(f)$ is a convex combination of $\rho(\infty) = 0$'s, we get $\MBlame^\rho_x(f) = 0$.

    \item[deriv. dep.] Let $f,h$ be Boolean functions and $x$ a variable such that $\Adiff_x f \geq \Adiff_x h$. We want to show $\MBlame^\rho_x(f) \geq \MBlame^\rho_x(h)$. First mind that the inequality
    $ %
      \mscs^\uf_x(f) \leq \mscs^\uf_x(h)
    $ %
    holds for all assignments $\uf$, which follows from the fact that $\mscs^\uf_x(f) = \shortestDist^\uf(\Adiff_x f)$ and $\mscs^\uf_x(h) = \shortestDist^\uf(\Adiff_x h)$ and by noting that $\ell \geq r$ implies $\shortestDist^\uf(\ell)  \leq \shortestDist^\uf(r)$ (every candidate solution to $\shortestDist^\uf(r)$ is automatically a candidate solution to $\shortestDist^\uf(\ell)$.) Turning back to derivative dependence, we us the fact that $\rho$ is monotonically non-increasing to obtain 
    \begin{align*}
      \MBlame^\rho_x(f) = \E_{\uf \in \{0,1\}^X}[ \rho(\mscs^\uf_x(f) ) ] \geq \E_{\uf \in \{0,1\}^X}[ \rho( \mscs^\uf_x(h)) ] = \MBlame^\rho_x(f).
    \end{align*}
    
    \item[\ref{prop:TI}] Let $f$ be a Boolean function, $\sigma$ a permutation and $x,y$ variables. We want to show (i) $\MBlame^\rho_x(f) = \MBlame^\rho_{\sigma(x)}(\sigma f)$ and (ii) $\MBlame^\rho_x(f) = \MBlame^\rho_x(\switchf{y}{f})$. For (i), let us first note that 
    $ %
      \mscs^\uf_x(f) = \mscs^{\sigma \uf}_{\sigma(x)}(\sigma f)
    $ %
    holds for all assignments $\uf$. We denote this fact by $(*)$. Too see why $(*)$ holds, note that
    \[ 
      \cofactor f {\ass x c}(\uf) = \cofactor {(\sigma f)} {\ass {\sigma(x)} c} (\sigma \uf)
    \] 
    for $c \in \{0,1\}$ due to \eqref{eq:perm-prop-4}, thus implying 
    \[ 
      \Adiff_x f(\uf) = \Adiff_{\sigma(x)} \sigma f(\sigma \uf),
    \]
    which shows that for all $S \subseteq X$,
    \[ 
      \Adiff_x f(\switch{S}{\uf}) = \Adiff_{\sigma(x)} \sigma f(\sigma(\switch{S}{\uf})) ) = \Adiff_{\sigma(x)} \sigma f(\switch{\sigma(S)}{\sigma \uf})
    \]
    due to \eqref{eq:perm-prop-10}. Thus, every candidate solution $S$ of $\mscs^\uf_x(f)$ translates to the candidate solution $\sigma(S)$ of $\mscs^{\sigma \uf}_{\sigma(x)}(\sigma f)$ and vice versa. Since $\sigma(S)$ has the same size as $S$, $(*)$ must follow. Finally,
    \begin{align*} 
      \MBlame^\rho_x(f) 
      &= \E_{\uf \in \{0,1\}^X}[\rho(\mscs^\uf_x(f))] 
        && \text{definition} \\
      &= \E_{\uf \in \{0,1\}^X}[\rho(\mscs^{\sigma \uf}_{\sigma(x)}(\sigma f))] 
        && (*) \\
      &= \E_{\uf \in \{0,1\}^X}[\rho(\mscs^\uf_{\sigma(x)}(\sigma f))]
        && \uf \mapsto \sigma \uf \text{ permutation of } \{0,1\}^X \\
      &= \MBlame^\rho_{\sigma(x)}(\sigma f).
    \end{align*}

    For (ii), note that for all $\uf$ over $X$, 
    \begin{align*} 
      \Adiff_x f(\uf) = 0 
      &\tiff f(\uf) = f(\switch{x}{\uf}) \\
      &\tiff \switchf{y}{f}(\switch{y}{\uf}) = \switchf{y}{f}(\switch{x}{(\switch{y}{\uf})}) \\
      &\tiff \Adiff_x (\switchf{y}{f})(\switch{y}{\uf}) = 0,
    \end{align*}
    which shows $\Adiff_x f(\uf) = \Adiff_x (\switchf{y}{f})(\switch{y}{\uf})$. Thus, $S$ is a candidate solution to $\mscs^\uf_x(f)$ iff it is a candidate solution to $\mscs^{\switch{y}{\uf}}_x(\switchf{y}{f})$. But this must mean that both values are equal, and we get $\MBlame^\rho_x(f) = \MBlame^\rho_x(\switchf{y}{f})$ using the fact that the mapping $\uf \mapsto \switch{y}{\uf}$ constitutes a bijection.
  \end{enumerate}

  \noindent Finally, let us assume that $\rho(k) = \lambda^k$ for a $0 \leq \lambda < 1$ and all $k \geq 1$. Then, clearly, $\rho(a+b) = \rho(a)\rho(b)$. Too see that $\MBlame^\rho$ is chain rule decomposable, suppose $f$ is modular in $g$ and $x \in\dep(g)$ a variable. We want to show $\MBlame^\rho_x(f) = \MBlame^\rho_x(g) \MBlame^\rho_g(f)$. This result can be obtained by application of \Cref{thr:decompositionmscs},
  \begin{align*}
    \MBlame^\rho_x(f) 
    &= \E_{\uf \in \{0,1\}^X}[ \rho( \mscs^\uf_x(f) )] \\
    &= \E_{\uf \in \{0,1\}^X}[ \rho( \mscs^{\uf}_x(g) + \mscs^{\uf}_g(f) )] \\
    &= \E_{\uf \in \{0,1\}^X}[ \rho( \mscs^{\uf}_x(g) ) \rho( \mscs^{\uf}_g(f) )] \\
    &= \E_{\uf \in \{0,1\}^{X}}[ \rho(\mscs^\uf_x(g))] \E_{\uf \in \{0,1\}^{X}}[\rho(\mscs^\uf_g(f))] \\
    &=  \MBlame^\rho_x(g) \MBlame^\rho_g(f).
  \end{align*}
  The penultimate equality is due to \Cref{thr:decompositionOfExpectations}, since $\dep( \rho(\mscs^{(\cdot)}_x(g)) ) \subseteq \dep(g)$ and $\dep( \rho(\mscs^{(\cdot)}_g(f)) ) \subseteq X \setminus \dep(g)$.
\end{appendixproof}

\subsection{Influence}\label{ssec:instances_influence}
\begin{toappendix}
    \subsection*{Influence} 
\end{toappendix}

The influence~\cite{ben-orCollectiveCoinFlipping1985,kahnInfluenceVariablesBoolean1988,o2014analysis} is a popular importance measure,
defined as the probability that flipping the variable changes the function's outcome 
for uniformly distributed assignments: 

\begin{definition}
  The \emph{influence} is the value function $\Inf$ defined by 
  $
    \Inf_x(f) = \E[ \Adiff_x f]
  $
  for all $f \in \boolfuncs$ and variables $x \in X$.
\end{definition}

It turns out that the influence is a special case of blame:
\begin{apxpropositionrep}
  \label{thr:influenceMBlameSpecialCase}
  $
    \Inf = \MBlame^\rhostep = \Blame^\rhostep.
  $
\end{apxpropositionrep}
\begin{appendixproof}
  Follows from the fact that $\Adiff_x f(\uf) = \rhostep( \mscs^{\uf}_x(f) ) = \rhostep(\scs^\uf_x(f))$ for all Boolean functions $f$, variables $x$ and assignments $\uf$.
\end{appendixproof}

Since $\rhostep(k)=0^k$ for $k \geq 1$, \Cref{thr:influenceMBlameSpecialCase} and \Cref{thr:MBlameSatisfactionOfProperties} 
show that the influence is a derivative-dependent, rank-preserving, and chain-rule decomposable IVF.

\subsubsection{Characterizing the Influence}
Call a value function $\Imp$ \emph{cofactor-additive} if
for all Boolean functions $f$ and variables $x \neq z$:
\begin{center}
  $\Imp_x(f) = \nicefrac{1}{2} \cdot \Imp_x(\cofactor f {\ass z 0}) + \nicefrac{1}{2} \cdot \Imp_x(\cofactor f {\ass z 1}).$
\end{center}
Using this notion, we \emph{axiomatically characterize} the influence as follows. 

\begin{toappendix}
  The additional additivity property states that the importance of a variable to a function should always be the average importance with respect to the cofactors of another variable. We can applying this property repeatedly to obtain special cases which are covered by \ref{prop:DUM} and \ref{prop:DIC}:
\end{toappendix}

\begin{apxtheoremrep}
  \label{thr:influenceAxiomatic}
A value function $\Imp$ satisfies \ref{prop:DIC}, \ref{prop:DUM}, and \emph{cofactor-additivity} if and only if $\Imp = \Inf$.
\end{apxtheoremrep}
\begin{appendixproof}
  Let us first establish that $\Inf$ indeed satisfies all of these properties. \Cref{thr:influenceMBlameSpecialCase} suffices to show that $\Inf$ satisfies \ref{prop:DIC} and \ref{prop:DUM}. For cofactor-additivity, note that
  \begin{align*}
    \Inf_x(f)
     & = \textstyle\frac{1}{2^{n}} \sum_{\uf \in \{0,1\}^X} \Adiff_x f(\uf)
     &                                                                                                                          & \text{definition}                            \\
     & = \textstyle\frac{1}{2^{n}} \sum_{\uf \in \{0,1\}^{X \setminus \{ z \}}} \Adiff_x f(\ass z 0; \uf) + \Adiff_x f(\ass z 1;\uf)
     &                                                                                                                          & \text{expanding w.r.t. } z \\
     & = \textstyle\frac{1}{2^{n}} \sum_{\uf \in \{0,1\}^{X \setminus \{ z \}}} \Adiff_x (\cofactor f {\ass z 0})(\uf) + \Adiff_x (\cofactor f {\ass z 1})(\uf)
     &                                                                                                                          & \Adiff_x f(\ass z c; \uf) = \Adiff_x (\cofactor f {\ass z c})(\uf) \text{ for } \uf \in \{0,1\}^{X \setminus \{ z\}}\\
     & = \textstyle\frac{1}{2}\frac{1}{2^{n}} \sum_{\uf \in \{0,1\}^X}\Adiff_x (\cofactor f {\ass z 0})(\uf) + \Adiff_x (\cofactor f {\ass z 1})(\uf)
     &                                                                                                                          & z\not\in \dep(\Adiff_x (\cofactor f {\ass z c}))              \\
     & = \textstyle\frac{1}{2} \Inf_x(\cofactor f {\ass z 0}) + \frac{1}{2}\Inf_x(\cofactor f {\ass z 1}).
     &                                                                                                                          & \text{definition}
  \end{align*}
  For the other direction, let us assume that $\Imp$ is cofactor-additive and that it satisfies \ref{prop:DIC} and \ref{prop:DUM}. Use cofactor-additivity repeatedly on every variable except for $x$ to get
  \begin{align*}
    \Imp_x(f) 
    &= \textstyle\frac 1 2 \Imp_x(\cofactor f {\ass {z_1} 0}) + \frac 1 2 \Imp_x(\cofactor f {\ass {z_1} 1}) \\
    &= \textstyle\frac 1 4 \Imp_x(\cofactor f {\ass {z_1} 0; \ass {z_2} 0}) + \frac 1 4 \Imp_x(\cofactor f {\ass {z_1} 0; \ass {z_2} 1}) 
      + \frac 1 4 \Imp_x(\cofactor f {\ass {z_1} 1; \ass {z_2} 0}) + \frac 1 4 \Imp_x(\cofactor f {\ass {z_1} 1; \ass {z_2} 1})  \\
    &= \dots \\
    &=  \textstyle\frac{1}{2^{n-1}} \sum_{\uf \in \{ 0,1\}^{X \setminus \{ x \}}} \Imp_x(\cofactor f \uf).
  \end{align*}
  Since $\uf$ assigns all variables except for $x$ a value, $\cofactor f {\uf}$ can only depend on $x$. So there are four possible functions that $\cofactor f \uf$ can represent: $f = x$, $f = \overline x$, $f = 0$ or $f= 1$. In the latter two cases, use \ref{prop:DUM} to get $\Imp_x(\cofactor f \uf) = 0$. In the former two cases, we can apply \ref{prop:DIC} to obtain $\Imp_x(\cofactor f \uf) = 1$. But then $\Imp_x(\cofactor f \uf) = \Adiff_x f(\uf)$ and therefore
  $ %
    \Imp_x(f)
    = \Inf_x(f)
  $ %
  after writing $\Imp_x(f)$ as a sum of assignments over $X$ instead of $X \setminus \{x \}$.
\end{appendixproof}

\begin{remark*}
A relaxed version of \emph{cofactor-additivity} assumes the existence of 
$\alpha_z,\beta_z \in \R$ for $z \in X$ such that for all $x \neq z$:
\begin{center}
  $\Imp_x(f) = \alpha_z \Imp_x(\cofactor f {\ass z 0}) + \beta_z \Imp_x(\cofactor f {\ass z 1})$.
\end{center}
This, together with the assumption that $\Imp$ satisfies \ref{prop:TI}, \ref{prop:DUM} and \ref{prop:DIC}, implies 
$\alpha_z = \beta_z = \nicefrac 1 2$. Hence, another characterization of the influence consists of \ref{prop:TI}, \ref{prop:DUM}, \ref{prop:DIC}, and \emph{relaxed cofactor-additivity}. 
\ifx\arxiv\undefined
\else
(See \Cref{thr:influenceAxiomatic2} in the appendix.)
\fi
\end{remark*}

\begin{toappendix}
  A different axiomatic characteristic of the influence uses the \emph{relaxed cofactor-additvity} property, which is defined as follows.
  Say that a value function $\Imp$ is \emph{relaxed cofactor-additive} if there are $\alpha_z,\beta_z \in \R$ for every $z \in X$ such that 
  for all $x \neq z$:
  \begin{align*}
    \Imp_x(f) = \alpha_z \Imp_x(\cofactor f {\ass z 0}) + \beta_z \Imp_x(\cofactor f {\ass z 1}).
  \end{align*}
  If we \emph{additionally} assume that $\Imp$ satisfies \ref{prop:TI}, \ref{prop:DIC} and \ref{prop:DUM}, then $\alpha_z = \beta_z = \nicefrac 1 2$, and so 
  it must be the influence:
  \begin{corollary}
    \label{thr:influenceAxiomatic2}
    A value function $\Imp$ satisfies \ref{prop:TI}, \ref{prop:DIC}, \ref{prop:DUM}, and \emph{relaxed cofactor-additivity} if and only if $\Imp = \Inf$.
  \end{corollary}
  \begin{proof}
    The direction from right to left is clear since the influence is a cofactor-additive IVF (cf. \Cref{thr:influenceAxiomatic}).
    For the direction from left to right, we only need to show that $\alpha_z = \beta_z = \nicefrac 1 2$ holds: since $\Imp$ is then 
    cofactor-additive, \Cref{thr:influenceAxiomatic} implies that it must be the influence.

    Pick an arbitrary $z \in X$ and choose $x \neq z$.
    Apply \emph{relaxed cofactor-additvity} and \ref{prop:DIC} to the dictator function:
    \begin{align*}
      \Imp_x(x) = \alpha_z\Imp_x(x) + \beta_z \Imp_x(x) \implies \beta_z + \alpha_z = 1.
    \end{align*}
    Moreover, note that due to \emph{relaxed cofactor-additvity}, \ref{prop:DIC} and \ref{prop:DUM},
    \begin{align*}
      \Imp_x(x z) &= \alpha_z \Imp_x(x) + \beta_z \Imp_x(0) = \alpha_z, \\
      \Imp_x(x \overline z) &= \alpha_z \Imp_x(0) + \beta_z \Imp_x(x) = \beta_z.
    \end{align*}
    The application of \ref{prop:TI} yields $\Imp_x(xz) = \Imp_x(x \overline z)$, and so: $\alpha_z = \beta_z =\nicefrac{1}{2}$.
  \end{proof}
\end{toappendix}

Moreover, we give a \emph{syntactic characterization} of the influence by a comparison
to the two-sided Jeroslow-Wang heuristic used for 
SAT-solving~\cite{jeroslowSolvingPropositionalSatisfiability1990,hookerBranchingRulesSatisfiability1995,marques-silvaImpactBranchingHeuristics1999}. This value is defined for families of sets of literals, which are sets of subsets of $X \cup \{ \overline z : z \in X \}$, and it weights subsets that contain $x$ or $\overline x$ by their respective lengths:

\begin{definition}[\cite{hookerBranchingRulesSatisfiability1995}]
  Let $\Dc$ be a family of sets of literals. The \emph{two-sided Jeroslow-Wang value} 
  for a variable $x$ is defined as
  \begin{center}
    $\JW_x(\Dc) = \textstyle\sum_{ C \in \Dc \ \text{s.t.} \ x \in C \text{ or } \overline x \in C } 2^{-|C|}$ %
  \end{center}
\end{definition}

We call a set $C$ of literals \emph{trivial} if there is a variable $x$ such that $x \in C$ 
and $\overline x \in C$. For a variable $x$, say that $\Dc$ is \emph{$x$-orthogonal} if for all $C, C' \in \Dc$, $C \neq C'$, there is a literal $\eta \not\in \{x,\overline x\}$ such that $\eta \in C$ and $\overline \eta \in C'$. Orthogonality is well-studied for DNFs~\cite{HammerC2011}. The two-sided Jeroslow-Wang value and 
the influence agree up to a factor of two for some families of sets of literals when interpreting them as DNFs:

\begin{toappendix}
  The following connection to the two-sided Jeroslow-Wang heuristic is established by counting the number of satisfying assignments. This heavily relies on the fact that every pair of clauses is distinguished by a literal that occurs negatively in one and positively in the other:
\end{toappendix}

\begin{apxtheoremrep}
  \label{thr:influenceJW}
  Let $\Dc$ be a family of sets of literals such that all of its elements are non-trivial, 
  and let $x$ be variable such that $\Dc$ is $x$-orthogonal. Then:
  \begin{center}
    $\Inf_x( \textstyle\bigvee_{C \in \Dc} \bigwedge_{\eta \in C} \eta ) \ = \ 2 \cdot \JW_x(\Dc).$
  \end{center}
\end{apxtheoremrep}
\begin{appendixproof}
  For a set of literals $C$, denote their conjunction as $\pi_{C} = \textstyle\bigwedge_{\eta \in C} \eta$. We can write $f = \textstyle\bigvee_{C \in D} \pi_{C}$ as
  \begin{align*}
    f = h_1 \lor x h_2 \lor \overline x h_3 = x (h_1 \lor h_2) \lor \overline x (h_1 \lor h_3),
  \end{align*}
  where
  \begin{align*}
    h_1 & = \textstyle\bigvee_{C \in \Dc \colon x, \overline x \not\in C}  \pi_{C},                       \\
    h_2 & = \textstyle\bigvee_{C \in \Dc \colon x \in C } \pi_{C \setminus \{ x \}},                     \\
    h_3 & = \textstyle\bigvee_{C \in \Dc \colon \overline x \in C } \pi_{C \setminus \{ \overline x \}}.
  \end{align*}
  Since all $C \in \Dc$ are non-trivial, it must be the case that no element of $\{ C \setminus \{ x \} \colon x \in C \in \Dc \}$, $\{ C \setminus \{ \overline x \} \colon \overline x \in C \in \Dc \}$ or $\{ C \in \Dc : x,\overline x \not\in C \}$ contains $x$ or $\overline x$. Therefore, neither $h_1$ nor $h_2$ nor $h_3$ can depend on $x$. Thus, $h_1 \lor h_2$ and $h_1 \lor h_3$ must indeed be the positive and negative cofactors of $f$ with respect to $x$. Hence,
  \begin{align*}
    \Inf_x(f) = \E[ (h_1 \lor h_2) \oplus (h_1 \lor h_3) ],
  \end{align*}
  which can be written as
  \begin{align*}
    \Inf_x(f) = \E[ \overline{h_1} h_2 \overline{h_3} ] + \E[\overline{h_1}\, \overline{h_2} h_3 ].
  \end{align*}
  The goal is now to show that $h_i \overline{h_j}$ is equal to $h_i$ for $i,j \in \{1,2,3\}$, $i \neq j$. Using the distributivity law, write the conjunction of $h_i$ and $h_j$ as
  \begin{align*}
    h_i h_j = \textstyle\bigvee_{C \in I} \bigvee_{C' \in J} \pi_{C} \pi_{C'},
  \end{align*}
  where $I$ and $J$ are chosen accordingly from $\{ C \in \Dc \colon x,\overline x \not\in C \}$, $\{ C \setminus \{ x \} \colon x \in C \in \Dc \}$ and $\{ C \setminus \{\overline x\} \colon \overline x \in C \in \Dc \} $. Since $\Dc$ is $x$-orthogonal, for every $C \in I$ and $C' \in J$, there is a literal $\eta$ such that $\eta \in C$ and $\overline \eta \in C'$. Thus, every inner term $\pi_{C} \pi_{C'}$ must evaluate to zero. But then $h_i h_j = 0$, and we can see that $h_i \overline{h_j} = h_i$ is implied. This allows the simplification
  \begin{align*}
    \Inf_x(f) = \E[h_3] + \E[h_2].
  \end{align*}
  Let us now compute the expected value of $h_2$, and, by a symmetric argument, that of $h_3$. Again, note that for $C, C' \in \Dc$ with $x \in C$ and $x \in C'$ and $C \neq C'$, it must be the case that $\pi_{C\setminus\{x\}} \pi_{C'\setminus\{x\}} = 0$, thus implying $\pi_{C\setminus\{x\}} \overline{\pi_{C'\setminus\{x\}}} = \pi_{C\setminus\{x\}}$. Therefore, for every $C' \in \Dc$ with $x \in C'$,
  \begin{align*}
    \E[h_2]
     & = \E[ \textstyle\bigvee_{C \in \Dc : x \in C} \pi_{C \setminus \{ x \}} ]                                                                                        \\
     & = \E[ \pi_{C' \setminus \{ x \}} ] + \E[ \textstyle\bigvee_{C \in \Dc \setminus \{ C' \}: x \in C}  \overline{\pi_{C'\setminus \{ x \}}} \pi_{C \setminus \{ x \}} ] \\
     & = \E[ \pi_{C' \setminus \{ x \}} ] + \E[ \textstyle\bigvee_{C \in \Dc \setminus \{ C' \}: x \in C}  \pi_{C \setminus \{ x \}} ]                         \\
     & = \dots = \textstyle\sum_{C \in \Dc : x \in C} \E[ \pi_{C \setminus \{ x \}} ].
  \end{align*}
  For $C \in \Dc$ with $x \in C$, the number of satisfying assignments of $\pi_{C \setminus \{ x \}}$ is simply $2^{n-(|C|-1)}$. Thus, $\E[ \pi_{C\setminus\{x\}}] = 2^{-|C|+1}$. We can now apply a symmetric argument to $h_3$, resolve the above equation of $\Inf_x(f)$ and then apply the fact that all $C \in \Dc$ are non-trivial to get
  \begin{align*}
    \Inf_x(f) = 2 \cdot \textstyle\sum_{C \in \Dc\colon x\in C \text{ or } \overline x \in C} 2^{-|C|}
    = 2 \cdot \JW_x(D).
  \end{align*}
\end{appendixproof}
A simple example that illustrates~\Cref{thr:influenceJW} would 
be $\Dc = \{ \{ x, y, z \}, \{ y, \overline z \} \}$. Note that we can 
interpret $\Dc$ as a CNF as well, since the influence does not distinguish between 
a function and its dual (\Cref{thr:propertyRelationsOfVFs}). 
Note that every Boolean function can be expressed by a family $\Dc$ that satisfies the conditions of~\Cref{thr:influenceJW}. For this, we construct the canonical DNF corresponding to $f$ and 
resolve all monomials that differ only in $x$. 
\ifx\arxiv\undefined
\else
(See \Cref{thr:existenceOfXOrthogonalDNF} in the appendix.)
\fi

\begin{toappendix}
  The following proposition shows that every Boolean function $f$ can be represented by DNF made up of a family of sets of literals $\Dc$ that satisfies the premise in~\Cref{thr:influenceJW}. Say that a family of sets of literals $\Ac$ \emph{is a DNF representation of} $f$ if $f = \bigvee_{C \in \Ac} \bigwedge_{\eta \in C} \eta$. Then:
  \begin{apxproposition}
    \label{thr:existenceOfXOrthogonalDNF}
    Suppose $f$ is a Boolean function and $x$ a variable. Then there exists a family of sets of literals $\Dc$ such that (i) $\Dc$ is a DNF representation of $f$, (ii) all $C \in \Dc$ are non-trivial and (iii) $\Dc$ is $x$-orthogonal.
  \end{apxproposition}
  \begin{proof}
    Let us first define $\Rc$ as the family of sets of literals that corresponds to the canonical DNF of $f$, i.e.
    \begin{align*}
      \Rc = \{ \{ z : \uf(z) = 1 \} \cup \{ \overline z : \uf(z) = 0 \} : \uf \in \{0,1\}^X, f(\uf) = 1 \}.
    \end{align*}
    We now resolve all sets of literals $C$ for which $C \cup \{ x \}\in \Rc$ and $C \cup \{ \overline x \} \in \Rc$ holds:
    \begin{align*}
      \Dc = \{ C \in \Rc : x, \overline x \not\in C \} \cup \{ C \setminus \{x,\overline x\} : C \cup \{ x\}, C \cup \{ \overline x \} \in \Rc \}.
    \end{align*}
    For example, if $f = y \land (x \lor z)$, then $\Rc = \{ \{ y, x, \overline z \}, \{ y, x, z \}, \{ y, \overline x, z \} \}$ and $\Dc = \{ \{ y,x,\overline z \}, \{ y,z \} \}$.

    If $\Ac$ is a DNF representation of $f$ and $C$ a set of literals such that $C \cup \{ x \}, C \cup \{\overline x \}\in \Ac$, then note that both $C \cup \{ x \}$ and $C \cup \{\overline x\}$ can be resolved to $C \setminus \{x,\overline x\}$ without changing the fact that the resulting family is a DNF representation of $f$. More precisely, $(\Ac \setminus \{ C \cup \{ x \}, C \cup \{\overline x\}\}) \cup (C \setminus \{x,\overline x\})$ must also be a DNF representation of $f$. Repeatedly applying this argument for all $C$ such that $C \cup \{x \}, C \cup \{\overline x \} \in \Rc$ shows that $\Dc$ must indeed be a DNF representation of $f$, i.e. satisfy (i). 
    For (ii), note that there is no $C \in \Rc$ that is trivial. Since every set in $\Dc$ is also a subset of a set in $\Rc$, no element of $\Dc$ can be trivial. 
    Finally, for (iii), note that $\Rc$ is orthogonal, i.e. for all $C,C' \in \Rc$, $C \neq C'$ there exists a literal $\eta$ such that $\eta \in C$ and $\overline \eta \in C'$. After resolving all sets that differ \emph{only} in $\eta \in \{x,\overline x\}$, all resulting pairs of sets must differ in at least one literal not in $\{x,\overline x\}$. Thus, $\Dc$ must be $x$-orthogonal. 
  \end{proof}
\end{toappendix}

\subsection{Cooperative Game Mappings}\label{ssec:cgms}
\begin{toappendix}
    \subsection*{Cooperative Game Mappings} 
\end{toappendix}

Attribution schemes analogous to what we call \emph{value functions} were already studied in the 
context of game theory, most often with emphasis on Shapley- and Banzhaf 
values~\cite{shapleyValueNPersonGames1953,banzhafWeightedVotingDoesn1965}.
They are studied w.r.t. \emph{cooperative games}, which are a popular way of modeling collaborative behavior. 
Instead of Boolean assignments, their domains are subsets (coalitions) of $X$. 
Specifically, cooperative games are of the form $v\colon 2^X \rightarrow \R$, in which the value $v(S)$ is associated 
with the payoff that variables (players) in $S$ receive when collaborating. Since more cooperation generally means higher payoffs, they are often assumed to be monotonically increasing w.r.t. set inclusion.
In its unconstrained form, they are essentially pseudo Boolean functions.

We denote by $\cgs$ the set of all cooperative games. 
If $\image(v) \subseteq \{0,1\}$, then we call $v$ \emph{simple}. 
For a cooperative game $v$, we denote by $\marginal_x v$ the cooperative game that computes the 
``derivative'' of $v$ w.r.t. $x$, which is $\marginal_x v(S) = v(S \cup \{ x \}) - v(S \setminus \{ x \})$. %
We compose cooperative games using operations such as $\cdot,+,-,\land,\lor$ etc., where $(v\circ w)(S) = v(S) \circ w(S)$. 
For $\mathbin{\sim} \in \{\geq,\leq, =\}$, we also write $v \sim w$ if $v(S) \sim w(S)$ for all $S \subseteq X$.
The set of variables $v$ depends 
on is defined as $\dep(v) = \{ x \in X: \marginal_x v \neq 0 \}$.

Cooperative game mappings map Boolean functions to cooperative games.
Specific \emph{instances} of such mappings have previously been
investigated by~\cite{hammerEvaluationStrengthRelevance2000,biswasInfluenceSetVariables2022}.
We provide a \emph{general definition} of this concept to 
show how it can be used to construct IVFs.

\begin{definition}[CGM]
    \label{def:coalitionGameMapping}
    A \emph{cooperative game mapping (CGM)} is a function $\tau \colon \boolfuncs \rightarrow \cgs$ with $f\mapsto \tau_f$. 
    We call $\tau$ \emph{importance inducing} if for all $x,y \in X$, permutations $\sigma\colon X\rightarrow X$,
    and $f,g,h \in \boolfuncs$:
    \begin{enumerate}[leftmargin=2.2cm]
        \item[(\namedlabel{prop:CGM_B}{\makeprop{Bound\textsubscript{CG}}})] $0 \leq \marginal_x \tau_f \leq 1$.
        \item[(\namedlabel{prop:CGM_DUM}{\makeprop{Dum\textsubscript{CG}}})] $\marginal_x \tau_f = 0$ if $x\not\in \dep(f)$.
        \item[(\namedlabel{prop:CGM_DIC}{\makeprop{Dic\textsubscript{CG}}})] $\marginal_x \tau_x = \marginal_x \tau_{\overline x}=1$.
        \item[(\namedlabel{prop:CGM_TI}{\makeprop{Type\textsubscript{CG}}})] 
        (i) $\tau_f(S) = \tau_{\sigma f}(\sigma(S))$ and \\
        (ii) $\tau_f(S) = \tau_{\switchf{y}{f}}(S)$ for all $S \subseteq X$.
        \item[(\namedlabel{prop:CGM_SUD}{\makeprop{ModEC\textsubscript{CG}}})] $\marginal_x \tau_f \geq \marginal_x \tau_h$ if \\
        (i) $f$ and $h$ are monotonically modular in $g$,  \\
        (ii) $\cofactor f {\ass g 1} \geq \cofactor h {\ass g 1}$ and $\cofactor h {\ass g 0} \geq \cofactor f {\ass g 0}$ and  \\
        (iii) $x \in \dep(g)$.
    \end{enumerate}
    We call $\tau$ \emph{unbiased} if $\tau_g = \tau_{\overline g}$ for all $g\in\boolfuncs$.
\end{definition}

An example is the \emph{characteristic} CGM $\zeta$ given by $\zeta_f(S) = f(\idf{S})$, where $\idf{S}(x) = 1$ iff $x \in S$.
We study various importance-inducing CGMs in the following sections.
Note that $\zeta$ is not importance inducing: for example, it violates \ref{prop:CGM_B} since $\marginal_x \zeta_f(\varnothing)=-1$ for $f = \overline x$.

The restriction to \emph{importance-inducing} CGMs ensures that compositions with the Banzhaf or 
Shapley value are valid IVFs (\Cref{thr:cgmsInduceIVFs}). These CGMs satisfy 
properties that are related to \Cref{def:importanceValueFunction}: $\tau_f$ should be 
monotone ($0 \leq \marginal_x \tau_f$), irrelevant variables of $f$ are also irrelevant for 
$\tau_f$ (\ref{prop:CGM_DUM}), etc. In an analogous fashion, we can think of properties related to~\Cref{def:optionalProperties}:

\begin{definition}
    \label{def:coalitionGameMappingOptional}
    A CGM $\tau$ is called
    \begin{itemize}%
        \item \emph{chain-rule decomposable}, if for all $f,g\in\boolfuncs$ such that $f$ is modular in $g$ and $x\in \dep(g)$:
        	\begin{center}
				$\marginal_x \tau_f = (\marginal_x \tau_g) (\marginal_g \tau_f),$
            \end{center}
	        where $\marginal_g \tau_f = \marginal_{x_g} \tau_{f[g/x_g]}$ for some $x_g \not\in \dep(f)$.
	    	We call $\tau$ \emph{weakly cain-rule decomposable} if this holds for all cases where $f$ is monotonically modular in $g$.
        \item \emph{derivative dependent}, if for all $f,g\in\boolfuncs$, $x\in X$
			\begin{center}
				$\Adiff_x f \geq \Adiff_x g \implies \marginal_x \tau_f \geq \marginal_x \tau_g.$
            \end{center}
    \end{itemize}
    
\end{definition}

Since (weak) rank-preservation for value functions uses an IVF in its premise, it cannot be stated naturally at 
the level of CGMs. 
Let us now define the following abstraction, which captures Shapley and Banzhaf values:

\begin{definition}%
    Call $\ImpCG \colon X \times \cgs \rightarrow \R, (x,v) \mapsto \ImpCG_x(v)$ 
    a \emph{value function for cooperative games}. Call $\ImpCG$
    an \emph{expectation of contributions} if there are weights $c(0),\ldots,c(n{-}1) \in \R$
    such that for all $v \in \cgs$ and $x\in X$:
    \begin{align*}
        \sum_{S \subseteq X \setminus \{ x \}} \hspace*{-0.75em} 
        c(|S|) = 1 \quad
        \text{and}\quad \ImpCG_x(v) = \hspace*{-0.75em}\sum_{S \subseteq X \setminus \{ x \}} \hspace*{-0.75em} c(|S|) \,{\cdot}\, \marginal_x v(S).
    \end{align*}
\end{definition}
If $\ImpCG$ is an expectation of contributions, then $\ImpCG_x(v)$ is indeed the expected 
value of $\marginal_x v(S)$ in which every $S \subseteq X \setminus \{x \}$ has 
probability $c(|S|)$. The \emph{Banzhaf} and \emph{Shapley values} are defined as the
expectations of contributions with weights:\\
\begin{tabular}{lcl}\\[-1em]
    $\cbanzhaf(k) = \textstyle\frac{1}{2^{n{-}1}}$ $(\Banzhaf)$ \hspace{1em} and \hspace{1em}
    $\cshapley(k) = \textstyle\frac{1}{n} \binom{n{-}1}{k}^{-1}$ $(\Shapley)$. \\[.4em]
\end{tabular}
Observe that there are $\binom{n{-}1}{k}$ sets of size $k \in \{ 0,\dots,n{-}1\}$, so the weights of the Shapley value indeed sum up to one.

If $\tau$ is a CGM, then its composition with $\ImpCG$
yields $(\ImpCG \circ \tau)_x(f) = \ImpCG_x(\tau_f)$, which is a value function for Boolean functions.
Then every composition with an expectation of contributions is an IVF
if the CGM is importance inducing:

\begin{apxlemmarep}
    \label{thr:cgmsInduceIVFs}
    If $\tau$ is an importance-inducing CGM and $\ImpCG$ an expectation of contributions, then $\ImpCG \circ \tau$ is an IVF. If $\tau$ is unbiased/derivative dependent, then so is $\ImpCG \circ \tau$. 
    Finally, if $\tau$ is (weakly) chain-rule decomposable, then so is $\Banzhaf \circ \tau$.
\end{apxlemmarep} 
\begin{appendixproof}
    Let $c(0),\dots,c(n{-}1) \in \R_{\geq 0}$ be $\ImpCG$'s weights. Let us first show that $\ImpCG \circ \tau$ is indeed an importance value function.
    \begin{enumerate}[leftmargin=1.5cm]
        \item[\ref{prop:B}] Suppose $f$ is a Boolean function and $x$ a variable. We want to show $0 \leq (\ImpCG \circ \tau)_x(f) \leq 1$. First, $0 \leq (\ImpCG \circ \tau)_x(f)$ follows from the fact that $\ImpCG \circ \tau$ is a linear combination of non-negative terms due to $c$'s non-negativity and $0 \leq \marginal_x \tau_f$. Since $\marginal_x \tau_f \leq 1$,
        \[ %
            (\ImpCG \circ \tau)_x(f) 
            \leq \textstyle\sum_{S \subseteq X \setminus \{ x \}} c(|S|)   
            = 1.
        \] %

        \item[\ref{prop:DUM}] Suppose $x\not\in\dep(f)$. We want to show $(\ImpCG \circ \tau)_x(f) = 0$. Since in that case $\marginal_x \tau_f(S) = 0$ for all $S \subseteq X$ due to \ref{prop:CGM_DUM}, mind that $(\ImpCG \circ \tau)_x(f)$ must also be zero since it is a linear combination of $\marginal_x \tau_f(S)$ for $S \subseteq X \setminus \{ x \}$.
        \item[\ref{prop:DIC}] Suppose $x$ is a variable. We want to show $(\ImpCG \circ \tau)_x(x) = 1$. $(\ImpCG \circ \tau)_x(\overline x) = 1$ then follows from \ref{prop:TI}. Since $\tau$ satisfies \ref{prop:CGM_DIC}, note that $\marginal_x \tau_x(S) = 1$ for all $S \subseteq X$. Now use the fact that $(\ImpCG \circ \tau)_x(f)$ is a convex combination of $\marginal_x \tau_x(S) = 1$'s, which implies $(\ImpCG \circ \tau)_x(f) = 1$.

        \item[\ref{prop:TI}] Suppose $f$ is a Boolean function, $x,y$ variables and $\sigma$ is a permutation. We want to show that (i) $(\ImpCG \circ \tau)_x(f) = (\ImpCG \circ \tau)_{\sigma(x)}(\sigma f)$ and (ii) $(\ImpCG \circ \tau)_x(f) = (\ImpCG \circ \tau)_x(\switchf{y}{f})$. 
        
        For (i), note that $\marginal_x \tau_f(S) = \marginal_{\sigma(x)} \tau_{\sigma f}(\sigma(S))$ due to \ref{prop:CGM_TI}. Thus, since $|S| = |\sigma(S)|$,
        \begin{align*}
            (\ImpCG \circ \tau)_x(f) 
            &= \textstyle\sum_{S \subseteq X \setminus \{ x \}} c(|\sigma(S)|) \marginal_{\sigma(x)}  \tau_{\sigma f}(\sigma(S)) \\
            &= \textstyle\sum_{S \subseteq X \setminus \{ \sigma(x) \}} c(|S|) \marginal_{\sigma(x)} \tau_{\sigma f}(S) \\
            &= (\ImpCG \circ \tau)_{\sigma(x)}(\sigma f),
        \end{align*}
        as the mapping $S \mapsto \sigma(S)$ constitutes a bijection that takes $2^{X\setminus \{x \}}$ to $2^{X \setminus \{ \sigma(x)\}}$. 
        
        For (ii), note that $\tau_f = \tau_{\switchf{y}{f}}$ implies $\marginal_x \tau_f = \marginal_x \tau_{\switchf{y}{f}}$, and thus $(\ImpCG \circ \tau)_x(f) = (\ImpCG \circ \tau)_x(\switchf{y}{f})$. 

        \item[\ref{prop:SUD}] We can see that the premise, i.e. the requirements that $f$ encapsulates $h$ on $g$ and $x \in \dep(g)$, is exactly the same as in \ref{prop:CGM_SUD}. Therefore, it suffices to pick such $f,g,h$ and $x$ and then apply \ref{prop:CGM_SUD} to obtain $\marginal_x \tau_f \geq \marginal_x \tau_h$. Using non-negativity of $c$, we can see that $(\ImpCG \circ \tau)_x(f) \geq (\ImpCG \circ \tau)_x(h)$ must indeed hold.
    \end{enumerate}
    For the second part, 
    \begin{itemize}
        \item If $\tau$ is unbiased, then so is $\ImpCG \circ \tau$: Suppose $g$ is a Boolean function and $x$ a variable. We can clearly see that $(\ImpCG \circ \tau)_x(g) = (\ImpCG \circ \tau)_x(\overline g)$ must hold due to the fact that $\marginal_x \tau_g = \marginal_x \tau_{\overline g}$ is implied by $\tau_g = \tau_{\overline g}$.
        
        \item If $\tau$ is derivative dependent, then so is $\ImpCG \circ \tau$: Let $f,h$ be Boolean functions such that $\Adiff_x f \geq \Adiff_x h$. Since $\tau$ is derivative dependent, we can see that $\marginal_x \tau_f \geq \marginal_x \tau_h$ must hold. Use $c$'s non-negativity to see that $c(|S|) \marginal_x \tau_f(S) \geq c(|S|) \marginal_x \tau_h(S)$ is the case for all $S \subseteq X \setminus \{x \}$. Application of $\ImpCG$'s definition shows $(\ImpCG \circ \tau)_x(f) \geq (\ImpCG \circ \tau)_x(h)$.
       
        \item If $\tau$ is (weakly) chain-rule decomposable, then so is $\Banzhaf \circ \tau$: First note that if $v$ is a cooperative game and $\dep(v) \subseteq R$, then $\E_{S \subseteq X}[ v(S) ] = \E_{S \subseteq R}[ v(S) ]$. This follows from \Cref{thr:decompositionOfExpectations}, when viewing $2^X$ as $\{0,1\}^X$. Denote this by $(*)$.
        Let us now show the original claim. Suppose that $f$ is (monotonically) modular in $g$ and $x \in \dep(g)$. Then 
        \begin{align*}
            (\Banzhaf \circ \tau)_x(f) 
            &= \E_{S \subseteq X \setminus \{ x \}} [(\marginal_x \tau_f)(S)] 
                && \text{definition} \\
            &= \E_{S \subseteq X} [(\marginal_x \tau_g)(S) \cdot (\marginal_g \tau_f)(S)] 
                && \text{(weak) chain rule decomposability, } (*) \\
            &= \E_{S \subseteq \dep(g)} [\E_{T \subseteq X \setminus \dep(g)}[(\marginal_x \tau_g)(S\cup T) \cdot (\marginal_g \tau_f)(S \cup T)]]
                && \text{uniform distribution} \\
            &= \E_{S \subseteq \dep(g)} [\E_{ T \subseteq X \setminus \dep(g)}[(\marginal_x \tau_g)(S) \cdot (\marginal_g \tau_f)(T)] ]
                && \text{\ref{prop:CGM_DUM}} \\
            &= \E_{S \subseteq \dep(g)}[(\marginal_x \tau_g)(S)] \cdot \E_{T \subseteq X \setminus \dep(g)}[(\marginal_g \tau_f)(T)] 
                && \text{independent random variables} \\
            &= \E_{S \subseteq X}[(\marginal_x \tau_g)(S)] \cdot \E_{T \subseteq X}[(\marginal_g \tau_f)(T)] 
                && (*) \\
            &= \E_{S \subseteq X \setminus \{ x \}}[(\marginal_x \tau_g)(S)] \cdot \E_{T \subseteq X \setminus \{ x_g\}}[(\marginal_{x_g} \tau_{f[g/x_g]})(T)] 
                && (*) \\
            &= (\Banzhaf \circ \tau)_x(g) \cdot (\Banzhaf \circ \tau)_{g}(f).
        \end{align*}
        The fourth equality is true since
        \begin{align*}
          \dep(g) \cap T = \varnothing 
          &\implies
            \dep(\tau_g) \cap T = \varnothing 
            && \dep(\tau_g) \subseteq \dep(g) \text{ by } \ref{prop:CGM_DUM} \\
          & \implies
            \dep(\marginal_x \tau_g) \cap T = \varnothing 
            && \dep(\marginal_x \tau_g) \subseteq \dep(\tau_g) \\
          &\implies 
            \marginal_x \tau_g(R \cup \{ z \}) = \marginal_x \tau_g(R \setminus \{ z \}) 
            \quad\forall R \subseteq X, z \in T \\
          &\implies
            \marginal_x \tau_g(S \cup T) = \marginal_x \tau_g(S).
        \end{align*} 
        The identity $\marginal_g \tau_f(S \cup T) = \marginal_g \tau_f(T)$ applies for a similar reason. Set $\ell = f[g/x_g]$. Then:
        \begin{align*}
          (X\setminus \dep(g)) \cap S = \varnothing 
          &\implies
            \dep(\ell) \cap S = \varnothing 
            && \dep(\ell) \subseteq X \setminus \dep(g) \\
          &\implies
            \dep(\tau_\ell) \cap S = \varnothing 
            && \dep(\tau_\ell) \subseteq \dep(\ell) \text{ by } \ref{prop:CGM_DUM} \\
          & \implies
            \dep(\marginal_{x_g} \tau_\ell) \cap S = \varnothing 
            && \dep(\marginal_{x_g} \tau_\ell) \subseteq \dep(\tau_\ell) \\
          &\implies 
            \marginal_{x_g} \tau_\ell(R \cup \{ z \}) = \marginal_{x_g} \tau_\ell(R \setminus \{ z \}) 
            \quad\forall R \subseteq X, z \in S \\
          &\implies
            \marginal_{x_g} \tau_\ell(S \cup T) = \marginal_{x_g} \tau_{\ell}(T) \\
          &\implies 
            \marginal_{g} \tau_f(S \cup T) = \marginal_{g} \tau_{f}(T).
        \end{align*}
        For the seventh equality, note that we can apply $(*)$ due to $x\not\in \dep(\marginal_x \tau_g)$ and $x_g \not\in \dep(\marginal_{x_g} \tau_{\ell})$.
    \end{itemize}
\end{appendixproof}

In the following sections, we study two novel 
and the already-known CGM of 
\cite{hammerEvaluationStrengthRelevance2000}.
By \Cref{thr:cgmsInduceIVFs} we can focus on their 
properties as CGMs, knowing that any composition with the Shapley value or other expectations of contributions will induce IVFs.

\subsubsection{Simple\,Satisfiability-Biased\,Cooperative\,Game\,Mappings}
\begin{toappendix}
    \subsubsection*{Simple Satisfiability-Biased Cooperative Game Mappings} 
\end{toappendix}

The first CGM interprets the ``power'' of a coalition as its ability to force a function's outcome 
to one:
If there is an assignment for a set of variables that yields outcome one \emph{no matter} the values 
of other variables, we assign this set a value of one, and zero otherwise.

\begin{definition}
  \label{def:dominatingCoalitionGameMapping}
  The \emph{dominating CGM} $\omega$ is defined as
  \begin{center}
    $
    \omega_f(S) = \begin{cases}
      1 & \text{if } \exists \uf \in \{0,1\}^S.\; \forall \wf \in \{0,1\}^{X \setminus S}.\; f(\uf; \wf). \\
      0 & \text{otherwise.}
    \end{cases}
    $
  \end{center}
\end{definition}

\begin{example*}
 Let $f = x \lor (y \oplus z)$. We have $\omega_f(\{y,z\}) = 1$ since $\cofactor f \uf = 1$ for $\uf = \ass y 1;\ass z 0$. On the other hand, $\omega_f(\{y\})=0$, since $\ass x 0; \ass z 1$ resp. $\ass x 0;\ass z 0$ falsify $\cofactor f {\ass y 1}$ and $\cofactor f {\ass y 0}$.
\end{example*}

\begin{toappendix}
  If $\omega$ is the dominating cooperative game mapping (cf. Definition \ref{def:dominatingCoalitionGameMapping}), then we will also sometimes use the following inequalities, which are easy to show: If $s$ and $t$ are Boolean functions such that $s \geq t$, then $\omega_s \geq \omega_t$. Furthermore, one can then see that $\omega_s \lor \omega_t = \omega_s$. The following decompositions are helpful as well:

  \begin{apxlemma}
    \label{thr:dcgmDecompositionVariables}
    Suppose $f$ is a Boolean function. Then for all $S \subseteq X$ and variables $z$,
    \begin{align}
       & \omega_f(S \cup \{ z \}) =  (\omega_{\cofactor f {\ass z 1}} \lor \omega_{\cofactor f {\ass z 0}})(S)
      \label{eq:dcgmDecompositionUnion}                                        \\
       & \omega_f(S \setminus \{ z \}) =  \omega_{\cofactor f {\ass z 1} \land \cofactor f {\ass z 0}}(S)
      \label{eq:dcgmDecompositionDifference}
    \end{align}
  \end{apxlemma}
  \begin{proof}
    For \eqref{eq:dcgmDecompositionUnion}, note that
    \begin{align*}
      \omega_f(S \cup \{ z \})
       & = \exists c \in \{0,1\}.
      \exists \uf \in \{0,1\}^{S \setminus \{ z \}}.
      \forall \vf \in \{0,1\}^{X \setminus (S \cup \{ z \})}. f(\ass z c; \uf; \vf). \\
       & = \exists c \in \{0,1\}.
      \exists \uf \in \{0,1\}^{S}.
      \forall \vf \in \{0,1\}^{X \setminus S}. \cofactor f {\ass z c}(\uf; \vf).               \\
       & = \omega_{\cofactor f {\ass z 1}}(S) \lor \omega_{\cofactor f {\ass z 0}}(S).
    \end{align*}
    And for \eqref{eq:dcgmDecompositionDifference},
    \begin{align*}
      \omega_f(S \setminus \{ z \})
       & =
      \exists \uf \in \{0,1\}^{S \setminus \{ z \}}.
      \forall \vf \in \{0,1\}^{X \setminus (S \cup \{ z \})}.
      \forall c \in \{0,1\}.
      f(\ass z c; \uf; \vf).                                                                  \\
       & =
      \exists \uf \in \{0,1\}^{S}.
      \forall \vf \in \{0,1\}^{X \setminus S}. \cofactor f {\ass z 1}(\uf; \vf) \land \cofactor f {\ass z 0}(\uf;\vf). \\
       & = \omega_{\cofactor f {\ass z 1} \land \cofactor f {\ass z 0}}(S).
    \end{align*}
  \end{proof}

  The following decomposition of $\omega$ might be useful even in other contexts. (For example, for solving problems in $\Sigma^\mathtt{P}_2$.) Note that the decomposition provided in \Cref{thr:dcgmDecomposition} is essentially a case-distinction that distinguishes between all four values that $\omega_g(S)$ and $\omega_{\overline g}(S)$ can take.

  \begin{apxlemma}
    \label{thr:dcgmDecomposition}
    If $f$ is modular in $g$ with positive and negative cofactor $s$ and $t$, then
    \[ %
      \omega_f 
      = \omega_g \omega_{\overline g} (\omega_s \lor \omega_t)
      \lor \omega_g \overline{\omega_{\overline g}} \omega_s
      \lor \overline{\omega_g} \omega_{\overline g} \omega_t
      \lor \overline{\omega_g}\, \overline{\omega_{\overline g}} \omega_{s t}.
    \]  %
    Moreover, if $f$ is monotonically modular in $g$, then this simplifies to
    $ %
      \omega_g \omega_s \lor \overline{\omega_g} \omega_t.
    $ %
  \end{apxlemma}
  \begin{proof}
    Choose $S \subseteq X$ arbitrarily and define $U = S \cap (X \setminus \dep(g))$ and $V = S \cap \dep(g)$. This partitions $S$ into $S = U \cup V$ such that $s$ and $t$ only depend on variables in $U$ and $g$ only depends on variables in $V$. Thus,
    \begin{align*}
      \omega_f(S) = 1 \Longleftrightarrow \text{ there are } \uf \in \{0,1\}^{U}, \vf \in \{0,1\}^{V} \text{ such that } \cofactor f {\uf;\vf} = \cofactor g {\vf} \cofactor s {\uf} \lor \cofactor {\overline g} {\vf} \cofactor t {\uf} = 1.
    \end{align*}
    Furthermore, let us define:
    \begin{align*}
      v = \omega_g \omega_{\overline g} (\omega_s \lor \omega_t)
      \lor \omega_g \overline{\omega_{\overline g}} \omega_s
      \lor \overline{\omega_g} \omega_{\overline g} \omega_t
      \lor \overline{\omega_g}\, \overline{\omega_{\overline g}} \omega_{s t}.
    \end{align*}
    We want to establish $\omega_f = v$ by showing both directions $\omega_f(S) \leq v(S)$ and $\omega_f(S) \geq v(S)$:

    \noindent For ``$\leq$'', assume $\omega_f(S)= 1$, i.e. the existence of $\uf \in \{0,1\}^{U}, \vf \in \{0,1\}^{V}$ such that $\cofactor g \vf \cofactor s \uf \lor \cofactor {\overline g} \vf \cofactor t \uf = 1$. Let us distinguish into the following cases,
    \begin{enumerate}[leftmargin=3.5cm]
      \item[$\cofactor g \vf$ is one everywhere:] Then $\cofactor s \uf$ must be one everywhere. But this implies  $\omega_g \omega_s(S) = 1$ and therefore $v(S) = 1$.
      \item[$\cofactor g \vf$ is zero everywhere:] Then $\cofactor t \uf$ and $\cofactor {\overline g} {\vf}$ must be one everywhere. But this implies $\omega_{\overline g} \omega_t(S) = 1$ and therefore $v(S) = 1$.
      \item[$\cofactor g \vf$ is not constant:] Then there are $\wf_1, \wf_2$ over $\dep(g) \setminus V$ such that $\wf_1$ satisfies and $\wf_2$ falsifies $\cofactor g \vf$. However, due to $\cofactor f {\uf;\vf;\wf_1} = \cofactor f {\uf; \vf; \wf_2} = 1$, note that the first case implies $\cofactor s \uf = 1$ and the second $\cofactor t \uf=1$. Thus, $\omega_{st}(U) = 1$, which also establishes $\omega_s(U) = \omega_t(U) = (\omega_s \lor \omega_t)(U) = 1$ due to $s t \leq s$, $st \leq t$ and $st \leq s \lor t$. This implies $v(S) = 1$ as well.
    \end{enumerate}

    \noindent For ``$\geq$'', assume $v(S) = 1$, and let us distinguish into the following cases,
    \begin{enumerate}[leftmargin=3cm]
      \item[$\omega_g \omega_{\overline g}(V) = 1$:] Then $(\omega_s \lor \omega_t)(U) =1$. If $\omega_s(U) = 1$, then $\omega_{g s}(S)= 1$ and $gs \leq f$ implies $\omega_f(S) = 1$. If $\omega_t(U) = 1$, then $\omega_{\overline g t}(S) = 1$ and $\overline g t \leq f$ implies $\omega_f(S) = 1$.
      \item[$\omega_g \overline{\omega_{\overline g}}(V) = 1$:] Then $\omega_s(U) = 1$. By the same argument as above, note that $\omega_{gs}(S) = 1$, which implies $\omega_f(S) = 1$.
      \item[$\overline{\omega_g} \omega_{\overline g}(V) = 1$:] Symmetric.
      \item[$\overline{\omega_g}\, \overline{\omega_{\overline g}}(V) = 1$:] Then $\omega_{st}(U) = 1$. Note that $st \leq f$, which implies $\omega_f(S)= 1$.
    \end{enumerate}

    \noindent For the simplified formula, suppose that $f$ is monotonically modular in $g$. Mind that $s \geq t$ implies $\omega_s \lor \omega_t = \omega_s$ and $\omega_{st} = \omega_t$. But then $v$ reduces to
    $
      \omega_g \omega_s \lor \overline{\omega_g} \omega_t.
    $
  \end{proof}
\end{toappendix}

\begin{toappendix}
  We prove the following proposition using \Cref{thr:dcgmDecompositionVariables}, by checking that every property is satisfied.
\end{toappendix}

\begin{apxtheoremrep}
  \label{thr:dcgmImportanceInducing}
  The dominating CGM is weakly chain-rule decomposable and importance inducing.
\end{apxtheoremrep}
\begin{appendixproof}
  \hfill
  \begin{enumerate}[leftmargin=2.5cm]
    \item[\ref{prop:CGM_B}] Let $f$ be a Boolean function and $x$ be a variable. We want to show $0 \leq \marginal_x \omega_f \leq 1$. Since $\omega_f(S) \in \{0,1\}$ for all $S \subseteq X$, we can easily see that $\marginal_x \omega_f \leq 1$. Furthermore, note that for all $S \subseteq X$, by \Cref{thr:dcgmDecompositionVariables},
    \begin{align*}
      \omega_f(S \cup \{ x \}) = (\omega_{\cofactor f {\ass x 1}} \lor \omega_{\cofactor f {\ass x 0}})(S) \geq \omega_{\cofactor f {\ass x 1}}(S) \geq \omega_{\cofactor f {\ass x 1} \land \cofactor f {\ass x 0}}(S) = \omega_f(S \setminus \{ x \}),
    \end{align*}
    which implies $0 \leq \marginal_x \omega_f$.
    
    \item[\ref{prop:CGM_DIC}] Let $x$ be a variable. We want to show $\marginal_x \omega_x = 1$. $\marginal_x \omega_{\overline x} = 1$ then follows from \ref{prop:CGM_TI}. Note that $\omega_x(S) = 1$ if and only if $x \in S$ for $S \subseteq X$. Thus, $\marginal_x \omega_x(S) = 1$ for all $S \subseteq X$.
    
    \item[\ref{prop:CGM_DUM}] Let $f$ be a Boolean function and $x \not\in \dep(f)$ be a variable. We want to show $\marginal_x \omega_f = 0$. Since $f = \cofactor f {\ass x 1} = \cofactor f {\ass x 0}$, we simply apply \Cref{thr:dcgmDecompositionVariables} and get $\omega_f(S \cup \{ x \}) = \omega_f(S\setminus \{ x \})$ and thus $\marginal_x \omega_f (S)= 0$ for all $S \subseteq X$.
    
    \item[\ref{prop:CGM_TI}] Suppose $f$ is a Boolean function, $\sigma$ a permutation and $x,y$ variables. Note that for $\omega_f(S) = \omega_{\sigma f}(\sigma (S))$, 
          \begin{align*}
            \omega_{\sigma f}(S)
             & = \exists \uf \in\{0,1\}^S. \forall \vf \in \{0,1\}^{X \setminus S}.
            \sigma f( \uf;\vf )                                                                 \\
             & = \exists \uf \in\{0,1\}^S. \forall \vf \in \{0,1\}^{X \setminus S}.
            f(
            \sigma^{-1}\uf;
            \sigma^{-1}\vf )                                                  \\
             & = \exists \uf \in \{0,1\}^{\sigma^{-1}(S)}. \forall \vf \in \{0,1\}^{X \setminus \sigma^{-1}(S)}.
            f(\uf; \vf)                                                                                          
              && \eqref{eq:perm-prop-2} \\
             & = \omega_{f}(\sigma^{-1}(S))
          \end{align*}
          This is the same as asking for $\omega_f(S) = \omega_{\sigma f}(\sigma(S))$ for all $S \subseteq X$. For $\omega_f(S) = \omega_{\switchf{y}{f}}(S)$, note that we can apply \eqref{eq:perm-prop-2} for $\uf \mapsto \switch{y}{\uf}$ in an analogous way.

    \item[\ref{prop:CGM_SUD}] Suppose that (i) $f$, $h$ are monotonically modular in $g$ such that (ii) $\cofactor f {\ass g 1} \geq \cofactor h {\ass g 1}$ and $\cofactor h {\ass g 0} \geq \cofactor f {\ass g 0}$ and (iii) $x \in \dep(g)$. We want to show $\marginal_x \omega_f \geq \marginal_x \omega_h$. By \Cref{thr:dcgmDecomposition},
          \begin{align*}
            \omega_f
             & = \omega_g \omega_{\cofactor f {\ass g 1}} + (1-\omega_g) \omega_{\cofactor f {\ass g 0}}
          \end{align*}
          Therefore,
          \begin{align*}
            \marginal_x \omega_f
             & = (\marginal_x \omega_g) (\omega_{\cofactor f {\ass g 1}}
            - \omega_{\cofactor f {\ass g 0}}),                          \\
            \marginal_x \omega_h
             & = (\marginal_x \omega_g) (\omega_{\cofactor h {\ass g 1}}
            - \omega_{\cofactor h {\ass g 0}}).
          \end{align*}
          By using $\omega_{\cofactor f {\ass g 1}} \geq \omega_{\cofactor h {\ass g 1}}$ and $\omega_{\cofactor h {\ass g 0}} \geq \omega_{\cofactor f {\ass g 0}}$ and $\omega_{\cofactor f {\ass g 1}} \geq \omega_{\cofactor f {\ass g 0}}$ and $\omega_{\cofactor h {\ass g 1}} \geq \omega_{\cofactor h {\ass g 0}}$ (due to (i) and (ii)) and $\marginal_x \omega_g \geq 0$, we obtain $\marginal_x \omega_f \geq \marginal_x \omega_h$.

  \end{enumerate}
  Finally, to see that $\omega$ is weakly chain rule decomposable, suppose that $f$ is monotonically modular in $g$ and $x \in \dep(g)$. The same decomposition as in the proof for \ref{prop:CGM_SUD} applies, i.e.
  \begin{align*}
    \marginal_x \omega_f
      & = (\marginal_x \omega_g) (\omega_{\cofactor f {\ass g 1}}
    - \omega_{\cofactor f {\ass g 0}}).
  \end{align*}
  Denote $\ell = f[g/x_g]$ and note that for all $S \subseteq X$, since $\cofactor \ell {x_\ass g 1} = \cofactor f {\ass g 1} \geq \cofactor f {\ass g 0} = \cofactor \ell {x_\ass g 0}$ and due to \Cref{thr:dcgmDecompositionVariables},
  \begin{align*}
      & \omega_{\cofactor f {\ass g 1}}(S) = (\omega_{\cofactor \ell {x_\ass g 1}} \lor \omega_{\cofactor \ell {x_\ass g 0}})(S) = \omega_{\ell}(S \cup \{ x_g \}) \\
      & \omega_{\cofactor f {\ass g 0}}(S) = \omega_{\cofactor \ell {x_\ass g 0} \land \cofactor \ell {x_\ass g 1}}(S) = \omega_{\ell}(S \setminus \{ x_g \}),
  \end{align*}
  which suffices to show $\omega_{\cofactor f {\ass g 1}} - \omega_{\cofactor f {\ass g 0}} = \marginal_{x_g} \omega_\ell = \marginal_g \omega_f$. Thus, $\marginal_x \omega_f = (\marginal_x \omega_g) (\marginal_g \omega_f)$.
\end{appendixproof}

\begin{example*}
Let $\ZeroVF$ be the expectation of contributions with $c(0)=1$, i.e.,
    $ \ZeroVF_x(v) = v(\{x\}) - v(\varnothing).$
  By \Cref{thr:cgmsInduceIVFs} and \Cref{thr:dcgmImportanceInducing}, the mapping 
  \begin{center}
    $(\ZeroVF \circ \omega)_x(f) = \begin{cases} 
      1 & \text{if } f \neq 1 \text{ and } \cofactor f {\ass x 0} = 1 \text{ or } \cofactor f {\ass x 1} = 1 \\
      0 & \text{otherwise}
    \end{cases}$
  \end{center}
	is an IVF.
  Intuitively, $x$ has the highest importance if the function is 
  falsifiable and there is a setting for $x$ that forces the function to one.
  Otherwise, $x$ has an importance of zero.
\end{example*}

\paragraph{Biasedness and rank preservation.}
The dominating CGM is biased:
Consider $g = x \lor (y \oplus z)$ with $\overline g = \overline x \land (\overline y \oplus z)$. 
Note that $\omega_g(S) = 1$ for $S = \{ x \}$ while $\omega_{\overline g}(S) = 0$, which shows biasedness.
Composing $\omega$ with the Banzhaf value yields
  \begin{center}
    $(\Banzhaf \circ \omega)_{(\cdot)}(g) : \quad z : 0.25 \; = \; y : 0.25 \; < \; x : 0.75$, \\
   $(\Banzhaf \circ \omega)_{(\cdot)}(\overline g) : \quad z : 0.25 \; =  \;y : 0.25 \; =  \;x : 0.25$,
  \end{center}
One can force $g$ to one by controlling either $x$ or \emph{both} $y$ and $z$,
so $x$ is rated higher than the others.
But to force $\overline g$ to one, control over all variables is required, 
so all variables in $\overline g$ have the same importance.

Since $g$ is modular in $\overline g$, we also obtain a counterexample for rank preservation:
\begin{align*}
& (\Banzhaf \circ \omega)_y(\overline g) \geq (\Banzhaf \circ \omega)_x(\overline g) \\
\text{does not imply} \quad &(\Banzhaf \circ \omega)_y(g) \geq (\Banzhaf \circ \omega)_x(g).
\end{align*}
However, \emph{weak} rank preservation is fulfilled by $\Banzhaf \circ \omega$ 
since it is weakly chain-rule decomposable by \Cref{thr:dcgmImportanceInducing} and \Cref{thr:cgmsInduceIVFs}. Then the claim follows with \Cref{thr:propertyRelationsOfVFs}.

\paragraph{A dual to the dominating CGM.}
One can think of a dual notion of the CGM $\omega$ that reverses the order of both quantifiers. 
Intuitively, we are now allowed to 
choose an assignment \emph{depending} on the values of the remaining variables:

\begin{definition}
  The \emph{rectifying CGM} $\nu$ is defined as
  \begin{center}
    $
    \nu_f(S) = \begin{cases}
      1 & \text{if } \forall \wf \in \{0,1\}^{X \setminus S}.\; \exists \uf \in \{0,1\}^S.\; f(\uf; \wf). \\
      0 & \text{otherwise.}
    \end{cases}
    $
  \end{center}
\end{definition} 

If we compose $\nu$ with an expectation of contributions that
satisfies $c(k) = c(n{-}1{-}k)$ for all $k \in \{0,\dots,n{-}1\}$,
which is a condition satisfied both by the Shapley and Banzhaf values,
the induced importance of a variable equals its importance w.r.t. $\omega$ and the negated function:

\begin{toappendix}
  The proof of \Cref{thr:dcgmFunctionVsNegation} uses the fact that $\omega$ can be expressed in terms of $\nu$. The additional condition $c(k) = c(n{-}1{-}k)$ then allows us to rearrange the terms of the sum of the expectation of contributions:
\end{toappendix}

\begin{apxpropositionrep}
  \label{thr:dcgmFunctionVsNegation}
  Let $\ImpCG$ be an expectation of contributions with $c(k) = c(n{-}1{-}k)$ for all $k \in \{0,\dots,n{-}1\}$. 
  Then for all $g\in\boolfuncs$ and $x\in X$:
  \begin{center}
    $(\ImpCG \circ \omega)_x(g) = (\ImpCG \circ \nu)_x(\overline g)$
  \end{center}
\end{apxpropositionrep}
\begin{appendixproof}
  Note that $\omega_g(S) = 1 - \nu_{\overline g}(X \setminus S)$. Thus, $\marginal_x \omega_g(S) = \marginal_x \nu_{\overline g}(X \setminus S)$ and $\marginal_x \nu_{\overline g}(X \setminus S) = \marginal_x \nu_{\overline g}((X \setminus \{ x \}) \setminus S)$ since $\marginal_x \nu_{\overline g}$ does not depend on $x$. Therefore,
  \begin{align*}
    (\ImpCG \circ \omega)_x(g)
     & = \sum_{S \subseteq X \setminus \{ x \}} c(|S|) \marginal_x \omega_g(S)                                                                           \\
     & = \sum_{S \subseteq X \setminus \{ x \}}  c(|S|) \marginal_x \nu_{\overline g}(X \setminus S)                                                     \\
     & = \sum_{S \subseteq X \setminus \{ x \}}  c(|(X \setminus \{ x \}) \setminus S|) \marginal_x \nu_{\overline g}((X \setminus \{ x \}) \setminus S) \\
     & = \sum_{S \subseteq X \setminus \{ x \}}  c(|S|) \marginal_x \nu_{\overline g}(S)                                                                 \\
     & = (\ImpCG \circ \nu)_x(\overline g),
  \end{align*}
  due to the fact that $n{-}1-|S| = |(X \setminus \{ x \}) \setminus S|$ for $x \not\in S$ and that the 
  mapping $S \mapsto (X \setminus \{ x \}) \setminus S$ constitutes a bijection from $2^{X \setminus \{ x \}}$ to itself. (Note that $(X \setminus \{ x \}) \setminus S$ is just the complement of $S$ w.r.t. $X \setminus \{ x \}$.)
\end{appendixproof}

We now discuss %
connections to the influence. If a Boolean function is monotone, 
and we ``control'' a set of variables $S$, the best towards satisfaction (resp. falsification) 
is to set all variables in $S$ to one (resp. to zero).
This can be used to show that both 
$\Banzhaf \circ \omega$ and $\Banzhaf \circ \nu$ agree with the influence:
\begin{apxpropositionrep}
  \label{thr:dcgmAndInfluence}
  Let $f$ be a monotone Boolean function and $x$ a variable. Then
  $ %
    (\Banzhaf \circ \omega)_x(f) = (\Banzhaf \circ \nu)_x(f) = \Inf_x(f).
  $ %
\end{apxpropositionrep}
\begin{appendixproof}
  Recall that if $f$ is monotone in $z$, then by \Cref{thr:dcgmDecompositionVariables},
  \begin{align*}
     & \omega_f(S) = (\omega_{\cofactor f {\ass z 1}} \lor \omega_{\cofactor f {\ass z 0}})(S) = \omega_{\cofactor f {\ass z 1}}(S) = \omega_{\cofactor f {\ass z 1}}(S \setminus \{ z \}) && \text{if $z \in S$} \\
     & \omega_f(S) = \omega_{\cofactor f {\ass z 1} \land \cofactor f {\ass z 0}}(S) = \omega_{\cofactor f {\ass z 0}}(S) && \text{if $z \not\in S$}
  \end{align*}
  due to $z \not\in \dep(\cofactor f {\ass z 1})$. Since $f$ is monotone in all of its variables, we repeatedly apply this argument for all variables except for $x$ to get 
  \begin{align*}
    (\Banzhaf \circ \omega)_x(f)
     & = \frac{1}{2^{n{-}1}} \sum_{S \subseteq X \setminus \{ x \}} \omega_{f}(S \cup \{ x \}) - \omega_{f}(S)     \\
     & = \frac{1}{2^{n{-}1}} \sum_{S \subseteq X \setminus \{ x \}} \omega_{\cofactor f {\idf{S}}}(\{ x \}) - \omega_{\cofactor f {\idf{S}}}(\varnothing)     \\
     & = \frac{1}{2^{n{-}1}} \sum_{\uf \in \{0,1\}^{X \setminus \{ x \}}} \omega_{\cofactor f \uf}( \{ x \}) - \omega_{\cofactor f \uf}(\varnothing).
  \end{align*}
  By monotonicity of $f$, we have three options for $\cofactor f \uf$ if $\uf$ is an assignment over $X \setminus \{ x \}$: either $\cofactor f \uf = x$, $\cofactor f \uf = 0$ or $\cofactor f \uf = 1$. In the first case, $\omega_{\cofactor f \uf}(\{ x \}) = 1$ and $\omega_{\cofactor f \uf}(\varnothing) = 0$. In the last two, $\omega_{\cofactor f \uf}(\{ x \}) = \omega_{\cofactor f \uf}(\varnothing) = 0$. Therefore, $\omega_{\cofactor f \uf}(\{ x \}) - \omega_{\cofactor f \uf}(\varnothing) = \Adiff_x f(\uf)$, and we get $(\Banzhaf \circ \omega)_x(f) = \Inf_x(f)$. Finally, since (i) the influence does not distinguish between a function and its dual by \Cref{thr:propertyRelationsOfVFs}, (ii) the dual of a monotone boolean function is itself monotone, (iii) \ref{prop:TI} and (iv) \Cref{thr:dcgmFunctionVsNegation},
  \[ 
    \Inf_x(f) 
    \stackrel{\textrm{(i)}}{=} \Inf_x(\switchfset{X}{\overline f}{y}) 
    \stackrel{\textrm{(ii)}}{=} (\Banzhaf \circ \omega)_x(\switchfset{X}{\overline f}{y}) 
    \stackrel{\textrm{(iii)}}{=} (\Banzhaf \circ \omega)_x(\overline f) 
    \stackrel{\textrm{(iv)}}{=} (\Banzhaf \circ \nu)_x(f).
  \]
\end{appendixproof}

\subsubsection{A Constancy-Based Cooperative Game Mapping}
\begin{toappendix}
    \subsubsection*{A Constancy-Based Cooperative Game Mapping} 
\end{toappendix}

\newcommand{\timeout}{\textcolor{black!20}{timeout}}

\begin{table*}[h]
  \centering
    \resizebox{0.9\textwidth}{!}{%
    \begin{tabular}{crr rrr rrrr}
      \toprule
      & & & \multicolumn{3}{c}{(projected) model counting approaches} & \multicolumn{3}{c}{BDD-based approaches} \\ \cmidrule(r){4-6} \cmidrule(r){7-10}
        Instance & \#Variables & \#Clauses & Influence (CNF) & Influence (formula) & Blame & Construction & Influence & DCGM & Blame \\
        \midrule 
b02 & 26 & 66 & 5\,ms & 49\,ms & \timeout & 1\,ms & $<$1\,ms & 2\,ms & 3'649\,ms \\
b06 & 44 & 122 & 7\,ms & 99\,ms & \timeout & 3\,ms & $<$1\,ms & 6\,ms & 697'573\,ms \\
b01 & 45 & 120 & 7\,ms & 110\,ms & \timeout & 4\,ms & $<$1\,ms & 8\,ms & 3'068'667\,ms \\
b03 & 156 & 376 & 11\,ms & 442\,ms & \timeout & 53'934\,ms & 24\,ms & 1'776\,ms & \timeout \\
b13 & 352 & 847 & 34\,ms & 1'088\,ms & \timeout & \timeout & \timeout & \timeout & \timeout \\
b12 & 1'072 & 2'911 & 230\,ms & 8'555\,ms & \timeout & \timeout & \timeout & \timeout & \timeout \\
        \bottomrule
      \end{tabular}}
  \vspace{-.5em}
  \caption{\label{tab:iscas99-main}Computation time for instances of the ISCAS'99 dataset, timeout set to one hour. 
  BDD columns \emph{Influence}, \emph{DCGM} (construction of the BDD for the dominating CGM), and \emph{Blame} are \emph{without} 
  the BDD construction time for the initial CNF (cf. column \emph{Construction}).}
  \vspace{-.3em}
\end{table*}

Hammer, Kogan and Rothblum \cite{hammerEvaluationStrengthRelevance2000} (HKR) 
defined a CGM that measures the power of variables %
by how constant they make a function if assigned random values. 
It depends on the following notion of \emph{constancy measure}:

\begin{definition}\label{def:constancy}%
  We call a mapping $\kappa\colon [0,1] \rightarrow [0,1]$ a \emph{constancy measure} if 
  (i) $\kappa$ is convex, 
  (ii) $\kappa(0) = 1$, 
  (iii) $\kappa(x) = \kappa(1{-}x)$, and 
  (iv) $\kappa(\nicefrac{1}{2}) = 0$.
\end{definition}

\noindent The following functions are instances of constancy measures: 
  \begin{itemize}[noitemsep]
    \item $\kappaquad(a) = 4(a-\nicefrac{1}{2})^2$,
    \item $\kappalog(a) = 1 + a \mathrm{lb}(a) + (1{-}a) \mathrm{lb} (1{-}a)$ with $0 \mathrm{lb}(0) = 0$,
    \item $\kappaabs(a) = 2 | a-\nicefrac{1}{2} |$.
  \end{itemize}
For a constancy measure $\kappa$ and a Boolean function $f$, the \emph{$\kappa$-constancy of $f$} is the value $\kappa(\E[f])$, which measures how balanced the share of ones and zeros is.
It is close to one if $f$ is very unbalanced and close to zero if the share of zeros and ones in $f$ is (almost) the same. 
The power of a set of variables $S$ is now measured in terms of the expected $\kappa$-constancy of $f$ if variables in $S$ are fixed to random values:

\begin{definition}[\cite{hammerEvaluationStrengthRelevance2000}]
  Given a constancy measure $\kappa$, we define the CGM $\HKR^\kappa$ by
  \begin{center}
    $\HKR^\kappa_f(S) = \E_{\af \in \{0,1\}^S}[ \kappa( \E[\cofactor f \af] ) ].$
  \end{center}
\end{definition}

\begin{example*}
  Let $f = x \lor y\lor z$ and $S = \{ x \}$. We obtain 
  $
    \HKR^\kappa_f(S) = 
    \nicefrac 1 2 \cdot \kappa(\nicefrac 3 4) + \nicefrac 1 2 \cdot \kappa(1),
  $
  since
  \begin{center}
    $\E[\cofactor f {\ass x 0}] = \nicefrac 3 4 \quad\text{and}\quad \E[\cofactor f {\ass x 1}] = 1$.
  \end{center}
  Setting $x$ to zero does not determine $f$ completely, while setting it to one also sets $f$ to one, i.e., makes it constant.
  The measure then gives a lower value to the less-constant cofactor, a higher value to the more-constant cofactor and computes the average.
  For this example and $\kappa = \kappaabs$, we obtain $\HKR^\kappa_f(S) = \nicefrac 3 4$ due to
    $\kappa(\nicefrac 3 4) = \nicefrac 1 2$
    and $\kappa(1) = 1.$
\end{example*}

\Cref{thr:propertiesOfHKR} shows that $\HKR^\kappaquad$ is a chain-rule 
decomposable and importance-inducing CGM.
It is open whether other constancy measures 
are importance inducing too.

\begin{toappendix}
  The first four properties of importance inducing CGMs hold for Hammer, Kogan and Rothblum's CGM under all constancy measures $\kappa$, since we only have to rely on general properties of constancy measures. The ``trick'' for $\kappa = \kappaquad$ is that the quantity $\frac{1}{2} \kappa(a) + \frac{1}{2}\kappa(b) - \kappa(\frac{1}{2}a+\frac{1}{2}b)$ can be expressed as $(a-b)^2$, which in turn implies properties such as \ref{prop:CGM_SUD} for $\HKR^\kappa$:
\end{toappendix}

\begin{apxtheoremrep}
  \label{thr:propertiesOfHKR}
  Suppose $\kappa$ is a constancy measure. Then $\HKR^\kappa$ is an unbiased CGM that satisfies \ref{prop:CGM_B}, \ref{prop:CGM_DIC}, \ref{prop:CGM_DUM}, and \ref{prop:CGM_TI}. Further, $\HKR^\kappaquad$ is chain-rule decomposable and satisfies \ref{prop:CGM_SUD}.
\end{apxtheoremrep}
\begin{appendixproof}
  \hfill
  \begin{enumerate}[leftmargin=2.3cm]
    \item[Unbiasedness] Suppose that $g$ is a Boolean function and $x$ a variable. We want to show $\HKR^\kappa_g = \HKR^\kappa_{\overline g}$. Note that for all $S \subseteq X$ and assignments $\af$ over $S$,
          \begin{align*}
            \kappa(\E[\cofactor g \af]) = \kappa(\E[1-\cofactor {\overline{g}} \af]) = \kappa(1 - \E[\cofactor {\overline{g}} \af]) = \kappa(\E[\cofactor {\overline{g}} \af])
          \end{align*}
          due to $\kappa(a) = \kappa(1-a)$. But then it is easy to see that $\HKR^\kappa_g(S) = \HKR^\kappa_{\overline g}(S)$ for all $S \subseteq X$.

    \item[\ref{prop:CGM_B}] Suppose $f$ is a Boolean function and $x$ a variable. We want to show $0 \leq \marginal_x \HKR^\kappa_f \leq 1$. For $S \subseteq X$, note that $\HKR^\kappa_f(S)$ is a convex combination of terms limited by $0$ and $1$, so $\marginal_x \HKR^\kappa_f(S) \leq 1$. The requirement $\marginal_x \HKR^\kappa_f \geq 0$ states that $\HKR^\kappa_f$ is a monotone cooperative game, which was already established by \cite{hammerEvaluationStrengthRelevance2000}.

    \item[\ref{prop:CGM_DIC}] Suppose that $x$ is a variable. We want to show $\marginal_x \HKR^\kappa_x = 1$. $\marginal_x \HKR^\kappa_{\overline x} = 1$ then follows from \ref{prop:CGM_TI}. Note that if $S \subseteq X$ and $\af$ is an assignment over $S$ and $f=x$, it is the case that $\E[\cofactor f \af] = \frac{1}{2}$ if $x \not\in S$ and $\E[\cofactor f \af] \in \{0,1\}$ if $x \in S$. Thus, due to $\kappa(\frac{1}{2}) = 0$ and $\kappa(0) = \kappa(1) = 1$, we have $\marginal_{x} \HKR^\kappa_x(S) = 1$.

    \item[\ref{prop:CGM_DUM}] Suppose that $f$ is a Boolean function and $x \not\in \dep(f)$ a variable. Then for all $S \subseteq X$,
          \begin{align*}
            \marginal_x \HKR^\kappa_f(S)
             & = \E_{\af \in \{0,1\}^{S \setminus \{ x \}}} [
            \textstyle\frac{1}{2} \kappa( \E[\cofactor f {\ass x 1;\af}]) + \frac{1}{2} \kappa(\E[\cofactor f {\ass x 0; \af}]) - \kappa(\E[\cofactor f \af])
            ],
          \end{align*}
          where $\E[\cofactor f {\ass x 1; \af}] = \E[\cofactor f {\ass x 0;\af}] = \E[\cofactor f \af] $ due to $\cofactor f {\ass x 1} = \cofactor f {\ass x 0}$. But then $\marginal_x \HKR^\kappa_f(S) = 0$.

    \item[\ref{prop:CGM_TI}] Suppose that $f$ is a Boolean function, $\sigma$ a permutation and $x,y$ variables. We want to show (i) $\HKR^\kappa_f(S) = \HKR^\kappa_{\sigma f}(\sigma(S))$ and (ii) $\HKR^\kappa_f(S) = \HKR^\kappa_{\switchf{y}{f}}(S)$ for all $S \subseteq X$.

          For (i), take an arbitrary $S \subseteq X$ and assignment $\af$ over $S$. Then
          \begin{align*}
            \E_{\uf \in \{0,1\}^X}[ \cofactor {(\sigma f)} {\sigma \af}(\uf) ]
             & = \E_{\uf \in \{0,1\}^X}[ \cofactor f \af(\sigma^{-1} \uf) ]
             &                                       & \eqref{eq:perm-prop-4}                              \\
             & = \E_{\uf \in \{0,1\}^X}[ \cofactor f \af(\uf) ],
             &                                       & \text{uniform distribution, \eqref{eq:perm-prop-2}}
          \end{align*}
          Thus,
          \begin{align*}
            \HKR^\kappa_f(S)
             & = \E_{\af \in \{0,1\}^S}[ \kappa( \E[ \cofactor f \af ] ) ]                                                                      \\
             & = \E_{\af \in \{0,1\}^S}[ \kappa( \E[ \cofactor {(\sigma f)} {\sigma \af} ] )]                                                           \\
             & = \E_{\af \in \{0,1\}^{\sigma(S)}}[ \kappa( \E[ \cofactor{(\sigma f)} {\af} ] ) ]
             &                                                                        & \text{uniform distribution}, \eqref{eq:perm-prop-2} \\
             & = \HKR^\kappa_{\sigma f}(\sigma(S)).
          \end{align*}
          For (ii), take again $S \subseteq X$. Then, if $y \in S$,
          \begin{align*}
            \HKR^\kappa_f(S) 
            &= \textstyle\frac{1}{2} \HKR^\kappa_{\cofactor f {\ass y 1}}(S) + \textstyle\frac{1}{2} \HKR^\kappa_{\cofactor f {\ass y 0}}(S) \\
            &= \textstyle\frac{1}{2} \HKR^\kappa_{\cofactor {(\switchf{y}{f})} {\ass y 0}}(S) + \textstyle\frac{1}{2} \HKR^\kappa_{\cofactor {(\switchf{y}{f})} {\ass y 1}}(S)  \\
            &= \HKR^\kappa_{\switchf{y}{f}}(S).
          \end{align*}
          And if $y \not\in S$, note that for all assignments $\af$ over $S$, 
          \begin{align*}
            \E[\cofactor f \af] 
            &= \textstyle\frac{1}{2} \E[\cofactor f {\ass y 0; \af}] + \textstyle\frac{1}{2} \E[ \cofactor f {\ass y 1; \af} ] \\
            &= \textstyle\frac{1}{2} \E[\cofactor {(\switchf{y}{f})} {\ass y 1; \af}] + \textstyle\frac{1}{2} \E[ \cofactor {(\switchf{y}{f})} {\ass y 0; \af} ] \\
            &= \E[\cofactor {(\switchf{y}{f})} \af].
          \end{align*}
          This however implies 
          \begin{align*}
            \HKR^\kappa_f(S) 
            &= \E_{\af \in \{0,1\}^S}[ \kappa( \E[\cofactor f \af] ) ] \\
            &= \E_{\af \in \{0,1\}^S}[ \kappa( \E[\cofactor {(\switchf{y}{f})} \af] ) ] \\
            &= \HKR^\kappa_{\switchf{y}{f}}(S).
          \end{align*}
  \end{enumerate}

  \noindent Let us now focus on the case where $\kappa = \kappaquad$, i.e. $\kappa(a) = 4(a-\frac{1}{2})^2$. We can then see that
  \begin{align*}
    \textstyle\frac{1}{2}\kappa(a) + \frac{1}{2} \kappa(b) - \kappa(\frac{1}{2} a + \frac{1}{2} b)
    = (a-b)^2.
  \end{align*}
  Thus,
  \begin{align*}
    \marginal_x \HKR^\kappa_f(S)
     & = \E_{\af \in \{0,1\}^{S \cup \{ x \}}}[ \kappa(\E[\cofactor f \af]) ]
    - \E_{\af \in \{0,1\}^{S \setminus \{ x \}}}[ \kappa(\E[\cofactor f \af])]
     &                                                                                         & \text{definition}                                                                                             \\
     & = \E_{\af \in \{0,1\}^{S \setminus \{ x \}}}[
    \textstyle\frac{1}{2}\kappa(\E[\cofactor f {\af; \ass x 1}]) +
    \textstyle\frac{1}{2}\kappa(\E[\cofactor f {\af; \ass x 0}])
    - \kappa(\E[\cofactor f \af])]
     &                                                                                         & \text{unwrapping $\E$ with respect to $x$}                                                                    \\
     & = \E_{\af \in \{0,1\}^{S \setminus \{ x \}}}[ \E[ \cofactor {(\cofactor f {\ass x 1} - \cofactor f {\ass x 0})} {\af} ]^2 ]
     &                                                                                         & \text{above fact and } \E[\cofactor f \af] = \textstyle\frac{1}{2}\E[\cofactor f {\af;\ass x 1}] + \frac{1}{2}\E[\cofactor f {\af;\ass x 0}] \\
     & = \E_{\af \in \{0,1\}^{S}}[ \E[ \cofactor {(\cofactor f {\ass x 1} - \cofactor f {\ass x 0})} {\af} ]^2 ].
     &                                                                                         & \text{independence of $x$}
  \end{align*}
  This allows us to derive the remaining properties, i.e. the fact that $\HKR^\kappa$ is chain-rule decomposable and that it satisfies \ref{prop:CGM_SUD}:
  \begin{enumerate}[leftmargin=2.3cm]
    \item[c.r.d.] Suppose that $f$ is modular in $g$ and $x \in \dep(g)$. We want to show $\marginal_x \HKR^\kappa_f = (\marginal_x \HKR^\kappa_g) (\marginal_g \HKR^\kappa_f)$. First note that for $S \subseteq X$ and $\af; \bfat \in \{0,1\}^S$ with $A= S \cap \dep(g)$ and $B = S \cap (X \setminus \dep(g))$, due to \Cref{thr:decompositionOfExpectations},
          \begin{align*}
            \E[ \cofactor f {\af; \bfat} ] = \E[ \cofactor g \af ] \E[ (\cofactor f {\ass g 1})_{\bfat} ] + (1 - \E[ \cofactor g \af ] ) \E[ (\cofactor f {\ass g 0})_{\bfat} ]
          \end{align*}
          and therefore
          \begin{align*}
            \E[ \cofactor {(\cofactor f {\ass x 1} - \cofactor f {\ass x 0})} {\af; \bfat} ]
            & = \E[ \cofactor {\left( \cofactor g {\ass x 1} \cofactor f {\ass g 1} + (1-\cofactor g {\ass x 1}) \cofactor f {\ass g 0} - (\cofactor g {\ass x 0} \cofactor f {\ass g 1} + (1-\cofactor g {\ass x 0}) \cofactor f {\ass g 0}) \right)} {\af;\bfat}  ] 
              && \text{$f$ modular in $g$} \\
            & = \E[ \cofactor {((\cofactor g {\ass x 1} - \cofactor g {\ass x 0}) (\cofactor f {\ass g 1} - \cofactor f {\ass g 0}))} {\af;\bfat} ]
              && \text{simplifying} \\
            & = \E[ \cofactor {(\cofactor g {\ass x 1} - \cofactor g {\ass x 0})} {\af} \cofactor {(\cofactor f {\ass g 1} - \cofactor f {\ass g 0})} {\bfat} ] 
              && \text{independence} \\
            & = \E[ \cofactor{(\cofactor g {\ass x 1} - \cofactor g {\ass x 0})} {\af} ] \cdot \E[ \cofactor {(\cofactor f {\ass g 1} - \cofactor f {\ass g 0})} {\bfat}],
              && \text{\Cref{thr:decompositionOfExpectations}}
          \end{align*}
          which implies, using $\ell = f[g/x_g]$,
          \begin{align*}
            \marginal_x \HKR^\kappa_f(S)
             & = \E_{\af \in \{ 0,1\}^{S\cap \dep(g)}, \bfat \in \{0,1\}^{S \cap (X \setminus \dep(g)} }[\E[ \cofactor {(\cofactor g {\ass x 1} - \cofactor g {\ass x 0})} {\af} ]^2 \cdot \E[ \cofactor {(\cofactor f {\ass g 1} - \cofactor f {\ass g 0})} {\bfat}]^2]
             &                                                                                                                                                                                  & \text{uniform distribution}                              \\
             & = \E_{\af \in \{ 0,1\}^{S\cap \dep(g)}}[\E[ \cofactor {(\cofactor g {\ass x 1} - \cofactor g {\ass x 0})} {\af} ]^2] \cdot \E_{\bfat \in \{0,1\}^{S \cap (X \setminus \dep(g))}} [\E[\cofactor {(\cofactor f {\ass g 1} - \cofactor f {\ass g 0})} {\bfat}]^2]
             &                                                                                                                                                                                  & \text{independence}                                      \\
             & = \E_{\af \in \{ 0,1\}^{S}}[\E[ \cofactor {(\cofactor g {\ass x 1} - \cofactor g {\ass x 0})}{\af} ]^2] \cdot \E_{\bfat \in \{0,1\}^{S}} [\E[ \cofactor{(\ell_{x_g/1} - \ell_{x_g/0})}{\bfat}]^2]
             &                                                                                                                                                                                  & \text{\Cref{thr:decompositionOfExpectations}} \\
             & = \marginal_x \HKR^\kappa_g(S) \cdot \marginal_{x_g} \HKR^\kappa_{\ell}(S)                                                                                                                                                                        \\
             & = \marginal_x \HKR^\kappa_g(S) \cdot \marginal_{g} \HKR^\kappa_{f}(S).
          \end{align*}

    \item[\ref{prop:CGM_SUD}] Suppose that (i) $f,h$ are monotonically modular in $g$, (ii) $\cofactor f {\ass g 1} \geq \cofactor h {\ass g 1}$ and $\cofactor h {\ass g 0} \geq \cofactor f {\ass g 0}$ and (iii) $x \in \dep(g)$. We want to show $\marginal_x \HKR^\kappa_f \geq \marginal_x \HKR^\kappa_h$. Note that since \ref{prop:CGM_B} is satisfied and because $\HKR^\kappa$ is chain-rule decomposable, it suffices to show $\marginal_g \HKR^\kappa_f(S) \geq \marginal_g \HKR^\kappa_h(S)$ for all $S \subseteq X$. We have seen that
          \begin{align*}
            \marginal_g \HKR^\kappa_f(S) =
            \E_{\af \in \{0,1\}^{S}}[ \E[ \cofactor {(f_{\ass g 1} - \cofactor f {\ass g 0})} {\af}]^2].
          \end{align*}
          Note that due to $\cofactor f {\ass g 1} \geq \cofactor h {\ass g 1}$ and $\cofactor h {\ass g 0} \geq \cofactor f {\ass g 0}$ and $\cofactor h {\ass g 1} \geq \cofactor h {\ass g 0}$,
          \begin{align*}
            \cofactor {(\cofactor f {\ass g 1} - \cofactor f {\ass g 0})} {\af} \geq \cofactor {(\cofactor h {\ass g 1} - \cofactor h {\ass g 0})}{\af} \geq 0
          \end{align*}
          for every assignment $\af$ over $S$. This shows
          $ %
            \E[ \cofactor {(f_{\ass g 1} - \cofactor f {\ass g 0})}{\af}]^2 \geq \E[ \cofactor{(h_{\ass g 1} - \cofactor h {\ass g 0})} {\af}]^2,
          $ %
          so $\marginal_g \HKR^\kappa_f(S) \geq \marginal_g \HKR^\kappa_h(S)$ is true as well.
  \end{enumerate}
\end{appendixproof}

\begin{example*}
  For the special case where $\kappa = \kappaquad$, note that 
  \begin{center}
    $ \nicefrac 1 2 \cdot \kappa(a)+\nicefrac 1 2 \cdot \kappa(b) - \kappa(\nicefrac 1 2 \cdot a + \nicefrac 1 2 \cdot b) = (a-b)^2. $
  \end{center}
  Using $\E[f] = \nicefrac 1 2 \cdot \E[\cofactor f {\ass x 1}] + \nicefrac 1 2 \cdot \E[\cofactor f {\ass x 0}]$, this implies 
  \begin{center}
    $(\ZeroVF \circ \HKR^\kappa)_x(f) = (\E[\cofactor f {\ass x 1}] - \E[\cofactor f {\ass x 0}])^2,$
  \end{center}
  where $\ZeroVF$ is again the expectation of contributions with
  \begin{center}
    $\ZeroVF_x(v) = v( \{ x \}) - v(\varnothing)$.
  \end{center}
  The value $\ZeroVF \circ \HKR^\kappa$ is an IVF according \Cref{thr:cgmsInduceIVFs} and \Cref{thr:propertiesOfHKR}.
  In contrast to derivative-dependent IVFs, $\ZeroVF \circ \HKR^\kappa$ assigns 
  low values to variables in parity-functions: for $f = x \oplus y$, we have $\E[f_{\ass x 1}] = \E[f_{\ass x 0}]$, and thus 
  $(\ZeroVF \circ \HKR^\kappa)_{x}(f) = 0$.
\end{example*}

\paragraph{Derivative dependence.}

This property cannot be achieved, as witnessed by $f = x \oplus y$ and $g = x$. 
Due to $\Adiff_x f = \Adiff_x g$, it suffices to show that $\partial_x \HKR^\kappa_f \neq \partial_x \HKR^\kappa_g$ holds for all $\kappa$.
Note that
\begin{center}
  $
  \E[\cofactor f {\ass x 0}] = \nicefrac 1 2,
  \ \E[\cofactor f {\ass x 1 }] = \nicefrac 1 2,
  \ \E[\cofactor g {\ass x 0}] = 0,
  \ \E[\cofactor g {\ass x 1 }] = 1,
  $
\end{center}
and $\E[f] = \E[g] = \nicefrac 1 2$. Thus, for all constancy measures $\kappa$,
\begin{align*}
  \partial_x \HKR^\kappa_f(\varnothing) 
  &= \nicefrac 1 2 \cdot \kappa(\nicefrac 1 2) + \nicefrac 1 2 \cdot \kappa(\nicefrac 1 2) - \kappa(\nicefrac 1 2) = 0, \\
  \partial_x \HKR^\kappa_g(\varnothing) 
  &= \nicefrac 1 2 \cdot \kappa(1) 
   + \nicefrac 1 2 \cdot \kappa(0) - \kappa(\nicefrac 1 2) = 1,
\end{align*}
which shows $\marginal_x \HKR^\kappa_f \neq \marginal_x \HKR^\kappa_g$.

\section{Computing Importance Values}\label{sec:computation}

In this section, we present and evaluate computation schemes for blame, influence, and CGMs.
While there exists a practical approach based on model 
counting for the influence in CNFs~\cite{traxler2009variable}, we are only aware of naïve computations of 
CHK's blame~\cite{dubslaff2022causality}.
\ifx\arxiv\undefined
\else
Details are given in the appendix. 
\fi

\paragraph{Blame.} 
We focus on the modified blame. CHK's blame can be computed in a very similar fashion. 
Observe that for a Boolean function $f$ and $x\in X$, %
  \begin{center}
  $\MBlame^\rho_x(f) = \E[\gamma_0] + \textstyle\sum_{k=1}^{n{-}1} \rho(k)(\E[\gamma_k] - \E[\gamma_{k{-}1}]),$
  \end{center}
where $\gamma_k$ is the Boolean function for which $\gamma_k(\uf) = 1$ iff $\mscs^\uf_x(f) \leq k$. 
We devise two approaches for computing $\E[\gamma_k]$.
The first represents $\gamma_k$ through BDDs using the following recursion scheme: $\mscs^\uf_x(f) \leq k$ holds iff
\begin{itemize}[noitemsep]
  \item $k=0$ and $f(\uf) \neq f(\switchmain{\{x\}}{\uf})$, or
  \item $k>0$ and
  \begin{itemize}[noitemsep]
    \item $\mscs^\uf_x(f) \leq k{-}1$ or
    \item there is $y \neq x$ such that $\mscs^\uf_x(\switchf{y}{f}) \leq k{-}1$.
  \end{itemize}
\end{itemize}
This allows us to construct BDDs for $\gamma_k$ from $\gamma_{k{-}1}$, which lends itself to BDD-based approaches since $\gamma_k$ does not necessarily increase in size as $k$ grows.
The second approach %
introduces new existentially quantified variables in the input formula of $f$ %
to model occurrences of variables in critical sets of $\mscs^\uf_x(f)$.
With an additional cardinality constraint restricting the number of variables in critical sets to at most $k$,
we can use projected model counting to compute $\E[\gamma_k]$.

\begin{toappendix}
  \subsection{Computation of Blame Values} \label{sec:comp-blame}
  Suppose $f$ is a Boolean function and $x$ a variable. We are interested in computing the quantity $\Blame^\rho_x(f)$ for any share function $\rho$. First mind that $\scs^{\uf}_x(f)$ can only take values in $\{ 0,\dots,n{-}1,\infty\}$ since there are not more than $n{-}1$ elements a candidate solution to $\scs^\uf_x(f)$ can contain. Counting the number of times $\scs^\uf_x(f)=k$ occurs, we obtain
  \begin{align*}
    \Blame^\rho_x(f) 
    &= \frac{1}{2^n} \sum_{\uf \in \{0,1\}^X} \rho( \scs^\uf_x(f) ) \\
    &= \frac{1}{2^n} \sum_{k=0}^{n{-}1} \rho(k) | \{ \uf \in \{ 0,1\}^X : \scs^\uf_x(f) = k \} | \\
    &= \E[\gamma_{f,x,0}] + \sum_{k=1}^{n{-}1} \rho(k) ( \E[\gamma_{f,x,k}] - \E[\gamma_{f,x,k{-}1}] ),
  \end{align*}
  where $\gamma_{f,x,k}$ is a Boolean function such that $\gamma_{f,x,k}(\uf) = 1$ holds if and only if $\scs^\uf_x(f) \leq k$, so the number of its models must agree with the number of elements in the set $\{ \uf \in \{0,1\}^X : \scs^\uf_x(f) \leq k\}$.

  \subsubsection*{BDD-Based Approach}
  Let us first show how one can compute $\gamma_{f,x,k}$ for $\Blame^\rho_x(f)$ using BDD-based representations. So suppose that $f$ is represented by a BDD with an arbitrary order of variables. Using this representation, we recursively define the following Boolean function for $k \in \{0,\dots,n{-}1\}$ and $c \in \{0,1\}$:

  \begin{align*}
      \theta_{f,x,c,k} = \begin{cases}
          f \land \overline{\switchf{x}{f}}  & k= 0, c= 1 \\
          \overline{f} \land \switchf{x}{f}  & k= 0, c= 0 \\
          \theta_{f,x,c,k{-}1} \lor \bigvee_{y \in X \setminus \{ x \}} \switchf{y}{\theta_{f,x,c,k{-}1}} & k > 0
      \end{cases}
  \end{align*}
  Note that for BDDs, flipping variables and complementing and composing them once(!) using $\land,\lor$, etc. increases BDD sizes only polynomially and sometimes even decreases them. The last fact is especially important for $k > 0$, as it is often times the case that the BDD's size decreases as $k$ grows, which is due to the fact that $\theta_{f,x,c,k}$ becomes more and more satisfiable. Other approaches, such as the representation of $\theta_{f,x,c,k}$ by propositional formulas, do not have this advantage and will always ``explode''. We now show that for all assignments $\uf$, 
  \[ 
    \theta_{f,x,c,k}(\uf)=1 \tiff \scs^\uf_x(f,c) \leq k.
  \] 
  In order to see this, we need the following technical lemma:

  \begin{apxlemma}
      \label{thr:blameBDD1}
      For all Boolean functions $f$, variables $x,z$, constants $c \in \{0,1\}$ and $k \geq 0$,
      \begin{align*}
          \theta_{\switchf{z}{f},x,c,k} = \switchf{z}{\theta_{f,x,c,k}}.
      \end{align*}
  \end{apxlemma}
  \begin{proof}
      Note that for all Boolean functions $\ell$ and $h$, variables $y$ and component-wise operations $\circ$ (such as $\land,\lor$, etc.),
      \begin{align*}
          \switchf{z}{(\ell \circ h)} &= \switchf{z}{\ell} \circ \switchf{z}{h}, \\
          \overline{\switchf{z}{\ell}} &= \switchf{z}{\overline\ell}, \\
          \switchf{z}{\switchf{y}{\ell}} &= \switchf{y}{\switchf{z}{\ell}}.
      \end{align*}
      Using induction over $k$, we can see that these facts suffice to get us
      $ %
          \theta_{\switchf{z}{f},x,c,0} = \switchf{z}{\theta_{f,x,c,0}}
      $ %
      in the base case. For $k > 0$, use the induction hypothesis, which implies $\theta_{\switchf{z}{f}, x,c,k{-}1} = \switchf{z}{\theta_{f, x,c,k{-}1}}$, and simplify according to the rules given above to get $\theta_{\switchf{z}{f}, x,c,k} = \switchf{z}{\theta_{f,x,c,k}}$.
  \end{proof}

  \begin{apxtheorem}
      \label{thr:blameBDD2}
      For all Boolean functions $f$, variables $x$, $c \in \{0,1\}$, $k \geq 0$ and assignments $\uf$,
      \begin{align*}
          \theta_{f,x,c,k}(\uf) = 1 \tiff \scs^\uf_x(f,c) \leq k.
      \end{align*}
  \end{apxtheorem}
  \begin{proof}
      This can be shown by induction over $k$. If $k = 0$, then $\scs^\uf_x(f,c) = 0$ if and only if $c = f(\uf) \neq f(\switch{x}{\uf})$, which can be seen to be identical to asking for $\theta_{f,x,c,0}(\uf) = 1$. For a fixed $k > 0$, let us assume for the induction hypothesis that the claim of \Cref{thr:blameBDD2} already holds for $k{-}1$:
      \begin{align*}
        \scs^\uf_x(f,c) \leq k 
        &\tiff \scs^\uf_x(f,c) \leq k{-}1 \text{ or } \exists y \in X \setminus \{ x \}.\; \scs^\uf_x(\switchf{y}{f},c) \leq k{-}1 
          && (*) \\
        &\tiff \theta_{f,x,c,k{-}1}(\uf) =1 \text{ or } \exists y \in X \setminus \{ x \}.\; \theta_{\switchf{y}{f},x,c,k{-}1}(\uf) = 1  
          && \text{induction hypothesis} \\
        &\tiff \theta_{f,x,c,k{-}1}(\uf) =1 \text{ or } \exists y \in X \setminus \{ x \}.\; \switchf{y}{\theta_{f,x,c,k{-}1}}(\uf) = 1 
          && \text{\Cref{thr:blameBDD1}} \\
        &\tiff \theta_{f,x,c,k}(\uf)=1.
          && \text{definition}
      \end{align*}
      Here $(*)$ is true since for all $S \subseteq X \setminus \{ x \}$ and $y \in S$,
      \begin{align*}
        &c = f(\switch{S}{\uf}) \neq f(\switch{x}{\switch{S}{\uf}})  \\
        \tiff\quad &c = \switchf{y}{f}(\switch{S \setminus \{ y \}}{\uf}) \neq \switchf{y}{f}(\switch{x}{\switch{S \setminus \{ y \}}{\uf}}).
      \end{align*}
      Thus, if $S$ is a candidate solution to $\scs^\uf_x(f,c)$ of size $k$, then $S \setminus \{ y \}$ is a candidate solution to $\scs^\uf_x(\switchf{y}{f}, c)$ of size $k{-}1$ and vice versa.
  \end{proof}

  Now, let us turn back to the computation of CHK's blame. To obtain a BDD that models $\scs^\uf_x(f) \leq k$, we set 
  \[ 
    \gamma_{f,x,k} = f \theta_{f,x,1,k} \lor \overline f \theta_{f,x,0,k},
  \] 
  which introduces a case distinction on $f$. For the modified blame, we can use a similar approach and simply have to redefine $\gamma_{f,x,k}$ as  
  \begin{align*}
    \gamma_{f,x,k} = \begin{cases}
      f \oplus \switchf{x}{f} & k = 0 \\
      \gamma_{f,x,k{-}1} \lor \bigvee_{y \in X \setminus \{ x \}} \switchf{y}{\gamma_{f,x,k{-}1}}. & k > 0
    \end{cases}
  \end{align*}
  
\end{toappendix}

\begin{toappendix}
  \subsubsection*{Projected Model Counting Based Approach}

  Let $\phi$ be a propositional formula that represents $f$ and suppose $x \in \var(\phi) = \Phi$ occurs in $\phi$. The basis for this approach lays in the observation that we can devise a sequence of ``short'' propositional formulas $\{ \chi_k \}_{k=0}^{n{-}1}$ with some extra variables $N \subseteq X \setminus \Phi$ such that for all assignments $\uf$ over $\Phi$ and $k \in \{0,\dots,n{-}1\}$,
  \begin{align*}
    \exists \vf \in \{0,1\}^N . \chi_k(\uf; \vf)
    \tiff 
    \scs^\uf_x(f) \leq k.
  \end{align*}
  Counting the number of $\uf$'s that satisfy the left hand side is done through projected model counting. The general idea is that these new variables contained in $N$ certify the existence of a candidate solution $S$ of size $k$ or less for the right side. The definition of $\chi_k$ contains roughly two parts: one that states that $S$ cannot have more than $k$ elements, and one that makes sure that $S$ is indeed a candidate solution, i.e. satisfies $\phi(\uf) = \phi(\switch{S}{\uf}) \neq \phi(\switch{x}{\switch{S}{\uf}})$. The former part is ensured by using a cardinality constraint and the latter part by introducing a set of new variables that model the action of adding elements to $S$:
  
  First, introduce $|\Phi|-1$ new variables by setting $V = \{ v_y : y \in \Phi \setminus \{ x \} \} \subseteq X \setminus \Phi$. These variables are used in the replacement 
  \[ 
    \psi = \phi[y/(y \oplus v_y) : v_y \in V], 
  \] 
  which is easy to carry out on formulas and doubles $\phi$'s size in the worst case. Furthermore, we use the \emph{totalizer encoding} \cite{bailleuxEfficientCNFEncoding2003} to devise a formula $\atmost_{k}$ that introduces $O(n \log n)$ new variables denoted by the set $W$ such that $\var(\atmost_k) = W \cup V$ and 
  \begin{align*}
    \exists \wf \in \{0,1\}^W.\; \atmost_{k}(\vf; \wf) 
    \tiff 
    | \vf | \leq k.
  \end{align*}
  for all assignments $\vf$ over $V$. The length of $\atmost_k$ is in the range $O(n^2)$. Finally, we set:
  \begin{align*}
    \chi_k(\uf;\vf;\wf) = \atmost_k(\vf;\wf) \land (\phi(\uf) \leftrightarrow \psi(\uf;\vf)) \land (\phi(\uf) \oplus \switchf{x}{\psi}(\uf;\vf)),
  \end{align*}
  where $\leftrightarrow$ denotes the syntactic equivalence.
  Note that the latter part increases $\phi$'s size only linearly. Thus, if $m$ represents the size of $\phi$, then this transformation creates a formula of size $O(n^2 + m)$ and $O(n \log n)$ variables. The following theorem now shows our original claim:  
  \begin{apxtheorem}
    For all $k \in \{0,\dots,n{-}1\}$ and assignments $\uf$ over $\Phi$, we have 
    \begin{align*}
      \exists \vf \in \{0,1\}^V, \wf \in \{0,1\}^W . \chi_k(\uf; \wf; \vf) 
      \tiff 
      \scs^\uf_x(f) \leq k
    \end{align*}
  \end{apxtheorem}
  \begin{proof}
    Define $\sigma : V \rightarrow \Phi \setminus \{ x \}, v_y \mapsto y$. Then note that
    \begin{align}
      &\exists \vf \in \{0,1\}^V, \wf \in \{0,1\}^W.\; \chi_k(\uf;\wf;\vf) \nonumber \\
      &\tiff \exists \vf \in \{0,1\}^V.\; (|\vf|\leq k) \land (\phi(\uf) = \psi(\uf;\vf) \neq \psi((\switch{x}\uf);\vf)) 
        && \label{eq:blamePMC1} \\
      &\tiff \exists \vf \in \{0,1\}^V.\; (|\sigma \vf| \leq k) \land (\phi(\uf) = \phi(\uf\oplus (\sigma\vf; \ass x 0)) \neq \phi((\switch{x}{\uf}) \oplus (\sigma\vf; \ass x 0))) 
        && \label{eq:blamePMC2} \\
      &\tiff \exists \vf \in \{0,1\}^{\Phi\setminus \{ x \}}.\; (|\vf| \leq k) \land (\phi(\uf) = \phi(\uf\oplus (\vf;\ass x 0)) \neq \phi((\switch{x}{\uf}) \oplus (\vf;\ass x 0))) 
        && \label{eq:blamePMC3} \\
      &\tiff \exists S \subseteq \Phi\setminus \{ x \}.\; (|S| \leq k) \land (\phi(\uf) = \phi(\switch{S}{\uf}) \neq \phi(\switch{x}{\switch{S}{\uf}})) 
        && \label{eq:blamePMC4} \\
      &\tiff \scs^\uf_x(f) \leq k.
        && \label{eq:blamePMC5}  
    \end{align}
    Equivalence \eqref{eq:blamePMC1} moves the existential quantification of $\wf$ into $\chi_k$, replaces $\exists \wf \in \{0,1\}^W. \atmost_k(\vf;\wf)$ by the constraint $|\vf|\leq k$, and swaps $\switchf{x}{\psi}(\uf;\vf)$ by $\psi((\switch{x}{\uf}); \vf)$. The next equivalence \eqref{eq:blamePMC2} uses the fact that for all $\vf$ over $V$,
    \begin{align*}
      \psi(\uf; \vf) = \phi(\uf \oplus (\sigma\vf; \ass x 0)).
    \end{align*}
    We show this fact as follows. Note that for a single element $v_y \in V$ and $c,e \in \{0,1\}$,
    \begin{align*}
      \cofactor {\phi[y/(y \oplus v_y)]} {\ass y c; v_\ass y e} 
      &= \cofactor{\left( (y \oplus v_y) \cofactor \phi {\ass y 1} \lor (\overline{y \oplus v_y}) \cofactor \phi {\ass y 0} \right)} {\ass y c; v_\ass y e} \\
      &= \cofactor {\left( y \cofactor \phi {\ass y 1} \lor \overline{y} \cofactor \phi {\ass y 0} \right)} {\ass y {(c \oplus e)}} \\
      &= \cofactor \phi {(\ass y c) \oplus (\ass {\sigma(v_y)} e)}.
    \end{align*}
    Repeatedly applying this argument for all $v_y \in V$ gives us for all $\wf$ over $\Phi \setminus \{ x \}$,
    \begin{align*}
      \cofactor \psi {\wf;\vf} = \cofactor \phi {\wf \oplus \sigma \vf}
    \end{align*}
    and therefore, for $\uf = (\wf; \ass x c)$ over $\Phi$ with $W = \Phi \setminus \{ x \}$,
    \begin{align*}
      \psi(\uf; \vf) 
      = \cofactor \psi {\wf; \vf}(\ass x c) 
      = \cofactor \phi {\wf \oplus \sigma \vf}(\ass x c) 
      = \phi ((\wf \oplus \sigma \vf); \ass x c)
      = \phi (\uf \oplus (\sigma \vf; \ass x 0)).
    \end{align*}
    For \eqref{eq:blamePMC3}, we use the fact that $\sigma : \vf \mapsto \sigma \vf$ is a bijection from the set of assignments over $V$ to the set of assignments over $\Phi \setminus \{ x\}$. For \eqref{eq:blamePMC4}, we use again a bijection, but this time taking subsets of $\Phi \setminus \{ x \}$ to assignments over $\Phi \setminus \{ x \}$ through $S \mapsto \idf{S}$. Finally, for \eqref{eq:blamePMC5}, we simply apply the definition of $\scs$.

  \end{proof}
\end{toappendix}

\paragraph{Influence.} %
In case $f$ is given as a CNF formula, we use Traxler's method to compute the influence~\cite{traxler2009variable}.
For all other formulas, note that
standard satisfiability-preserving transformations do not preserve influence values:
For example, applying the Tseytin transformation to
$x {\lor} \overline x y$ results in a CNF where $x$ has a higher influence than $y$.

However, the influence is proportional to the number of models of $\Adiff_x f$. 
If $f$ is given by a BDD, computing a representation of $\Adiff_x f$ means squaring $f$'s size in the worst case, 
while the formula-based representation only doubles it. 
For the latter case, we can count the models of $\Adiff_x f $ using a Tseytin transformation 
and a standard model counter. %

\paragraph{BDD representations of satisfiability-biased CGMs.} 
The dominating CGM computes a simple game, which is essentially a 
Boolean function, and therefore permits a representation by BDDs. Moreover, 
using a BDD representation of $f$, we compute $\omega_f$ 
using a recursion on cofactors of variables $z$,
  \begin{center}
  $\cofactor {(\omega_f)}{\ass z 1}  = \omega_{\cofactor f {\ass z 1}}\!\!\lor \omega_{\cofactor f {\ass z 0}}$ \quad and \quad
  $\cofactor {(\omega_f)}{\ass z 0}  = \omega_{\cofactor f {\ass z 0} \land \cofactor f {\ass z 1}}$.
  \end{center}
The Banzhaf value of $x$ in $\omega_f$ is then just 
\begin{center} $\E[\cofactor {(\omega_{f})} {\ass x 1}] - \E[\cofactor {(\omega_{f})} {\ass x 0}],$\end{center} 
  which poses no effort once the BDD of $\omega_f$ is constructed. The rectifying CGM can be computed analogously. 

\begin{figure}
  \hspace*{-1em}\resizebox{0.54\textwidth}{!}{\input{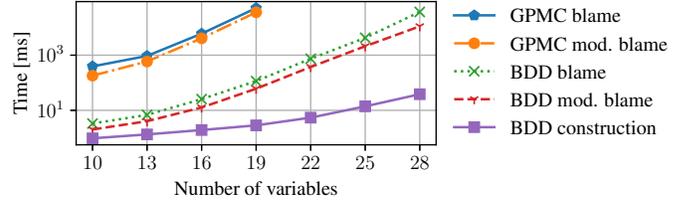}}
  \vspace{-2em}
  \caption{\label{fig:blame}Computation of blame values on random $(n,3n,7)$-CNFs (number of variables, 
  number of clauses, clause width). BDD times \emph{include} construction time of the BDD for the initial CNF.}
\vspace{-0.5em}
\end{figure}

\paragraph{Implementation and evaluation.}
We have implemented Traxler's method
and our new computation schemes
in Python, using BuDDy~\cite{lind-nielsenBuDDyBinaryDecision1999} %
as BDD backend with automatic reordering and %
GPMC~\cite{suzukiExtensionDPLLBasedModelCounting2015,suzukiImprovementProjectedModelCounting2017} for (projected) model counting.
To evaluate our approaches, we conducted experiments on Boolean functions given as CNFs that
were either randomly generated or generated from the ISCAS'99 dataset~\cite{davidson1999itc,compileProject}.
We always computed importance values w.r.t. the first variable in the input CNF
and averaged the timings over 20 runs each.
Our experiments were carried out on a Linux system with an i5-10400F CPU at 2.90GHz and 16GB of RAM.
To compare our BDD-based and model counting approaches, \Cref{fig:blame} shows timings for 
blame computations on random CNFs. Here, the BDD-based approach clearly outperforms
the one based on projected model counting. This is also reflected in real-world benchmarks from ISCAS'99 shown 
in~\Cref{tab:iscas99-main}, where the approach based on model counting runs into timeouts
for even small instances. 
\Cref{tab:iscas99-main} shows that computations for influence values based on model counting
scale better than the BDD-based approach, mainly due to an expensive initial BDD construction. 
Computing the BDD of the dominating CGM is done without much overhead once the BDD for the CNF is given.
\vspace{-0.4em}

\begin{toappendix}
  \subsection{Cooperative Game Mappings}
  For the scope of this section, we interpret the simple cooperative games produced by $\omega$ and $\nu$ as Boolean functions, e.g. write $\omega_f(\idf{S}) = \omega_f(S)$ or equivalently $\omega_f(\uf) = \omega_f(\{ x \in X : \uf(x) = 1 \})$ for all subsets $S \subseteq X$ resp. assignments $\uf$. Recall from \Cref{thr:dcgmDecompositionVariables} that for all variables $z$,
  \begin{align*}
    &\omega_f(S \cup \{ z \}) = (\omega_{\cofactor f {\ass z 1}} \lor \omega_{\cofactor f {\ass z 0}})(S), \\
    &\omega_f(S \setminus \{ z \}) = \omega_{\cofactor f {\ass z 1} \land \cofactor f {\ass z 0}}(S),
  \end{align*}
  which means for $\omega_f$'s cofactors w.r.t. $z$,
  \begin{align*}
    &\cofactor {(\omega_f)} {\ass z 1} = \omega_{\cofactor f {\ass z 1}} \lor \omega_{\cofactor f {\ass z 0}}, \\
    &\cofactor {(\omega_f)} {\ass z 0} = \omega_{\cofactor f {\ass z 1} \land \cofactor f {\ass z 0}}.
  \end{align*}
  Using the Shannon decomposition, note that 
  \begin{align*}
    \omega_f = z \cofactor {(\omega_{f})} {\ass z 1} \lor \overline z \cofactor {(\omega_{f})} {\ass z 0} = z (\omega_{\cofactor f {\ass z 1}} \lor \omega_{\cofactor f {\ass z 0}}) \lor \overline z \omega_{\cofactor f {\ass z 0} \land \cofactor f {\ass z 1}},
  \end{align*}
  which corresponds to the decomposition as presented in the main part. The Banzhaf value for $\omega_f$ is simply $\Banzhaf_x(\omega_f) = \E[\cofactor {(\omega_f)} {\ass x 1}] - \E[\cofactor {(\omega_f)} {\ass x 0}]$, which can easily be computed if $\omega_f$ is represented by a BDD. For the rectifying CGM, note that 
  \begin{align*}
    &\nu_f(S \cup \{ z \}) = \nu_{\cofactor f {\ass z 1} \lor \cofactor f {\ass z 0}}(S), \\
    &\nu_f(S \setminus \{ z \}) = (\nu_{\cofactor f {\ass z 1}} \nu_{\cofactor f {\ass z 0}})(S),
  \end{align*}
  which produces the decomposition 
  \begin{align*}
    \nu_f = z \nu_{\cofactor f {\ass z 1} \lor \cofactor f {\ass z 0}} \lor \overline z \nu_{\cofactor f {\ass z 1}} \nu_{\cofactor f {\ass z 0}}.
  \end{align*}
  The remainder is analogous.

  \subsection{Experiments and Results}
  
  \begin{figure}[ht]
    \centering
    \begin{subfigure}{0.9\textwidth}
      \resizebox{\textwidth}{!}{\input{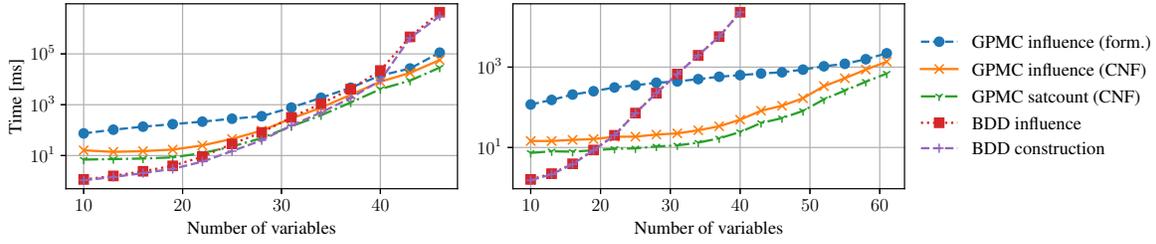}}
      \caption{Performance trends for influence. \emph{GPMC satcount (CNF)} is the time spent counting the number of satisfying assignments of the initial CNF and \emph{GPMC influence (CNF)} is the time of the method from \cite{traxler2009variable}.}
      \label{fig:performance-trends-influence}
    \end{subfigure}%
    \hfill
    \begin{subfigure}{0.9\textwidth}
      \resizebox{\textwidth}{!}{\input{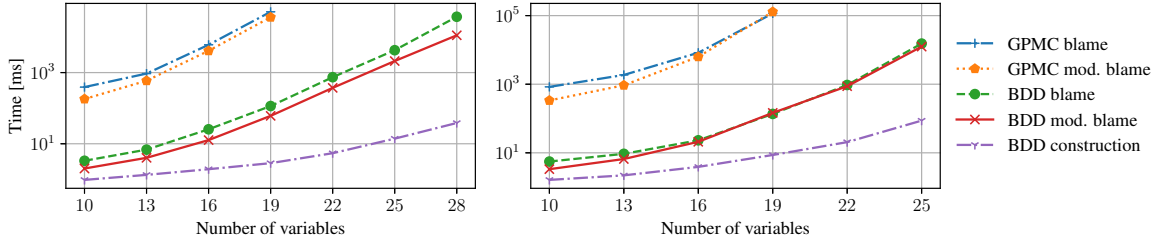}}
      \caption{Performance trends for blame.}
      \label{fig:performance-trends-blame}
    \end{subfigure}%
    \hfill
    \begin{subfigure}{0.9\textwidth}
      \resizebox{\textwidth}{!}{\input{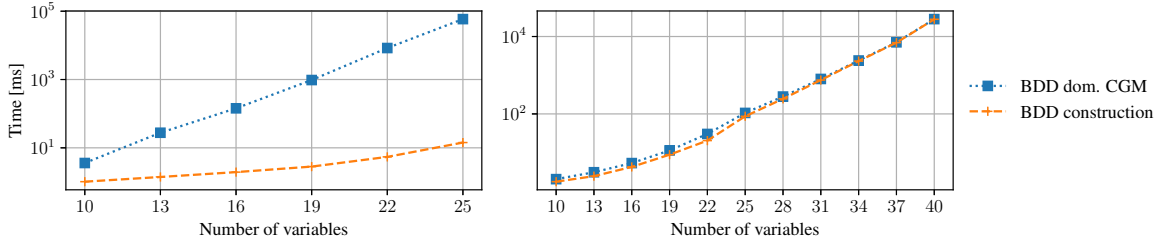}}
      \caption{Performance trends for the dominating CGM.}
      \label{fig:performance-trends-dcgm}
    \end{subfigure}%
    \caption{Performance results for random $(n,3n,7)$-CNFs (left column) and random $(n,7n,4)$-CNFs (right column). \emph{BDD construction} is always the time used to compute the BDD representing the initial CNF and is \emph{included} in all BDD-based methods.}
    \label{fig:performance-trends}
  \end{figure}

  \begin{figure}
    \centering
    \resizebox{0.8\textwidth}{!}{\input{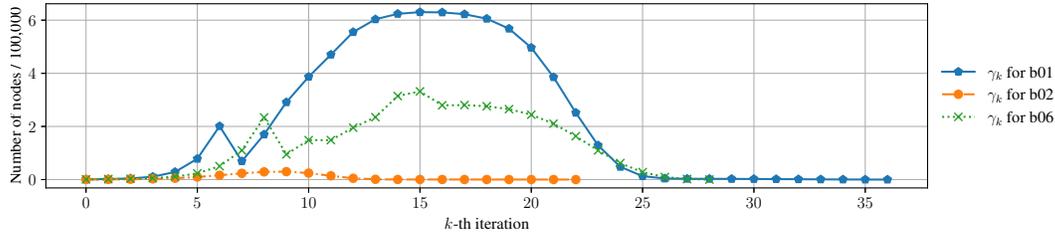}}
    \caption{Number of nodes in the BDD of $\gamma_k$ used in the computation of the blame for growing $k$. (cf. \Cref{sec:computation}.)}
    \label{fig:bdd-sizes}
  \end{figure}

  \begin{table}[ht]
    \centering
    \resizebox{0.9\textwidth}{!}{%
    \begin{tabular}{crr rrr rrrr}
      \toprule
      & & & \multicolumn{3}{c}{(projected) model counting approaches} & \multicolumn{3}{c}{BDD-based approaches} \\ \cmidrule(r){4-6} \cmidrule(r){7-10}
        Instance & \#Variables & \#Clauses & Influence (CNF) & Influence (formula) & Blame & Construction & Influence & DCGM & Blame \\
        \midrule 
b02 & 26 & 66 & 5\,ms & 49\,ms & \timeout & 1\,ms & $<$1\,ms & 2\,ms & 3'649\,ms \\
b06 & 44 & 122 & 7\,ms & 99\,ms & \timeout & 3\,ms & $<$1\,ms & 6\,ms & 697'573\,ms \\
b01 & 45 & 120 & 7\,ms & 110\,ms & \timeout & 4\,ms & $<$1\,ms & 8\,ms & 3'068'667\,ms \\
b03 & 156 & 376 & 11\,ms & 442\,ms & \timeout & 53'934\,ms & 24\,ms & 1'776\,ms & \timeout \\
b09 & 169 & 417 & 22\,ms & 483\,ms & \timeout & 23'373\,ms & 13\,ms & 1'137\,ms & \timeout \\
b08 & 180 & 455 & 26\,ms & 577\,ms & \timeout & 11'703\,ms & 18\,ms & 1'729\,ms & \timeout \\
b10 & 201 & 525 & 27\,ms & 658\,ms & \timeout & 19'079\,ms & 44\,ms & 3'221\,ms & \timeout \\
b13 & 352 & 847 & 34\,ms & 1'088\,ms & \timeout & \timeout & \timeout & \timeout & \timeout \\
b07 & 435 & 1'132 & 64\,ms & 2'058\,ms & \timeout & \timeout & \timeout & \timeout & \timeout \\
b04 & 730 & 1'919 & 324\,ms & 4'147\,ms & \timeout & \timeout & \timeout & \timeout & \timeout \\
b11 & 765 & 2'104 & 606\,ms & 5'872\,ms & \timeout & \timeout & \timeout & \timeout & \timeout \\
b05 & 966 & 2'798 & 173\,ms & 9'150\,ms & \timeout & \timeout & \timeout & \timeout & \timeout \\
b12 & 1'072 & 2'911 & 230\,ms & 8'555\,ms & \timeout & \timeout & \timeout & \timeout & \timeout \\
        \bottomrule
      \end{tabular}}
      \caption{\label{tab:iscas99}Computation time for instances of the ISCAS'99 dataset, timeout set to one hour. 
      BDD columns \emph{Influence}, \emph{DCGM} (construction of the BDD for the dominating CGM), and \emph{Blame} are \emph{without} 
      the BDD construction time for the initial CNF (cf. column \emph{Construction}).}
    \label{tab:iscas99}
  \end{table}

  The proposed algorithms were tested on two datasets of randomly generated CNFs and the ISCAS'99 benchmark. The former two are specified by tuples $(n,m,k)$ where $n$ is the number of variables, $m$ the number of clauses and $k$ the width of every clause.
  The first dataset contains random $(n,3n,7)$-CNFs which have expected values from at least $0.3$ up to $0.9$. These CNFs are very balanced; the probability that a random assignment satisfies a random CNF from this dataset is $0.54$. The second dataset consists of random $(n,7n,4)$-CNFs which have expected values from $7 \cdot 10^{-14}$ up to $0.0195$. Only small amounts of assignments satisfy CNFs of this dataset. For each $n$, we have sampled $10$ CNFs and computed the mean running time of every algorithm. This includes parsing and constructing the BDDs with fixed but essentially random variable orders. The (modified) blame is computed for $\rho = \rhofrac$. Blame and influence values are computed with respect to the first variable.
 
  Let us compare the results of both datasets as shown in \Cref{fig:performance-trends}:
  \begin{itemize}
    \item {\bf \Cref{fig:performance-trends-influence}}: For formula based approaches, computing the influence includes parsing the CNF as a formula $\phi$, applying the transformation $\phi \mapsto \cofactor \phi {\ass x 0} \oplus \cofactor \phi {\ass x 1}$ and then re-translating the result back into a CNF via a Tseytin transformation. This increases the original size only linearly, which is arguably the reason for why GPMC's performance is close (or becomes close over time) to the approach of \cite{traxler2009variable} for both datasets. Our BDD-based implementation is faster for smaller numbers of variables, but is overtaken by GPMC eventually. As one can see on the right side, GPMC is much faster for random CNFs with low expected values. Note the influence is computed with little to no overhead once the BDD construction is done.
    \item {\bf \Cref{fig:performance-trends-blame}}: Recall that we have presented two possible ways of computing the (modified) blame. In our case, the projected model counting approach via GPMC performs worse than the BDD-based approach. Note that the modified blame is computed slightly faster than the blame in both cases. However, the computation is still limited to small amounts of variables. %
    \item {\bf \Cref{fig:performance-trends-dcgm}}: Our proposed approach for the representation of the dominating cooperative game mapping through BDDs performs best in the second dataset. This is because most of these CNFs have a very low share of satisfied assignments, so one expects the BDDs to have small sizes, which in turn allows one to compute the dominating cooperative game mapping efficiently. (Only few subsets $S \subseteq X$ satisfy $\omega_f(S) = 1$.) Thus, the bottle neck for sparsely satisfied CNFs seems to be the BDD's construction time.
  \end{itemize}
  
  We conclude that (i) the influence can be computed efficiently using modern model counting software, (ii) the blame can be computed efficiently for smaller amounts of variables using a BDD-based approach and (iii) given a BDD representation of a Boolean function, the dominating cooperative game mapping can usually be computed efficiently. For both influence and dominating cooperative game mapping, the BDD construction time is the biggest bottle neck.

  Finally, let us consider the ISCAS'99 benchmark \cite{compileProject}. \Cref{tab:iscas99} shows how well our proposed approaches perform on average for 20 runs. The smallest instances of this benchmark still permit the computation of blame values, which was done for the instances b02, b06 and b01. When plotting the number of nodes in the BDDs representing the functions $\{ \gamma_{k}\}_{k=0}^n$ (recall that for a fixed Boolean function $f$ and variable $x$, $\gamma_k(\uf) = 1$ if and only if $\scs^\uf_x(f) \leq k$), we obtain \Cref{fig:bdd-sizes}. Observe that as $k$ grows, $\gamma_{k}$ becomes more and more satisfiable. Consequently, we can see that the sizes of BDDs representing $\gamma_k$ first increase but later on decrease until they hit zero nodes (i.e. are one everywhere).
\end{toappendix}

\section{Conclusion}

This paper introduced IVFs as a way to formally reason about importance of variables in Boolean functions.  
We established general statements about IVFs, also providing insights on notions of importance 
from the literature by showing that they all belong to the class of IVFs.
Apart from revealing several relations between known IVFs, we have shown how 
to generate new ones inspired by cooperative game theory.

For future work, we  %
will study properties with strict importance inequalities,
IVFs for sets of variables, IVFs for pseudo Boolean functions,
and global values %
similar to the \emph{total influence}~\cite{o2014analysis}.
On the empirical side, %
the generation of splitting rules for SAT-solvers and
variable-order heuristics for BDDs based on different 
instances of IVFs are promising avenues to pursue.

~\\\textbf{Acknowledgments.}
The authors were partly supported by the DFG through
the DFG grant 389792660 as part of \href{https://perspicuous-computing.science}{TRR~248} and
the Cluster of Excellence EXC 2050/1 (CeTI, project ID 390696704, as part of Germany's Excellence Strategy)
and ``SAIL: SustAInable Life-cycle of Intelligent Socio-Technical Systems'' (Grant ID NW21-059D), 
funded by the program ``Netzwerke 2021'' of the Ministry of Culture and Science of the State of North Rhine-Westphalia, Germany.

\bibliographystyle{named}
\bibliography{../bibliography}

\end{document}